\documentclass[%
 reprint,
 superscriptaddress,
 amsmath,amssymb,
 aps,
 floatfix,
]{revtex4-2}

\usepackage{graphicx}
\usepackage{dcolumn}
\usepackage{bm}
\usepackage{hyperref}
\usepackage[dvipsnames, table]{xcolor}
\usepackage{mathtools}
\usepackage{amsthm}
\newtheorem{proposition}{Proposition}[section]
\theoremstyle{definition}
\newtheorem{definition}{Definition}[section]

\usepackage[caption=false]{subfig}

\begin{document}
    
    \title{Linear cross-entropy certification of quantum computational advantage \\ in Gaussian Boson Sampling}
    
    \author{Javier Martínez-Cifuentes}
	\affiliation{%
		Department of Engineering Physics, \'Ecole Polytechnique de Montréal, Montréal, QC, H3T 1J4, Canada
	}
    \author{Hubert de Guise}
	\affiliation{%
		Department of Physics, Lakehead University, Thunder Bay, ON, P7B 5E1, Canada
	}%
    \author{Nicol\'as Quesada}
	\affiliation{%
		Department of Engineering Physics, \'Ecole Polytechnique de Montréal, Montréal, QC, H3T 1J4, Canada
	}%
 
    \date{\today}

    \begin{abstract}
        Validation of quantum advantage claims in the context of Gaussian Boson Sampling (GBS) currently relies on providing evidence that the experimental samples genuinely follow their corresponding ground truth, i.e., the theoretical model of the experiment that includes all the possible losses that the experimenters can account for. This approach to verification has an important drawback: it is necessary to assume that the ground truth distributions are computationally hard to sample, that is, that they are sufficiently close to the distribution of the ideal, lossless experiment, for which there is evidence that sampling, either exactly or approximately, is a computationally hard task. This assumption, which cannot be easily confirmed, opens the door to classical algorithms that exploit the noise in the ground truth to efficiently simulate the experiments, thus undermining any quantum advantage claim. In this work, we argue that one can avoid this issue by validating GBS implementations using their corresponding ideal distributions directly. We explain how to use a modified version of the linear cross-entropy, a quantity that we call the LXE score, to find reference values that help us assess how close a given GBS implementation is to its corresponding ideal model.  Finally, we analytically compute the score that would be obtained by a lossless GBS implementation.
    \end{abstract}

    \maketitle
    
    \section{\label{sec:introduction} Introduction} 
    Gaussian Boson Sampling (GBS)~\cite{hamilton2017gaussian,kruse2019detailed} is a model of quantum computation with the possibility to demonstrate quantum computational advantage~\cite{hangleiter2023computational}, i.e., to show that quantum devices are capable of significantly outperforming classical computers at specified computational tasks. Broadly speaking, an ideal GBS experiment consists of sending single-mode squeezed states into a Haar-random lossless interferometer, i.e., a random network of beamsplitters and waveplates. The photon number distribution of the output light is then measured using photon number resolving or threshold detectors~\cite{quesada2018gaussian}. It has been shown that, under some reasonable conjectures~\cite{grier2022complexity,kruse2019detailed, hamilton2017gaussian}, it is computationally hard to sample exactly or approximately from the probability distribution of the resulting experimental outcomes (i.e., the distribution of the strings of non-negative integers resulting from the measurement). Here, approximate sampling is understood as sampling from a probability distribution that is close in total variation distance to the ideal GBS distribution~\cite{aaronson2011computational}. The hardness of the task of sampling is what gives this model of computation the possibility of demonstrating quantum advantage.

    Real-world implementations of GBS, such as those reported in Refs.~\cite{zhong2020quantum, zhong2021phase, madsen2022quantum, deng2023gaussian}, suffer from losses and imperfections in preparation, transmission and detection of light. It is therefore expected that the theoretical probability distributions associated to these experiments, commonly referred to as \textit{ground truth distributions}, will differ from the \textit{ideal distribution} (which does not include any type losses in its definition). However, it is \textit{assumed} that the ground truth distributions are close in total variation distance to the ideal distribution, sufficiently close as to be able to invoke the hardness of approximate sampling in order to claim that it is computationally hard to classically reproduce the outcomes of the experiment. This assumption allows the experimenters to claim that their implementations achieve a quantum advantage.

    Presuming that the ground truth distributions are hard to sample, a given GBS implementation must be followed by a validation of its outcomes~\cite{hangleiter2023computational}, i.e., a verification that the experiment is operating correctly. This validation amounts to providing evidence that the experimental samples closely follow the ground truth distribution. Since the direct estimation of the total variation distance between the real distribution of the samples and the ground truth requires exponentially many runs of the experiment~\cite{hangleiter2019sample, hangleiter2023computational}, the validation of GBS implementations usually relies on other \textit{sample-efficient} statistical measures (i.e., measures requiring polynomially many samples~\cite{hangleiter2023computational}) that assess how correlated are the experimental outcomes to the ground truth distribution. 
    
    Among the current, most widely used techniques we can find Bayesian hypothesis testing~\cite{bentivegna2015bayesian}, which consists in demonstrating how likely is the ground truth to explain the experimental samples relative to an adversarial classical model. Another strategy is to use the heavy output generation test and other cross-entropy measures~\cite{zhong2020quantum, hangleiter2023computational}, which intend to verify if classically generated samples are able to produce ``heavier outputs'' in the ground truth distribution, i.e., events with higher ground truth probability, than the experimental samples.
    A further technique is to compare the correlations between modes present in the photon number resolving or threshold detection samples with those predicted by the ground truth model, and by adversarial models or samplers~\cite{zhong2021phase, phillips2019benchmarking,cardin2022photon}. A fourth validation method is based on binning the detectors in groups of certain sizes, and then analytically computing the corresponding grouped probability distribution~\cite{drummond2022simulating, bressanini2023gaussian}. This procedure is done for the ground truth and for adversarial models. These analytical distributions are then compared with those obtained using the experimental samples.

    All of the previously mentioned techniques have their own limitations. For instance, some rely on the computation of probabilities of individual samples, a task whose cost grows exponentially with the number of photons or clicks detected in different output modes of the interferometer, thus making the techniques \textit{computationally inefficient} validation strategies. Moreover, one can find situations in which these techniques reach contradicting conclusions about the validity of a set of experimental samples~\cite{martinez2023classical}. In addition, when verifying the experimental outcomes against adversarial samplers that do not have a well-defined probability distribution associated to them, it is not even possible to perform some of these tests (e.g., Bayesian testing)~\cite{villalonga2021efficient}. These issues suggest that the current widely used GBS validation methods are not sufficient to readily tell if the experimental samples unambiguously follow from the ground truth distribution and, consequently, they cannot readily validate a quantum advantage claim. This also opens the door to classical algorithms that generate samples that outperform the experimental outcomes at these validation tasks~\cite{martinez2023classical, villalonga2021efficient}, or classical strategies that ``spoof'' some of these validation techniques~\cite{oh2023spoofing}.

    Perhaps more importantly, the degree of confidence in these methods for verifying quantum advantage claims is strongly tied to the confidence in the hardness of the task of sampling from the ground truth distribution. However, at this time, there is not a reliable statistical measure that allows us to determine whether any ground truth distribution is sufficiently close to the ideal model of its corresponding experiment. This fact may allow classical algorithms to exploit the losses in the experiments in order to generate samples that, by using the same techniques for verifying GBS implementations, are found to follow the corresponding ground truth distributions more closely~\cite{oh2024classical}.

    The current state of affairs of GBS validation differs significantly from that of the field of Random Circuit Sampling (RCS) verification. It has been identified that the linear cross-entropy benchmark (XEB)~\cite{hangleiter2023computational, arute2019quantum, boixo2018characterizing}, a quantity constructed from the linear cross-entropy between the outcomes of a random circuit and its corresponding ideal distribution, serves as a witness of quantum advantage in RCS implementations. The reason behind the success of this metric, despite having its own limitations~\cite{gao2024limitations}, is that random circuits are extremely sensitive to noise; the presence of errors in experimental implementations leads to samples whose real distributions are uncorrelated with the ideal distribution~\cite{boixo2018characterizing, hangleiter2023computational,dalzell2024random}, and, moreover, exponentially close to the uniform distribution in the sample space. The XEB for samples following the uniform distribution identically vanish. When the experimental samples follow directly from the ideal distribution, the XEB is equal to 1. Furthermore, obtaining an XEB different enough from zero is considered to be a computationally hard task~\cite{hangleiter2023computational}. Quantum advantage is thus verified when this benchmark is sufficiently different from zero.

    GBS experiments are not as sensitive to loss as RCS implementations are to gate errors, and this hinders the definition of a validation metric as decisive as the XEB for the verification of GBS quantum advantage claims. Nevertheless, there are two important features of RCS validation that could be used to overcome several of the difficulties surrounding the verification of GBS experiments. The first of them is validating the experimental outcomes directly against their corresponding  ideal model. This relieves the verification process from the assumption that it is computationally hard to sample from the ground truth distribution. The second one is the determination of \textit{reference values} (of a given statistical measure) associated to the ideal distribution, as well as to other probability distributions that may differ from that of an ideal implementation. These reference values can be used to assess how far the experimental samples are from the ideal distribution. 
    
    In this work, we take the XEB as a blueprint for defining a figure of merit for GBS validation having the first of the two previously mentioned features. We construct this quantity, which we call the \textit{linear cross-entropy score} (LXE score), from a normalized version of the linear cross-entropy between the ideal GBS distribution and a test probability distribution (corresponding to the actual distribution of the experimental outcomes, or the distributions of adversarial models or samplers), averaged over the Haar measure of the unitary group, and evaluated at the limit of a large number of modes. We include the average over Haar-random unitaries in the definition of the score in order to study the behavior of the normalized linear cross-entropy for a typical implementation of GBS. On the other hand, we focus on setups with a large number of modes because this is the regime in which the hardness of the task of sampling is more manifest. 

    In addition, we propose a validation strategy that exploits the second of the aforementioned features. The technique consists of computing the LXE score for the most commonly used classical adversarial samplers or models, as well as for the ideal GBS distribution, and use them as reference values. We then would compare these reference values to the estimated score of the experimental outcomes. In this way, we can assess how close is a given GBS implementation to its corresponding ideal model, and to its most challenging adversaries.
    
    After discussing the details of the definition of the LXE score, we set off to determine the first two reference values to be used in the validation strategy. To do this, we focus on implementations using photon number resolving detectors. The first of these reference values corresponds to a model that leads to a uniform probability distribution for each sector of the total number of detected photons, $N$, in the experimental samples. This value follows directly from the normalization of the linear cross-entropy and is equal to 1. The second reference value, the \textit{ideal score}, corresponds to a GBS implementation using single-mode squeezed states as input of the first $R$ modes of a lossless interferometer, with the remaining $M-R$ modes receiving the vacuum state. We consider two versions of this GBS setup: the first uses input squeezed states with the same squeezing parameter, while the second uses squeezed states with different squeezing parameters.  

    By expressing the linear cross-entropy between two GBS distributions as an integral over several real parameters (a technique that makes its computation independent from the use of measurement outputs), and by using the Weingarten Calculus~\cite{collins2022weingarten} in order to compute the average over Haar-random unitaries, we were able to find an analytical expression for the ideal score as a function of $N$ and $R$. To the extent of our knowledge, this approach to the analysis of cross-entropy measures in the context of GBS has not been employed before.

    We find that part of the dependence of the ideal score on the parameter $R$ can be expressed as a polynomial of degree $2N$, whose coefficients can be computed by counting the number of undirected graphs with a certain number of connected components. For setups that use input states with different squeezing parameters, the coefficients also depend on the lengths of the connected components of the graph. This result is akin to the conclusions found by Ehrenberg et al.~\cite{ehrenberg2023transition, ehrenberg2024second} in their recent study on anticoncentration in GBS. In their work, they computed the first and second moments of the output GBS distribution in the photon-collision-free limit, in which nearly all detection patterns have at most one photon in each mode. In this regime, the output distribution can be approximated by the modulus squared of hafnians of Gaussian random matrices~\cite{deshpande2022quantum, arkhipov2012bosonic, grier2022complexity}. The authors developed a graph-theoretical method for computing the moments of this approximate distribution, and found that the second moment can be expressed as a polynomial of degree $2N$ in $R$, whose coefficients are determined by counting the number of graphs with a given number of connected components. Moreover, they relate the first and second moments to the LXE score of an ideal model that uses input squeezed states with the same squeezing parameter~\cite{ehrenberg2024second}.
    Even though their expression for the ideal score is very similar to our findings, the definition and origin of the graphs involved in the computation of our results differ from those used in Ref.~\cite{ehrenberg2023transition, ehrenberg2024second}. At this time, it is necessary to make a more thorough analysis of the computation of the coefficients in both cases in order to find a clear relation between these two results.

    Focusing on the case $R=M$, we find that the ideal score has a particularly simple expression in terms of $N$, and it is the same whether we use input squeezed states with the same squeezing parameter or not. We compare this analytical expression with numerical estimations of the ideal score for setups that are not in limit of $M\rightarrow\infty$. This comparison allows us to investigate how good our results approximate the estimated ideal score of real-world GBS implementations. Additionally, we also make numerical estimations of the score for a simple GBS model with transmission losses, which allows us to study the behavior of the LXE score in the presence of noise.

    It is worth mentioning that determining probability amplitudes for pure Gaussian states requires the computation of hafnians of matrices with half the size of those used in the computation of probabilities of mixed states. This implies that the estimation of the LXE score can be done for a range of detected photons approximately twice as large as that used in the validation of recent GBS experiments~\cite{zhong2020quantum, zhong2021phase, madsen2022quantum, deng2023gaussian}, thus representing a significant improvement in the validation of these implementations.

    We consider that the LXE score will be of significant importance and utility to the field of GBS verification. Moreover, our computation of the ideal score sets the stage for an alternative approach to the validation of GBS implementations, one where we assert that verifying GBS should require the evaluation of probabilities corresponding to unitary models, not to mixed states ground truths.
    
    This paper is organized as follows: In Sec.~\ref{sec:gbs_setup}
    we describe the GBS setup that we will consider throughout the article, and we will introduce the concept of \textit{model} of a GBS implementation. In Sec.~\ref{sec:lxe_gbs} we will discuss how to use the LXE score to validate GBS experiments. Sec.~\ref{sec:lxe_score_sqz} is devoted to the detailed computation of the ideal LXE score for setups that use input squeezed states with the same squeezing parameter, while Sec.~\ref{sec:different_squeezing} shows how to compute the score for models that use different squeezing parameters. In Sec.~\ref{sec:no_vacuum} we will bring attention to some interesting features of the ideal score when $R=M$. Here, we also present the numerical estimations of the score for GBS setups that are not in the limit of $M\rightarrow\infty$, and for some simple models that include transmission losses. Finally, we will conclude in Sec.~\ref{sec:discussion}.

    \section{\label{sec:gbs_setup} GBS setup}

    Consider a GBS setup (see Fig.~\ref{fig:gbs_setup}) where $M$ modes are prepared in single-mode non-displaced Gaussian states. The nature of these states need not be specified at this point; they could be squeezed, thermal, squashed, etc. The initial $M$-mode state is sent through an arbitrary linear, lossless interferometer, which is mathematically described by a $M\times M$, \textit{Haar-random} unitary matrix $\bm{U}$. The output light of the interferometer is then measured using photon number resolving detectors, which can determine the number of output photons in each mode; i.e., they measure the output state of the system in the Fock basis.
    
    The output state of the interferometer, $\hat{\rho}$, which is also and $M$-mode non-displaced Gaussian state, is completely described by a $2M\times 2M$, complex Husimi covariance matrix $\bm{\Sigma}$~\cite{kruse2019detailed, serafini2017quantum}, whose entries are computed as $\Sigma_{j,k}=\frac{1}{2}\langle\{\hat{\xi}_j,\hat{\xi}_k^\dagger\}\rangle+\frac{1}{2}\delta_{j,k}$. The $\{\hat{\xi}_k\}$ are the components of the operator vector $\bm{\hat{\xi}}=(\hat{a}_1,\dots,\hat{a}_M,\hat{a}_1^\dagger,\dots,\hat{a}_M^\dagger)$, where $\hat{a}_k^{\dagger}$ and $\hat{a}_k$ are the bosonic creation and annihilation operators of the output modes, which satisfy the canonical commutation relations $[\hat{a}_j,\hat{a}_k]=[\hat{a}_j^\dagger,\hat{a}_k^\dagger]=0$ and $[\hat{a}_j,\hat{a}_k^\dagger]=\delta_{j,k}$. $\{\hat{a}, \hat{b}\}:=\hat{a}\hat{b}+\hat{b}\hat{a}$ stands for the anticommutator of operators $\hat{a}$ and $\hat{b}$, $\langle\hat{a}\rangle=\mathrm{Tr}(\hat{a}\hat{\rho})$, and $\delta_{j,k}$ is the Kronecker delta.
    
    The result of the photon number resolving measurement, i.e. the detection pattern, is a $M$-string of non-negative integers $\bm{n}=(n_1,\dots,n_M)$, where each $n_k$ represents the number of photons detected at mode $k$ ($n_k=0$ means that no light has been detected). By defining the matrix 
    \begin{equation}
        \bm{A}=\bm{X}\left(\mathbb{I}_{2M}-\bm{\Sigma}^{-1}\right),\quad\bm{X}=\begin{pmatrix}\bm{0}&\mathbb{I}_M\\\mathbb{I}_M&\bm{0}\end{pmatrix}, 
        \label{eq:matrix_model}
    \end{equation}
    where $\mathbb{I}_M$ is the identity matrix of size $M\times M$, we can compute the probability of detecting the outcome $\bm{n}$ as~\cite{hamilton2017gaussian, kruse2019detailed} 
    \begin{equation}
        \Pr(\bm{n}|\bm{A})=\frac{\Pr(\bm{0}|\bm{A})}{\bm{n}!}\mathrm{haf}\left[\bm{A}_{\bm{n}}\right],
        \label{eq:gbs_probability_distribution}
    \end{equation}
    where $\Pr(\bm{0}|\bm{A})=\sqrt{\mathrm{det}\left(\mathbb{I}_{2M}-\bm{X}\bm{A}\right)}$ is the vacuum probability, $\bm{n}!=\prod_{k=1}^Mn_k!$, and
    \begin{equation}
        \mathrm{haf}[\bm{O}]=\frac{1}{2^m m!}\sum_{\sigma\in S_{2m}}\prod_{j=1}^mO_{\sigma(2j-1),\sigma(2j)}
        \label{eq:hafnian_definition}
    \end{equation}
    is the \textit{hafnian} of the symmetric $2m\times 2m$ matrix $\bm{O}$~\cite{barvinok2016combinatorics}. Here, $S_m$ stands for the symmetric group of degree $m$, i.e. the group of all permutations of $m$ objects.
    
    The matrix $\bm{A}_{\bm{n}}$ is constructed by taking the $k$-th and $(k+M)$-th rows and columns of $\bm{A}$ and repeating them $n_k$ times. If $n_k=0$ the corresponding rows and columns are removed. Notice that the size of this matrix is $2N\times 2N$, where $N= \sum_{k=1}^M n_k$ is the \textit{total number of detected photons} in the outcome $\bm{n}$.
    
    Eq.~\eqref{eq:gbs_probability_distribution} will also hold for descriptions of the GBS setup that include Gaussian noise and losses (transmission loss is a good example of this type of operations). In these cases, the output state of the interferometer will remain Gaussian~\cite{serafini2017quantum}, and we can readily define a matrix $\bm{\Sigma}$ or $\bm{A}$ that completely describes the output state. On the other hand, non-Gaussian noise will lead to output states that cannot be completely defined by an $\bm{A}$ matrix.
    
    Focusing on setups whose descriptions include only Gaussian processes, we can consider matrix $\bm{A}$ to contain all the relevant information about the theoretical description of a GBS experiment using photon number resolving detectors. We will therefore refer to $\bm{A}$ as the \textit{model} of the GBS setup, and we will interpret $\Pr(\bm{n}|\bm{A})$ as the probability of obtaining the outcome $\bm{n}$ given model $\bm{A}$. 

    \begin{figure}[!t]
        \centering
        \includegraphics[scale=0.45]{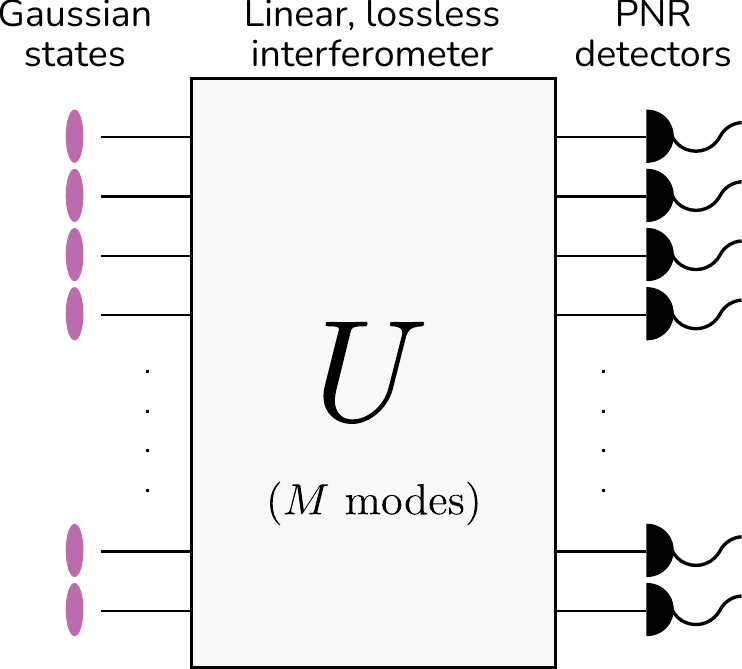}
        \caption{Description of a general, ideal GBS setup. A set of input single-mode, non-displaced Gaussian states are sent into a linear, lossless interferometer described by a $M\times M$, Haar-random, unitary matrix $\bm{U}$. The output of the interferometer is measured using photon number resolving detectors.}
        \label{fig:gbs_setup}
    \end{figure}
    
    As mentioned before, for a GBS setup that uses input single-mode squeezed states and a lossless interferometer, it has been shown that, assuming that certain conjectures hold, sampling from the distribution in Eq.~\eqref{eq:gbs_probability_distribution}, either exactly or approximately, is a computationally hard task whose cost increases exponentially with the size and rank of $\bm{A}_{\bm{n}}$~\cite{hamilton2017gaussian,kruse2019detailed,grier2022complexity}. We will refer to the theoretical description of this setup 
    as the \textit{ideal squeezed state model}, and we will denote it $\bm{A}_{\mathrm{sqz}}$.

    Any real-world implementation of GBS using squeezed states will suffer from inevitable losses, and their presence makes the experiment differ from the ideal model. Nevertheless, GBS experiments are designed to minimize noise as much as possible, so the probability distribution of the experimental samples is close enough to the ideal model and, by invoking the hardness of approximate sampling, it remains computationally hard to generate samples from this distribution. Theoretical models of GBS experiments using squeezed states that also include all the losses that the experimenters can account for are called \textit{ground truth models}.

    The verification of quantum advantage claims in the recent implementations of GBS experiments has relied on demonstrating that the \textit{real distribution} of the samples genuinely follows from the ground truth of the experiments~\cite{zhong2020quantum, zhong2021phase, madsen2022quantum, deng2023gaussian}. The main strategy to do this is to use a number of statistical measures to rule out \textit{adversarial models} or \textit{adversarial samplers}. An adversarial model is a theoretical description of a GBS setup that leads to a probability distribution that can be sampled efficiently. These classical models can be obtained, for instance, by using classical Gaussian states (such as thermal, distinguishable squeezed, or squashed states~\cite{madsen2022quantum, deng2023gaussian, martinez2023classical}) as input of the interferometers, an approach that intends to represent the effect of losses on the input squeezed states. On the other hand, an adversarial sampler is any efficient algorithm (not necessarily motivated by a physical model) that generates samples that intend to ``spoof'' the ground truth distribution.  
    
    An important drawback of this approach is that ruling out adversarial models and samplers with respect to the ground truth gives no information about how close this model is to the ideal squeezed state model. This opens the door to classical algorithms that reproduce the results of the experiments by exploiting the noise in the ground truth~\cite{oh2024classical}. This issue can be avoided by verifying GBS implementations using the ideal squeezed state model directly. In the next section we argue how we could do this using the linear cross-entropy between GBS models.
    
    \section{\label{sec:lxe_gbs} Linear cross-entropy score for GBS validation}

     Let $\bm{A}$ and $\bm{B}$ be two different models for the same GBS setup. We define the linear cross-entropy between $\bm{A}$ and $\bm{B}$, for a given total number of detected photons $N$, as   
    \begin{equation}
        \mathrm{LXE}\left(\bm{A},\bm{B};N\right)=\sum_{\bm{n}\in K(N)} \frac{\Pr(\bm{n}|\bm{A})\Pr(\bm{n}|\bm{B})}{\Pr(N|\bm{A})\Pr(N|\bm{B})},
        \label{eq:lxe_definition}
    \end{equation}
    where $K(N)=\{\bm{n}\,\vert\,\sum_{k=1}^M n_k=N\}$, and
    \begin{equation}
        \Pr(N|\bm{A})=\sum_{\bm{n}\in K(N)}\Pr(\bm{n}|\bm{A})
        \label{eq:probability_of_total_N}
    \end{equation}
    is the probability of detecting a total of $N$ photons given model $\bm{A}$.

    The linear cross-entropy belongs to a class of multiplicative measures of similarity between probability distributions that are collectively referred to as \textit{cross-entropy measures}~\cite{hangleiter2023computational}. These \textit{sample-efficient} measures (i.e., measures that can be estimated using polynomially many experimental samples) have been widely used for verifying quantum advantage claims in the field of quantum random sampling (see for instance Refs.~\cite{arute2019quantum,madsen2022quantum}). We are interested in using $\mathrm{LXE}\left(\bm{A},\bm{B};N\right)$ to assess the correlation between the ideal squeezed state model of a GBS implementation and the real distribution of the experimental samples.    
    
    The definition of $\mathrm{LXE}\left(\bm{A},\bm{B};N\right)$ closely resembles the definition of the \textit{linear cross-entropy benchmark} (XEB)~\cite{arute2019quantum, hangleiter2023computational}, a quantity which, despite having some limitations~\cite{gao2024limitations}, has been extensively used for the validation of Random Circuit Sampling (RCS) implementations. Roughly speaking, RCS consists of repeatedly applying cycles of randomly selected one- and two-qubit gates over a set of input qubits and then measuring their final state. The action of all the gates is represented by a unitary matrix $\bm{U}$, which can be approximated by a Haar-random unitary when the depth of the circuit (i.e. the number of cycles) is sufficiently large. The outcome of the measurement is a bit-string, i.e., a sequence of zeros and ones. Sampling from the ideal probability distribution of the outcomes $\bm{s}$ of a RCS implementation, $P_{\bm{U}}(\bm{s})$, is considered to be a computationally hard task whose cost grows exponentially with the number of input qubits and the depth of the circuit. Let $Q_{\bm{U}}(\bm{s})$ be the real distribution of the bit-strings $\bm{s}$. Then, the XEB reads~\cite{hangleiter2023computational}
    \begin{equation}
        \mathcal{F}(Q_{\bm{U}}, P_{\bm{U}})=2^N\sum_{\bm{s}}Q_{\bm{U}}(\bm{s})P_{\bm{U}}(\bm{s})-1,     \label{eq:XEB_definition}
    \end{equation}
    where $N$ is the number of input qubits and the sum is over the set of all possible bit-strings of length $N$. Notice that this expression can be seen as a normalized, shifted version of the linear cross-entropy between two models of the same RCS implementation.

    The definition of $\mathcal{F}(Q_{\bm{U}}, P_{\bm{U}})$ is such that we can identify two reference values for a \textit{typical} instance of RCS: if $Q_{\bm{U}}(\bm{s})= P_{\bm{U}}(\bm{s})$, $\mathbb{E}_{\bm{U}}[\mathcal{F}(Q_{\bm{U}}, P_{\bm{U}})]\approx 1$. If $Q_{\bm{U}}(\bm{s})=1/2^N$ (that is, the uniform distribution over bit-strings), $\mathbb{E}_{\bm{U}}[\mathcal{F}(Q_{\bm{U}}, P_{\bm{U}})]= 0$~\cite{arute2019quantum, hangleiter2023computational}. Here, $\mathbb{E}_{\bm{U}}[\cdot]$ indicates an average over the Haar measure of the unitary group. The presence of errors in experimental implementations of RCS leads to experimental samples with probability distributions that are uncorrelated with the ideal distribution, and that are exponentially close to the uniform distribution over bit-strings~\cite{dalzell2024random, boixo2018characterizing, arute2019quantum}. This implies that noisy implementations of RCS generally obtain values of the XEB close to zero. Moreover, obtaining an XEB satisfying $\mathcal{F}(Q_{\bm{U}}, P_{\bm{U}})> b/2^N$, with $b>1$, is considered to be a computationally hard task~\cite{hangleiter2023computational}. Thus, demonstrating quantum advantage for RCS implementations amounts to showing that the estimation of $\mathcal{F}(Q_{\bm{U}}, P_{\bm{U}})$ using experimental samples is sufficiently different from zero.

    There is no evidence that GBS architectures present the same sensibility to losses as RCS implementations. This makes elusive the definition of a measure as decisive as the XEB for the validation of GBS experiments. Nevertheless, one can use the idea of identifying reference values for \textit{typical} instances of GBS as a tool for the verification of quantum advantage claims. Indeed, in the spirit of the XEB, it is possible to use the average values (over Haar-random unitaries) of normalized versions of the linear cross-entropy between the ideal squeezed state model and the most commonly used adversarial models and samplers in order to obtain these reference values. We could consider, among the most widely used adversarial models, those associated to GBS setups using thermal, squashed or distinguishable squeezed states~\cite{madsen2022quantum, deng2023gaussian, zhong2021phase, martinez2023classical}. In addition to these, we could use classical algorithms that mimic the marginals, up to a certain order, of the ideal GBS distribution as adversarial samplers~\cite{madsen2022quantum, villalonga2021efficient}. The validation of quantum advantage would then amount to comparing how much a given GBS implementation \textit{scores} with respect to the other models and samplers. Notice that this approach makes no use of the ground truth distribution of the corresponding GBS experiment and, consequently, it is not necessary to make any assumptions about the hardness of sampling from it.

    In the remainder of this section, we will workout the details of the definition of this normalized, Haar-averaged version of the linear cross-entropy, which we will refer to as the \textit{linear cross-entropy score} (LXE score). Only one additional feature will be included in this definition: we focus on the asymptotic behavior of the linear cross-entropy as $M\rightarrow \infty$. We do this to take into account that the arguments justifying the complexity of GBS commonly require the setups to have a large number of modes, growing quadratically with the mean number of photons in the input squeezed states~\cite{grier2022complexity}. Moreover, by virtue of Levy's lemma~\cite{ledoux2001concentration}, one would expect that, as $M$ increases, typical instances of GBS will obtain values of the LXE score closer to the mean.
    
    Following the definition of the XEB, we associate the first reference value of the LXE score with the value that would be obtained by the uniform distribution over the sample space over which the linear cross-entropy is defined. Unlike in RCS, the entire sample space of detection patterns in GBS is infinite dimensional. This is why Eq.~\eqref{eq:lxe_definition} is defined for sets of samples with the same total number of detected photons. The first reference value of the LXE score is related to a GBS model that leads to a uniform probability distribution for each sector of $N$. Let us call such model $\bm{A}_{\mathrm{uni}}$. We can readily notice that $\Pr(\bm{n}|\bm{A}_{\mathrm{uni}})=f(|\bm{n}|)$, where $f(x)$ is a real-valued function satisfying $\sum_{N=0}^{\infty} f(N)=1$, and $0\leq f(N)\leq 1$ for all $N$, while $|\bm{n}|=\sum_{k=1}^M n_k$. This implies that, for an arbitrary model $\bm{B}$,
    \begin{align}
        &\mathrm{LXE}(\bm{A}_\mathrm{uni},\bm{B};N)=\frac{f(N)}{\Pr(\bm{n}|\bm{A}_{\mathrm{uni}})}\sum_{\bm{n}\in K(N)}\frac{\Pr(\bm{n}|\bm{B})}{\Pr(N|\bm{B})}\nonumber\\
        &=\frac{f(N)}{\sum_{\bm{n}\in K(N)}f(|\bm{n}|)}=\left(\sum_{\bm{n}\in K(N)}1\right)^{-1}\!\!=|K(N)|^{-1},
        \label{eq:lxe_uniform_model}
    \end{align}
    where $|K(N)|$ is the number of elements in $K(N)$. This value is equivalent to the number of weak $M$-compositions of $N$, i.e., the number of ordered partitions of $N$ having $M$ parts (with some of the parts allowed to be zero). It can be shown that $|K(N)|=\binom{M+N-1}{N}$~\cite{comtet1974advanced, flajolet2009analytic}. We will use $|K(N)|^{-1}$ as a normalization factor for the linear cross-entropy between any two models, thus setting the first reference value of the score to 1. It is worth mentioning that this same normalization term was used for the definition of other cross-entropy measures in Refs.~\cite{madsen2022quantum, oh2024classical}. 

    Even though we defined $\bm{A}_{\mathrm{uni}}$ according only to the properties of its corresponding probability distribution, it is important to keep in mind that this model can truly be associated to a GBS setup. It can be shown (see Appendix~\ref{app:misc}) that GBS setups using identical thermal states at the input of every mode of a lossless interferometer lead to probability distributions that are uniform for every sector of the total number of detected photons.

    With the normalization factor in place, we may now express the LXE score for a model $\bm{B}$ as
    \begin{equation}
        s(\bm{B};N)=\lim_{M\rightarrow\infty}  \binom{M+N-1}{N}\mathbb{E}_{\bm{U}}\left[\mathrm{LXE}(\bm{A}_{\mathrm{sqz}},\bm{B};N)\right].
        \label{eq:lxe_score_definition}
    \end{equation}
    Although this definition is adequate for finding reference values corresponding to some adversarial models, the validation of a set of experimental or adversarial samples requires an estimate of the LXE score rather than an analytical computation. This is due to the fact that not every adversarial sampler has an associated GBS model and, moreover, we have no information about the actual probability distribution that the experimental samples follow. Consider a set of $L$ samples $\{\bm{n}_k\}_{k=1}^L$ with the same total number of detected photons $N$, the estimated value of the score can be computed as 
    \begin{equation}
        \bar{s}(N)=\mathcal{C}(\bm{A}_{\mathrm{sqz}},N)\frac{1}{L}\sum_{k=1}^L\Pr(\bm{n}_k|\bm{A}_{\mathrm{sqz}}),
        \label{eq:lxe_score_estimation}
    \end{equation}
    where
    \begin{equation}
        \mathcal{C}(\bm{A}_{\mathrm{sqz}},N)=\binom{M+N-1}{N}\left[\Pr(N|\bm{A}_{\mathrm{sqz}})\right]^{-1}.
        \label{eq:coefficient_of_estimator}
    \end{equation}
    We can interpret the estimator $\bar{s}(N)$ as the average of the probabilities of each individual sample with respect to the probability distribution of the ideal model, multiplied by a normalization factor. 

    As mentioned before, the linear cross-entropy is a sample efficient measure, and thus $\bar{s}(N)$ can be estimated using polynomially many samples. Other sample efficient techniques, such as Bayesian testing or the heavy output generation test, typically require $10^3$ to $10^4$ samples per value of $N$ to be determined~\cite{zhong2020quantum, zhong2021phase, madsen2022quantum, martinez2023classical}. Due to the similarity between the computation of $\bar{s}(N)$ and these other validation techniques, we can expect that $\bar{s}(N)$ can also be estimated using this same range of number of samples.
    
    All the probabilities involved in the estimate of the score should be computed using a unitary matrix that is closely related to the actual sub-unitary matrix that describes the action of the lossy interferometer used in the experiment. If the GBS implementation is programmable, one can have access to the information about the ideal (lossless) configurations of all the gates (i.e., all the beamsplitters and phase shifters) used in the experiment. One then determines $\bar{s}(N)$ using the unitary matrix describing the action of all these ideal, unitary gates. 
    
    If the experiment has limited programmability and we only have access to the square sub-unitary matrix $\bm{T}$ describing the lossy interferometer, we can find a unitary matrix associated with $\bm{T}$ using its singular value decomposition~\cite{serafini2017quantum}. Indeed, we can always write $\bm{T}=\bm{U}_1\bm{D}\bm{U}_2^\dagger$, where $\bm{U}_1$, $\bm{U}_2$ are unitary and $\bm{D}$ is a diagonal matrix whose entries are the singular values of $\bm{T}$, which, in turn, are related to the transmission losses in the experiment. If we had no losses, we would be able to replace $\bm{D}$ by $\mathbb{I}_M$ and obtain the unitary transmission matrix $\bm{U}_1\bm{U}_2^\dagger$. On this account, we can interpret $\bm{U}_1\bm{U}_2^\dagger$ as the closest unitary matrix to $\bm{T}$, thus making it a reasonable choice for the computation of the LXE score. 

    However, it is important to mention that even if a given GBS implementation is not fully programmable, the authors of the experiment will be able to describe their ideal intended computation, i.e., they will have knowledge of the unitary $\bm{U}$ associated to their implementation.  

    The definition of the LXE score can be readily generalized to include GBS models with non-Gaussian noise. As mentioned in Sec.~\ref{sec:gbs_setup}, these models cannot be completely defined by a matrix of the form of Eq.~\eqref{eq:matrix_model}. However, if we have complete knowledge of the probability distribution associated to a given non-Gaussian model, we can still use Eq.~\eqref{eq:lxe_score_definition} to compute the LXE score, we need only replace $\Pr(\bm{n}|\bm{B})$ and $\Pr(N|\bm{B})$ in the definition of the linear cross-entropy by the adequate expressions of the probability distributions of interest. If we do not have complete knowledge of the probability distribution, but we have a set of samples associated to the non-Gaussian model, we can still use Eq.~\eqref{eq:lxe_score_estimation} to estimate the score.

    We associate the second reference value of the LXE score to the value that we would obtain if the real distribution of the experimental samples were $\Pr(\bm{n}|\bm{A}_{\mathrm{sqz}})$. The analytical computation of this \textit{ideal score} will be the subject of the next section. 
    
    \section{\label{sec:lxe_score_sqz} LXE score for the ideal squeezed state model}

    Consider a GBS implementation where the first $R$ of the $M$ input modes of the interferometer receive identical single-mode squeezed states with squeezing parameter $r$, while the remaining $M-R$ modes receive the vacuum state. The matrix $\bm{A}_{\mathrm{sqz}}$ describing this setup reads~\cite{hamilton2017gaussian, kruse2019detailed}
    \begin{equation}
        \bm{A}_{\mathrm{sqz}}=\tanh(r)\,\bm{V}\oplus\bm{V}^*,
        \label{eq:sqz_model_definition}
    \end{equation}
    where we define 
    \begin{align}
    \bm{V}&=\bm{U}\bm{\zeta}\bm{U}^{\mathrm{T}},\nonumber \\
    \bm{\zeta}&=\mathbb{I}_R\oplus\bm{0}_{M-R}     
    \end{align}
    (we explicitly indicate the size of the null matrix for clarity). The main result that we prove in this work states that
    \begin{equation}
        s(\bm{A}_\mathrm{sqz};2N)=\frac{4^{N}(N!)^2}{(2N)!}\left[\frac{(R-2)!!}{(R+2N-2)!!}\right]^2\sum_{\ell=1}^{2N}c_\ell R^{\ell}.
        \label{eq:main_result_lxe_sqz}
    \end{equation}
    where the $\{c_\ell\}$ are non-negative coefficients.

    \begin{figure}[!t]
        {
          \includegraphics[width=0.8\columnwidth]{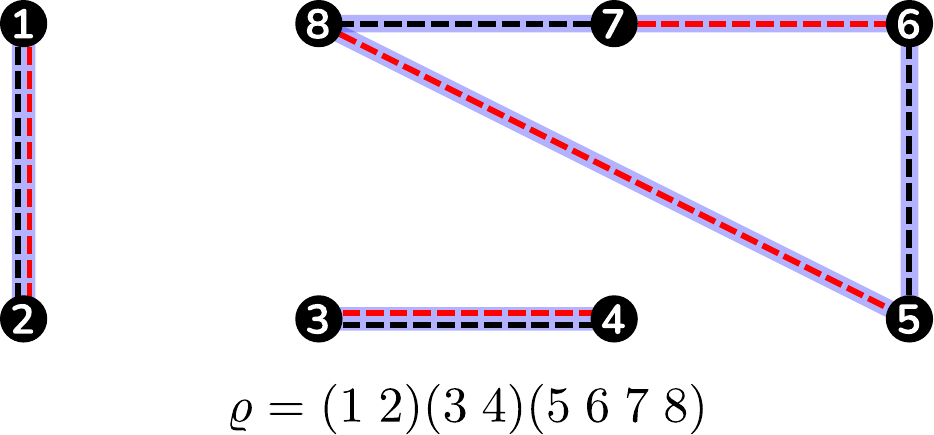}%
        }\\
        \vspace{0.5cm}
        {%
          \includegraphics[width=0.8\columnwidth]{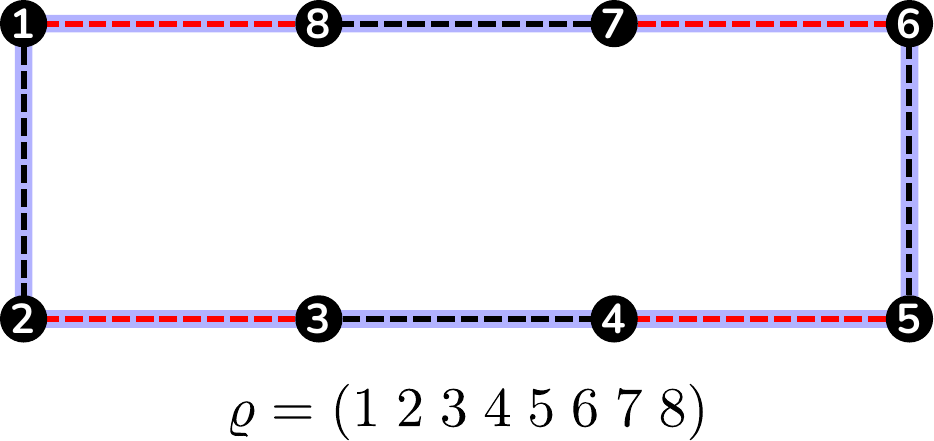}%
        }    
        \caption{Illustration of the definition of the undirected graphs $\bm{\Gamma}(\varrho)$ for $4N=8$. The vertices of the graph are represented by numbered black circles. The edges corresponding to $\{(2k-1,2k)\,\vert\,k\in\{1,\dots,2N\}\}$ are shown as black dashed lines, while the edges corresponding to $\{(\varrho(2k-1),\varrho(2k))\,\vert\,k\in\{1,\dots,2N\}\}$ are shown as red dashed lines. Each connected component of the graph is highlighted with a light blue, thick line. The top graph corresponds to the permutation, written in cycle notation, $\varrho = (1\,2)(3\,4)(5\,6\,7\,8)$. This permutation transforms the set of indices $\bm{g}=(g_1,g_2,g_3,g_4,g_5,g_6,g_7,g_8)$ as $\varrho(\bm{g})=(g_2,g_1,g_4,g_3,g_6,g_7,g_8,g_5)$, and has a total of three connected components. The bottom graph corresponds to the permutation $\varrho = (1\,2\,3\,4\,5\,6\,7\,8)$, which transforms $\bm{g}$ as $\varrho(\bm{g})=(g_2,g_3,g_4,g_5,g_6,g_7,g_8,g_1)$. In this case there is only one connected component.}
        \label{fig:definition_of_graphs}
    \end{figure}
    
    Unlike the ideal XEB reference value, which remains constant no matter the number of qubits in the random circuit, the ideal LXE score varies for each sector of the total number of detected photons. This, however, does not mean that $s(\bm{A}_\mathrm{sqz};2N)$ cannot be used as a tool for validation of GBS implementations. Indeed, we need only determine how different the scores of other models or samplers (which might also depend on $N$) are from the reference curve established by the ideal squeezed state model.

    A few definitions must be set in place before expressing how to compute the coefficients $\{c_\ell\}$. Let $\bm{j}=(j_1,\dots,j_{2N})$ be a fixed sequence of \textit{different} indices. We define the permutation $\Omega_{\bm{k}}\in S_{2N}$ (where we recall that $S_m$ stands for the symmetric group of degree $m$) by its action on $\bm{j}$ as
    \begin{align}
        \begin{split}
            \Omega_{\bm{k}}&[(j_1,\dots,j_{2N})]=\\
            &\bigoplus_{a=1}^N\bigoplus_{p=1}^{k_a}\omega_{a}[(j_{2v_{a-1}+2a(p-1)+1},\dots,j_{2v_{a-1}+2ap})],
        \end{split}
        \label{eq:big_omega_permutation_def}
    \end{align}
    where $v_a = \sum_{p=1}^a pk_p$, $v_0\equiv 0$, and the permutation $\omega_a \in S_{2a}$ (with $2a\leq 2N$) is defined by its action on a sequence $(g_1,\dots,g_{2a})$ as $\omega_a[(g_1,g_2,\dots,g_{2a-1}, g_{2a})]=(g_2,g_3,\dots,g_{2a},g_1)$. 
    Moreover, let $\varrho \in S_{4N}$, and, following Ref.~\cite{matsumoto2012general}, let $\bm{\Gamma}(\varrho)$ be an undirected graph whose vertices are $\{1,\dots,4N\}$, and whose edges are $\{(2k-1,2k)\,\vert\,k\in\{1,\dots,2N\}\}$ and $\{(\varrho(2k-1),\varrho(2k))\,\vert\,k\in\{1,\dots,2N\}\}$ (see Fig.~\ref{fig:definition_of_graphs} for an illustration of this definition). Then,
    \begin{align}
        \begin{split}
            c_\ell = \!\!\!\sum_{\bm{k}^{(c)}, \bm{l}^{(c)}}&\,\prod_{a=1}^N\frac{1}{k_a!l_a!(2a)^{k_a+l_a}}\\
            &\!\times\!\!\!\sum_{\sigma\in S_{2N}}b_\ell\left(\bm{j}\oplus\sigma(\bm{j}),\Omega_{\bm{k}}(\bm{j})\oplus\Omega_{\bm{l}}\circ\sigma(\bm{j})\right),
        \end{split}
        \label{eq:coefficients_sqz_lxe_score}
    \end{align}
    where $\bm{k}^{(c)}$ is a vector of $N$ non-negative integers that satisfy $k_1+2k_2+\dots+Nk_N=N$ (we define $\bm{l}^{(c)}$ in an analogous way), and 
    \begin{align}
        \begin{split}
            b_\ell(\bm{g},\bm{h})=\vert&\{\varrho\in S_{4N}\,\vert\,\bm{h}=\varrho(\bm{g})\text{ and }\\
            &\quad\bm{\Gamma}(\varrho)\text{ has }\ell\text{ connected components}\}\vert.
        \end{split}
        \label{eq:b_coefficients_definition}
    \end{align}

    Throughout the rest of this section, we explain in detail how to obtain the result in Eq.~\eqref{eq:main_result_lxe_sqz}, and we justify the definitions in Eqs.~\eqref{eq:big_omega_permutation_def} to~\eqref{eq:b_coefficients_definition}. 
    
    We divide our calculation of the ideal score into the following stages:
    \begin{enumerate}
        \item Express the linear cross-entropy as an integral over a number of real parameters. This form will have the advantage of not depending on the measurement outcomes $\bm{n}$.
        \item Write $\mathrm{LXE}(\bm{A}_{\mathrm{sqz}}, \bm{A}_{\mathrm{sqz}}; N)$ as a polynomial in the entries of $\bm{U}$. This will prove useful when computing the average over the Haar measure of the unitary group.
        \item Compute the integral over the real parameters.
        \item Calculate the average over Haar-random unitaries as $M\rightarrow\infty$.
    \end{enumerate}
    As a final stage, we gather all the results of the previous steps and complete the computation of the ideal score

    \subsection{\label{sec:integral_lxe}Integral form of the linear cross-entropy}
    Let $\bm{A}$ and $\bm{B}$ be two models for the same GBS setup. Adding the probabilities of every possible detection pattern $\bm{n}$ we can see that
    \[\sum_{\bm{n}}\Pr(\bm{n}|\bm{A})=\sum_{\bm{m}}\Pr(\bm{m}|\bm{B})=1,\] which implies that
    \[\sum_{\bm{n},\bm{m}}\Pr(\bm{n}|\bm{A})\Pr(\bm{m}|\bm{B}) = 1.\] 
    We can recast this expression in terms of the hafnian using Eq.~\eqref{eq:gbs_probability_distribution}:
    \begin{equation}
        \sum_{\bm{n},\bm{m}} \frac{\mathrm{haf}\left[\bm{A}_{\bm{n}}\right]}{\bm{n}!}\frac{\mathrm{haf}\left[\bm{B}_{\bm{m}}\right]}{\bm{m}!}=\frac{1}{\Pr(\bm{0}|\bm{A})\Pr(\bm{0}|\bm{B})}.
        \label{eq:integral_lxe_1}
    \end{equation}

    Define the matrix 
    \begin{align}
       \bm{D}(\bm{\phi})=\mathrm{diag}(e^{i\phi_1},\dots,e^{i\phi_M})  \, ,\label{eq:defD}
    \end{align}
    where $\bm{\phi}=(\phi_1,\dots,\phi_M)$ is a vector of real parameters with $\phi_k\in[0,2\pi]\,\,\forall k$, and let $\bm{W}(\phi)=\bm{D}(\bm{\phi})\oplus\bm{D}(\bm{\phi})$. Notice that $\bm{W}^*(\bm{\phi})=\bm{W}(-\bm{\phi})$. Transforming matrices $\bm{A}$ and $\bm{B}$ in Eq.~\eqref{eq:integral_lxe_1} according to $\bm{A}\rightarrow \alpha \bm{W}(\bm{\phi})\bm{A}\bm{W}(\bm{\phi})$ and $\bm{B}\rightarrow \beta \bm{W}^*(\bm{\phi})\bm{B}\bm{W}^*(\bm{\phi})$, where $\alpha,\,\beta\in [0,1)$, we obtain (we will drop the explicit dependence of matrices on $\bm{\phi}$ in order to shorten the notation)
    \begin{align}
        \begin{split}
            \!\!\sum_{\bm{n},\bm{m}}&\frac{\mathrm{haf}\left[\alpha\left(\bm{W}\bm{A}\bm{W}\right)_{\bm{n}}\right]}{\bm{n}!}\frac{\mathrm{haf}\left[\beta\left(\bm{W}^*\bm{B}\bm{W}^*\right)_{\bm{m}}\right]}{\bm{m}!}\\
            &\quad\quad=q(\alpha,\bm{\phi},\bm{A})\,q(\beta,-\bm{\phi},\bm{B}),
        \end{split}
        \label{eq:integral_lxe_2}
    \end{align}
    where 
    \begin{align}
        \begin{split}
            q(\alpha,\bm{\phi},\bm{A})&=\left[\mathrm{det}\left(\mathbb{I}_{2M}-\alpha\bm{X}\bm{W}\bm{A}\bm{W}\right)\right]^{-1/2}\\
            &=\left[\mathrm{det}\left(\mathbb{I}_{2M}-\alpha\bm{\Omega}\bm{A}\right)\right]^{-1/2},
        \end{split}
        \label{eq:q_function_def}
    \end{align}
    and $\bm{\Omega}(\bm{\phi})=\bm{W}(\bm{\phi})\bm{X}\bm{W}(\bm{\phi})$\footnote{To write the second line of this relation, we used Sylvester's theorem, which states that if $\bm{A}$ and $\bm{B}$ are $m\times n$ and $n\times m$ matrices, respectively, such that $\bm{AB}$ and $\bm{BA}$ are trace-class, then $\mathrm{det}(\mathbb{I}_m + AB) = \mathrm{det}(\mathbb{I}_n + BA)$}.

    We can factorize the matrices inside the hafnian as $\left(\bm{W}\bm{A}\bm{W}\right)_{\bm{n}} = \bm{\bar{W}}_{\bm{n}}\bm{A}_{\bm{n}}\bm{\bar{W}}_{\bm{n}}$, where $\bm{\bar{W}}_{\bm{n}}$ is a diagonal matrix obtained from $\bm{W}$ by repeating $n_k$ times the entry $e^{i\phi_k}$ (see Appendix~\ref{app:misc} for a proof of this statement). For instance, let $M=2$ and $\bm{W}=\mathrm{diag}(e^{i\phi_1}, e^{i\phi_2}, e^{i\phi_1}, e^{i\phi_2})$. If $\bm{n}=(1,2)$, then 
    $\bm{\bar{W}}_{\bm{n}}=\mathrm{diag}(e^{i\phi_1}, e^{i\phi_2}, e^{i\phi_2},e^{i\phi_1}, e^{i\phi_2}, e^{i\phi_2})$. Notice that $\bm{\bar{W}}_{\bm{n}}$ has size $2|\bm{n}|\times 2|\bm{n}|$, where $|\bm{n}|=\sum_{k=1}^M n_k$.
    
    This factorization, in turn, allows us to use the following property of the hafnian~\cite{barvinok2016combinatorics}: for any symmetric $m\times m$ matrix $\bm{O}$ and a diagonal matrix $\bm{S}=\mathrm{diag}(s_1,\dots,s_m)$
    \begin{equation}
        \mathrm{haf}[\bm{S}\bm{O}\bm{S}]=\left(\prod_{k=1}^m s_k\right)\mathrm{haf}[\bm{O}].
        \label{eq:hafnian_property}
    \end{equation}
    We have
    \begin{align}
        \begin{split}
            \mathrm{haf}\left[\alpha\left(\bm{W}\bm{A}\bm{W}\right)_{\bm{n}}\right]&=\mathrm{haf}\left[\alpha\bm{\bar{W}}_{\bm{n}}\bm{A}_{\bm{n}}\bm{\bar{W}}_{\bm{n}}\right]\\
            &=\alpha^{|\bm{n}|}\prod_{k=1}^{2|\bm{n}|}(\bm{\bar{W}}_{\bm{n}})_{k,k}\,\mathrm{haf}\left[\bm{A}_{\bm{n}}\right].
        \end{split}
        \label{eq:integral_lxe_3}
    \end{align}
    Noting that $\prod_{k=1}^{2|\bm{n}|}(\bm{\bar{W}}_{\bm{n}})_{k,k}=\prod_{k=1}^Me^{2in_k\phi_k} = e^{2i\bm{n}\cdot\bm{\phi}}$, we obtain
    \begin{equation}
        \mathrm{haf}\left[\alpha\left(\bm{W}\bm{A}\bm{W}\right)_{\bm{n}}\right]=\alpha^{|\bm{n}|}e^{2i\bm{n\cdot\bm{\phi}}}\,\mathrm{haf}\left[\bm{A}_{\bm{n}}\right].
        \label{eq:integral_lxe_4}
    \end{equation}
    Using the same argument, it can be shown that $\mathrm{haf}[\beta(\bm{W}^*\bm{B}\bm{W}^*)_{\bm{m}}]=\beta^{|\bm{m}|}e^{-2i\bm{m}\cdot\bm{\phi}}\,\mathrm{haf}[\bm{B}_{\bm{m}}]$. Replacing these expressions into Eq.~\eqref{eq:integral_lxe_2}, we obtain
    \begin{align}
        \begin{split}
            &q(\alpha,\bm{\phi},\bm{A})\,q(\beta,-\bm{\phi},\bm{B})=\\
            &\sum_{\bm{n},\bm{m}}\alpha^{|\bm{n}|}\beta^{|\bm{m}|}e^{2i(\bm{n}-\bm{m})\cdot\bm{\phi}}\,\frac{\mathrm{haf}\left[\bm{A}_{\bm{n}}\right]}{\bm{n}!}\frac{\mathrm{haf}\left[\bm{B}_{\bm{m}}\right]}{\bm{m}!}.
        \end{split}
        \label{eq:integral_lxe_5}
    \end{align}

    Integrating both sides of the last equation with respect to $d\bm{\phi}=d\phi_1\cdots d\phi_M$, and taking into account that \[\int_0^{2\pi}e^{2i(\bm{n}-\bm{m})\cdot\bm{\phi}}d\bm{\phi}=(2\pi)^M\delta_{\bm{m},\bm{n}},\] where $\delta_{\bm{m},\bm{n}}=\delta_{m_1,n_1}\cdots\delta_{m_M,n_M}$, we can write
    \begin{align}
        \begin{split}
            \frac{1}{(2\pi)^M}&\int_0^{2\pi}q(\alpha,\bm{\phi},\bm{A})\,q(\beta,-\bm{\phi},\bm{B})\,d\bm{\phi}=\\
            &\sum_{N=0}^\infty\alpha^{N}\beta^{N}\!\!\sum_{\bm{n}\in K(N)}\frac{\mathrm{haf}\left[\bm{A}_{\bm{n}}\right]}{\bm{n}!}\frac{\mathrm{haf}\left[\bm{B}_{\bm{n}}\right]}{\bm{n}!},
        \end{split}
        \label{eq:integral_lxe_6}
    \end{align}
    where we have taken into account that we can decompose the sum over all possible detection patterns as $\sum_{\bm{n}}=\sum_{N=0}^\infty \sum_{\bm{n}\in K(N)}$.

    By repeatedly differentiating with respect to $\alpha$ and $\beta$, and then evaluating at $\alpha=\beta=0$, we can single out the sum involving only the elements of $K(N)$. The result of this procedure reads:
    \begin{align}
        \begin{split}
            \frac{1}{(2\pi)^M}&\int_0^{2\pi}\left.\frac{\partial^Nq(\alpha,\bm{\phi},\bm{A})}{\partial \alpha^N}\,\frac{\partial^Nq(\beta,-\bm{\phi},\bm{B})}{\partial\beta^N}\right|_{\substack{\alpha=0\\\beta=0}}\,d\bm{\phi}\\
            &=(N!)^2\!\!\sum_{\bm{n}\in K(N)}\frac{\mathrm{haf}\left[\bm{A}_{\bm{n}}\right]}{\bm{n}!}\frac{\mathrm{haf}\left[\bm{B}_{\bm{n}}\right]}{\bm{n}!}.
        \end{split}
        \label{eq:integral_lxe_7}
    \end{align}

    We may readily observe that the previous expression leads to the definition of the linear cross-entropy; we need only include the vacuum probabilities for both models and the term $\Pr(N|\bm{A})\Pr(N|\bm{B})$. Putting all these pieces together, we obtain the final expression of the integral form:
    \begin{align}
        \begin{split}
            \mathrm{LXE}&\left(\bm{A},\bm{B};N\right)=\frac{1}{(2\pi)^M}\,\mathcal{D}(\bm{A},\bm{B};N)\\
            &\!\!\times\int_0^{2\pi}\left.\frac{\partial^Nq(\alpha,\bm{\phi},\bm{A})}{\partial \alpha^N}\,\frac{\partial^Nq(\beta,-\bm{\phi},\bm{B})}{\partial\beta^N}\right|_{\substack{\alpha=0\\\beta=0}}\,d\bm{\phi},
        \end{split}
        \label{eq:integral_lxe_end}
    \end{align}
    where 
    \begin{equation}
        \mathcal{D}(\bm{A},\bm{B};N)=\frac{1}{(N!)^2}\frac{\Pr(\bm{0}|\bm{A})\Pr(\bm{0}|\bm{B})}{\Pr(N|\bm{A})\Pr(N|\bm{B})}.
        \label{eq:coefficient_lxe_1}
    \end{equation}
    
    Although it may seem that the determination of $\Pr(N|\bm{A})\Pr(N|\bm{B})$, and therefore $\mathcal{D}(\bm{A},\bm{B};N)$, requires the knowledge of all detection patterns in $K(N)$, by using similar arguments to those that led to Eq.~\eqref{eq:integral_lxe_5}, it can be shown that  
    \begin{equation}
        \mathcal{D}(\bm{A},\bm{B};N)=\left[ 
        \frac{\partial^N q(\alpha, \bm{0},\bm{A})}{\partial \alpha^N} 
        \left.\frac{\partial^N q(\beta, \bm{0},\bm{B})}{\partial \beta^N}\right|_{\substack{\alpha=0\\\beta=0}}
    \right]^{-1}.
        \label{eq:coefficient_lxe_2}
    \end{equation}
    A detailed proof of this statement can be found in Appendix~\ref{app:misc}.

    We turn now our attention to the task of finding an expression for $q(\alpha,\bm{\phi},\bm{A})$ that allows us to readily compute its partial derivatives with respect to $\alpha$. We will do this by recasting $q(\alpha,\bm{\phi},\bm{A})$ into a power series in $\alpha$.

    Notice that $\mathbb{I}_{2M}-\alpha\bm{\Omega}\bm{A}=\exp\left[\log\left(\mathbb{I}_{2M}-\alpha\bm{\Omega}\bm{A}\right)\right]$, and recall the relation $\mathrm{det}[\exp(\bm{A})]=\exp[\mathrm{tr}(\bm{A})]$ (we use $\mathrm{tr}(\cdot)$ to indicate the trace of a matrix, while $\mathrm{Tr}(\cdot)$ indicates the trace of an operator). Using the Taylor series expansion $\log(1+x)=\sum_{l=1}^\infty(-1)^{l+1}x^l/l$, we may write $q(\alpha,\bm{\phi},\bm{A})$ as
    \begin{equation}
        q(\alpha,\bm{\phi},\bm{A})=\exp\left[\sum_{l=1}^\infty y_l\frac{\alpha^l}{l}\right],\quad y_l=\frac{1}{2}\mathrm{tr}[(\bm{\Omega}\bm{A})^l]\, ,
        \label{eq:q_function_expansion_1}
    \end{equation}
    where the $\bm{\phi}$ dependence of every $y_l$ is through $\bm{\Omega}=\bm{\Omega(\phi)}$.
    Details about the convergence of this series expansion for a wide number of GBS models, including the ideal squeezed state model, can also be found in Appendix~\ref{app:misc}. 
    
    In the form of Eq.~\eqref{eq:q_function_expansion_1}, $q(\alpha,\bm{\phi},\bm{A})$ becomes the generating function of the \textit{cycle index} of the symmetric group $Z_n$~\cite{comtet1974advanced,roberts2009applied,wilf2005generatingfunctionology}, which leads to the expression
    \begin{equation}
        q(\alpha,\bm{\phi},\bm{A}) = \sum_{n=0}^\infty Z_n(y_1,\dots,y_n)\alpha^n,
        \label{eq:q_function_expansion_2}
    \end{equation}
    where
    \begin{equation}
        Z_n(y_1,\dots,y_n)=\sum_{\bm{k}^{(c)}}\,\prod_{a=1}^n\frac{1}{k_a!a^{k_a}}\prod_{a=1}^ny_a^{k_a},
        \label{eq:cycle_index}
    \end{equation}
    and the sum extends over all possible $\bm{k}=(k_1,\dots,k_n)$ whose non-negative, integer components satisfy the constraint $k_1 + 2k_2+\cdots+ nk_n = n$ (we use the notation $\bm{k}^{(c)}$ to indicate the constraint over the components of $\bm{k}$.
    
    We can now readily see that 
    \begin{equation}
        \left.\frac{\partial^N q(\alpha, \bm{\phi},\bm{A})}{\partial \alpha^N}\right|_{\alpha=0}= N!Z_N[\bm{y}(\bm{\phi},\bm{A})],
        \label{eq:derivative_q_function}
    \end{equation}
    where we have defined the vector $\bm{y}(\bm{\phi},\bm{A})=(y_1,\dots,y_N)$ in order to make explicit the dependence of the $\{y_k\}$ on $\bm{\phi}$ and $\bm{A}$. Using this expression, we can recast the linear cross-entropy as
    \begin{align}
        \begin{split}
            \mathrm{LXE}&\left(\bm{A},\bm{B};N\right)=\frac{1}{(2\pi)^M}\,\bar{\mathcal{D}}(\bm{A},\bm{B};N)\\
            &\quad\times\int_0^{2\pi}Z_N[\bm{y}(\bm{\phi},\bm{A})]Z_N[\bm{y}(-\bm{\phi},\bm{B})]\,d\bm{\phi},
        \end{split}
        \label{eq:integral_lxe_cycle}
    \end{align}
    with 
    \begin{align}
    \bar{\mathcal{D}}(\bm{A},\bm{B};N)=(N!)^2\mathcal{D}(\bm{A},\bm{B};N)\, .
    \label{eq:calDdefine}
    \end{align}

    \subsection{\label{sec:lxe_as_polynomial}LXE as a polynomial in the entries of Haar-random unitaries}

    Consider now the ideal squeezed state model $\bm{A}_{\mathrm{sqz}}$. From the definition of $q(\alpha,\bm{\phi},\bm{A})$ we can see that
    \begin{align}
        q(\alpha,\bm{0},\bm{A}_{\mathrm{sqz}})&=\left[\mathrm{det}\left[\mathbb{I}_{2M}-\alpha\tanh(r)\bm{X}(\bm{V}\oplus\bm{V}^*)\right]\right]^{-1/2}\nonumber\\
        &=\left[\mathrm{det}\left(\mathbb{I}_{M}-\alpha^2\tanh^2(r)\,\bm{\zeta}\right)\right]^{-1/2}\nonumber\\
        &=\left[1-\alpha^2\tanh^2(r)\right]^{-R/2},
        \label{eq:q_function_sqz_at_0}
    \end{align}
    where we remind the reader that the squeezed states are sent in the first $R$ modes of the interferometer. Expanding this expression in a Taylor series about $\alpha=0$, we can prove the relation 
    \begin{equation}
        \left.\frac{\partial^{2N}q(\alpha,\bm{0},\bm{A}_{\mathrm{sqz}})}{\partial\alpha^{2N}}\right|_{\alpha=0}=(2N)!\tanh^{2N}(r)\binom{\frac{R}{2}+N-1}{N}.
        \label{eq:der_q_function_sqz_at_0}
    \end{equation}
    Note that the derivatives for odd $N$ identically vanish when evaluated at $\alpha=0$. This is due to the fact that squeezed states have support only over Fock states with an even number of photons and, moreover, we are considering a lossless (i.e., energy conserving) interferometer. Using Eq.~\eqref{eq:der_q_function_sqz_at_0} we reach the result 
    \begin{equation}
        \bar{\mathcal{D}}(\bm{A}_{\mathrm{sqz}},\bm{A}_{\mathrm{sqz}};2N) = \frac{1}{\tanh^{4N}(r)}\binom{\frac{R}{2}+N-1}{N}^{-2}.
        \label{eq:coeff_lxe_sqz}
    \end{equation}

    For the general case of $\bm{\phi}\neq \bm{0}$, we can write
    \begin{align}
        \begin{split}
            q(\alpha,\bm{\phi},\bm{A}_{\mathrm{sqz}})&=\mathrm{det}\left[\mathbb{I}_{2M}-\alpha\tanh(r)\bm{\Omega}(\bm{V}\oplus\bm{V}^*)\right]^{-1/2}\\
            &=\mathrm{det}\left[\mathbb{I}_{M}-\alpha^2\tanh^2(r)\bm{D}^2\bm{V}\bm{D}^2\bm{V}^*\right]^{-1/2},\\
        \end{split}
        \label{eq:q_function_sqz}
    \end{align}
    which can be recast in the form of Eq.~\eqref{eq:q_function_expansion_2} as
    \begin{equation}
        q(\alpha,\bm{\phi},\bm{A}_{\mathrm{sqz}}) = \sum_{n=0}^\infty Z_n(u_1,\dots,u_n)\tanh^{2n}(r)\alpha^{2n},
        \label{eq:q_function_sqz_2}
    \end{equation}
    where $u_k=\frac{1}{2}\mathrm{tr}\left[(\bm{D}^2\bm{V}\bm{D}^2\bm{V}^*)^k\right]$. Each $u_k$ depends on $\bm{\phi}$ through $\bm{D}=\bm{D(\phi)}$. 
    Just as before, we can readily see that the derivatives of $q(\alpha,\bm{\phi},\bm{A}_{\mathrm{sqz}})$ are different from zero when evaluated at $\alpha=0$ only if $N$ is even. Defining the vector $\bm{u}(\bm{\phi,\bm{U}})=(u_1,\dots,u_N)$, these derivatives can be expressed as
    \begin{equation}
        \left.\frac{\partial^{2N}q(\alpha,\bm{\phi},\bm{A}_{\mathrm{sqz}})}{\partial\alpha^{2N}}\right|_{\alpha=0}=(2N)!\tanh^{2N}(r)Z_N[\bm{u}(\bm{\phi},\bm{U})],
        \label{eq:der_q_function_sqz}
    \end{equation}
    where $Z_N$ is given in Eq.~\eqref{eq:cycle_index}. We can therefore compute $\mathrm{LXE}(\bm{A}_{\mathrm{sqz}}, \bm{A}_{\mathrm{sqz}};2N)$ using the relation
    \begin{align}
        \begin{split}
            \mathrm{LXE}&\left(\bm{A}_{\mathrm{sqz}},\bm{A}_{\mathrm{sqz}};2N\right)=\frac{1}{(2\pi)^M}\,\binom{\frac{R}{2}+N-1}{N}^{-2}\\
            &\quad\times\int_0^{2\pi}Z_N[\bm{u}(\bm{\phi},\bm{U})]Z_N[\bm{u}(-\bm{\phi},\bm{U})]\,d\bm{\phi}.
        \end{split}
        \label{eq:integral_lxe_cycle_sqz}
    \end{align}

    The definition of the $\{u_k\}$ as traces of powers of the matrix $\bm{D}^2\bm{V}\bm{D}^2\bm{V}^*$ implies they can be expanded as polynomials in the entries of matrices $\bm{V}$, $\bm{D}$, and their complex conjugates. Consequently, the integrand in Eq.~\eqref{eq:integral_lxe_cycle_sqz} will also be a polynomial in the entries of these matrices. The dependence of each term in the expansion on the entries of $\bm{V}$ will have the following general structure: $V_{g_1,g_2}\cdots V_{g_{2l-1},g_{2l}}V^*_{h_1,h_2}\cdots V^*_{h_{2m-1},h_{2m}}$, where $\bm{g}=(g_1,\dots,g_{2l})$ and $\bm{h}=(h_1,\dots,h_{2m})$ are sequences of indices that take values in subsets of $\{1,\dots, M\}$. These terms can be recast as polynomials in the entries of the matrix $\bm{U}$: 
    \begin{align}
        \begin{split}
            V_{g_1,g_2}\cdots V_{g_{2l-1},g_{2l}}V^*_{h_1,h_2}&\cdots V^*_{h_{2m-1},h_{2m}}\\
            &=\sum_{\bm{\mu},\bm{\nu}}\zeta_{\bm{\mu}}\zeta_{\bm{\nu}}\,\mathcal{U}\left(\bm{g},\bm{{\bar{\mu}}}\,\vert\,\bm{h},\bm{\bar{\nu}}\right),
        \end{split}
        \label{eq:monomial_V_function_def}
    \end{align}
    where the indices in $\bm{\mu}=(\mu_1,\dots,\mu_l)$, $\bm{\nu}=(\nu_1,\dots,\nu_m)$ take values in $\{1,\dots, M\}$, so $\sum_{\bm{\mu}}\equiv\sum_{\mu_1=1}^M\dots\sum_{\mu_{l}=1}^M$, and 
    $\bm{\bar{\mu}}=(\mu_1,\mu_1,\dots,\mu_l,\mu_l)$, $\bm{\bar{\nu}}=(\nu_1,\nu_1,\dots,\nu_m,\nu_m)$. We conveniently write
    \begin{align}
        \zeta_{\bm{\mu}}=\zeta_{\mu_1}\cdots\zeta_{\mu_l}\, ,
        \label{eq:zetadefine}
    \end{align} 
    with $\{\zeta_{k}\}$ the diagonal entries of $\bm{\zeta}$, and
    \begin{align}
        \begin{split}
            \!\!\mathcal{U}\left(\bm{g},\bm{{\bar{\mu}}}\,\vert\,\bm{h},\bm{\bar{\nu}}\right)&=U_{g_1,\mu_1}U_{g_2,\mu_1}\cdots U_{g_{2l-1},\mu_l}U_{g_{2l},\mu_l}\\
            &\times U^*_{h_1,\nu_1}U^*_{h_2,\nu_1}\cdots U^*_{h_{2m-1},\nu_m}U^*_{h_{2m},\nu_m}.   
        \end{split}
        \label{eq:monomial_U_function_def}
    \end{align}

    Consider the sequences $\bm{j}=(j_1,\dots, j_{2N})$ and $\bm{j}'=(j'_1,\dots, j'_{2N})$, where $j_k,j_l'\in\{1,\dots, M\}\,\,\forall k,l$. Moreover, consider the permutation $\Omega_{\bm{k}}\in S_{2N}$ defined in Eq.~\eqref{eq:big_omega_permutation_def}: 
    \begin{align*}
        \begin{split}
            \Omega_{\bm{k}}&[(j_1,\dots,j_{2N})]=\\
            &\bigoplus_{a=1}^N\bigoplus_{p=1}^{k_a}\omega_{a}[(j_{2v_{a-1}+2a(p-1)+1},\dots,j_{2v_{a-1}+2ap})],
        \end{split}
    \end{align*}
    where, let us recall, $v_a = \sum_{p=1}^a pk_p$, $v_0\equiv 0$, and the permutation $\omega_a$ transforms the sequence $(g_1,\dots, g_{2a})$ as $\omega_a[(g_1,g_2,\dots,g_{2a-1}, g_{2a})]=(g_2,g_3,\dots,g_{2a},g_1)$. Then, we can write (see Appendix~\ref{app:index_structure} for details)
    \begin{align}
        \begin{split}
            \mathrm{LXE}&\left(\bm{A}_{\mathrm{sqz}},\bm{A}_{\mathrm{sqz}};2N\right)=\binom{\frac{R}{2}+N-1}{N}^{-2}\\
            &\times\sum_{\bm{k}^{(c)}, \bm{l}^{(c)}}\,\prod_{a=1}^N\frac{1}{k_a!l_a!(2a)^{k_a+l_a}}\sum_{\bm{j},\bm{j}'}I({\bm{j}},{\bm{j}'})\\
            &\times\sum_{\bm{\mu},\bm{\nu}}\zeta_{\bm{\mu}}\zeta_{\bm{\nu}}\,\mathcal{U}\left(\bm{j}\oplus\bm{j}',\bm{\bar{\mu}}\,\vert\,\Omega_{\bm{k}}(\bm{j})\oplus\Omega_{\bm{l}}(\bm{j}'),\bm{\bar{\nu}}\right),
        \end{split}
        \label{eq:lxe_sqz_cycle_U_polynomial}
    \end{align}
    with
    \begin{equation}
        I(\bm{j},\bm{j}')=\frac{1}{(2\pi)^M}\int_0^{2\pi}\exp\left[\sum_{m\in \bm{j},n\in \bm{j}'}\!\!\!\!2i\,(\phi_m-\phi_n)\right]d\bm{\phi}.
        \label{eq:phases_integral}
    \end{equation}

    \subsection{\label{sec:integral_over_phases}Integrating the phases away}
    
    By inspection of Eq.~\eqref{eq:phases_integral} we can recognize that $I(\bm{j},\bm{j}')$ will vanish whenever the sum inside the exponential is different from zero. Indeed, for this case, there must be at least one term of the form $e^{2 iz \phi_p}$, for some non-zero integer $z$ and some $p\in \bm{j}\text{ or }\bm{j}'$, that is not canceled out and, when integrated with respect to $d\phi_p$, makes the whole integral vanish. When the sum inside the exponential is equal to zero, $I(\bm{j},\bm{j}')=1$. We may then think of $I(\bm{j},\bm{j}')$ as an indicator function that, given a fixed $\bm{j}$, allows us to keep track of all the ways we can set $\bm{j}'$ in order to make the sum inside the exponential vanish. Furthermore, notice that the sum inside the exponential will be identically zero whenever $\bm{j}'$ is a permutation of $\bm{j}$. Since $\bm{j}$ might have indices with repeated values, we must take into account that only the \textit{different} permutations that take $\bm{j}$ into $\bm{j}'$ should be identified by $I(\bm{j},\bm{j}')$. 
    
    In Appendix~\ref{app:integral} we describe how to use the previous considerations to write Eq.~\eqref{eq:phases_integral} in terms of Kronecker deltas. The final result reads:
    \begin{align}
        \begin{split}
            I(\bm{j},\bm{j'})&=\sum_{\Lambda\in\mathcal{Q}[\bm{j}]}\frac{1}{\Lambda!}F(\bm{j}',\bm{j}[\Lambda,\{j_\lambda\}])\\
            &\qquad\times\left[\prod_{\lambda\in\Lambda}\prod_{f\in\lambda}\delta_{j_\lambda,f}\right]\prod_{(\lambda\neq\mu)\in\Lambda}(1-\delta_{j_\lambda,j_\mu}).
        \end{split}
        \label{eq:phases_integral_result}
    \end{align}
    
    In this expression $\Lambda$ represents a \textit{set partition} of $\bm{j}$, i.e., a collection of non-empty, mutually disjoint subsets of $\bm{j}$ (which are usually called \textit{blocks}), whose union is equal to $\bm{j}$. $\mathcal{Q}[\bm{j}]$ is the set of all partitions of $\bm{j}$. The set $\{j_\lambda\}$, which depends on a given partition $\Lambda$, is called the set of \textit{representative indices} of $\Lambda$, and is constructed by choosing one element, any element, of each block $\lambda \in \Lambda$. 
    
    The sequence of indices $\bm{j}[\Lambda,\{j_\lambda\}]$ is constructed from $\bm{j}$ and $\Lambda$ by using the following prescription: take a partition $\Lambda \in \mathcal{Q}[\bm{j}]$ and choose a representative index $j_\lambda$ for each block $\lambda \in \Lambda$; then replace all the elements in $\bm{j}$ that belong to the same block $\lambda$ by the corresponding representative index $j_\lambda$. For example, consider $\bm{j}=(j_1,j_2,j_3,j_4)$ and $\Lambda = \{\{j_1,j_3\},\{j_2,j_4\}\}$. Let $\{j_3, j_2\}$ be the representative indices of the partition, then $\bm{j}[\Lambda,\{j_3,j_2\}]=(j_3,j_2,j_3,j_2)$.
    
    Notice that the number of representative indices is equal to the number of blocks in $\Lambda$. Let $|\lambda|$ denote the length (i.e. the number of elements) of each block $\lambda \in \Lambda$, then $\Lambda!=\prod_{\lambda\in\Lambda}|\lambda|!$.
    
    For two sequences of indices $\bm{g}=(g_1,\dots,g_m)$ and $\bm{h}=(h_1,\dots,h_m)$, $F(\bm{h},\bm{g})$ is defined as
    \begin{equation}
        F(\bm{h},\bm{g})=\sum_{\sigma \in S_m}\prod_{a=1}^m\delta_{h_a,\sigma(g_a)},
        \label{eq:permanent_deltas}
    \end{equation}
    where the $\{\sigma(g_a)\}$ stand for the components of $\sigma(\bm{g})$.   

    Combining Eqs.~\eqref{eq:lxe_sqz_cycle_U_polynomial} and~\eqref{eq:phases_integral_result}, and after a careful manipulation of all the Kronecker deltas involved (see Appendix~\ref{app:integral} for details), we can express the LXE as
    \begin{align}
        \begin{split}
            &\mathrm{LXE}\left(\bm{A}_{\mathrm{sqz}},\bm{A}_{\mathrm{sqz}};2N\right)=\binom{\frac{R}{2}+N-1}{N}^{-2}\\
            &\times\sum_{\bm{k}^{(c)}, \bm{l}^{(c)}}\,\prod_{a=1}^N\frac{1}{k_a!l_a!(2a)^{k_a+l_a}}\sum_{\Lambda\in \mathcal{Q}[\bm{j}]}\sum_{\mathrm{diff.}\{j_\lambda\}}\sum_{\sigma\in S_{2N}}\frac{1}{\Lambda!}\\
            &\times\sum_{\bm{\mu},\bm{\nu}}\zeta_{\bm{\mu}}\zeta_{\bm{\nu}}\,\mathcal{U}\left(\bm{j}_\Lambda\oplus\sigma(\bm{j}_\Lambda),\bm{\bar{\mu}}\,\vert\,\Omega_{\bm{k}}(\bm{j}_\Lambda)\oplus\Omega_{\bm{l}}\circ\sigma(\bm{j}_\Lambda),\bm{\bar{\nu}}\right),
        \end{split}
        \label{eq:lxe_sqz_after_integral}
    \end{align}
    where we introduce the short notation $\bm{j}_\Lambda\equiv \bm{j}[\Lambda,\{j_\lambda\}]$, and $\sigma\circ\tau$ indicates the composition of permutations. The subscript $\mathrm{diff.}\{j_\lambda\}$ indicates that the sum  must be done over representative indices taking \textit{different} values in $\{1,\dots,M\}$.

    \subsection{\label{sec:average_haar}Average over the Haar-random unitaries in the asymptotic limit}

    Following the result in Eq.~\eqref{eq:lxe_sqz_after_integral}, we can see that computing the average value of $\mathrm{LXE}\left(\bm{A}_{\mathrm{sqz}},\bm{A}_{\mathrm{sqz}};2N\right)$ over Haar-random unitaries amounts to calculating the expected value of the polynomial 
    \begin{equation}
        \sum_{\bm{\mu},\bm{\nu}}\zeta_{\bm{\mu}}\zeta_{\bm{\nu}}\,\mathcal{U}\left(\bm{j}_\Lambda\oplus\sigma(\bm{j}_\Lambda),\bm{\bar{\mu}}\,\vert\,\Omega_{\bm{k}}(\bm{j}_\Lambda)\oplus\Omega_{\bm{l}}\circ\sigma(\bm{j}_\Lambda),\bm{\bar{\nu}}\right).
        \label{eq:polynomial_to_integrate}
    \end{equation}
    This task can be tackled by using Weingarten Calculus~\cite{collins2022weingarten, collins2003moments}.

    We will use two key results concerning the Weingarten Calculus for the unitary group. The statements of these theorems, adapted to the notation we used throughout the article, can be found in Appendix~\ref{app:weingarten}. The first of these results can be found in Lemma 3 of Ref.~\cite{matsumoto2012general}, and allows us to write 
    \begin{align}
        \begin{split}
            &\mathbb{E}_{\bm{U}}\left[\,\mathcal{U}\left(\bm{j}_\Lambda\oplus\sigma(\bm{j}_\Lambda),\bm{\bar{\mu}}\,\vert\,\Omega_{\bm{k}}(\bm{j}_\Lambda)\oplus\Omega_{\bm{l}}\circ\sigma(\bm{j}_\Lambda),\bm{\bar{\nu}}\right)\right]=\\
            &\quad\sum_{\varrho,\tau\in S_{4N}}\Delta\left[\Omega_{\bm{k}}(\bm{j}_\Lambda)\oplus\Omega_{\bm{l}}\circ\sigma(\bm{j}_\Lambda)\,\vert\,\varrho\left(\bm{j}_\Lambda\oplus\sigma(\bm{j}_\Lambda)\right)\right]\\
            &\qquad\qquad\qquad\times\Delta\left[\bm{\bar{\nu}}\,\vert\,\tau(\bm{\bar{\mu}})\right]\mathrm{Wg}_{4N}(\varrho^{-1}\circ\tau;M),
        \end{split}
        \label{eq:average_monomial}
    \end{align}
    where $\Delta[\bm{g}\,\vert\,\bm{h}]=\prod_{a=1}^{m}\delta_{g_a,h_a}$ and $\mathrm{Wg}_m(\sigma;M)$ stands for the \textit{Weingarten function} for the unitary group $U(M)$~\cite{collins2022weingarten, collins2003moments}.     
    
    Combining Eqs.~\eqref{eq:polynomial_to_integrate} and~\eqref{eq:average_monomial} we obtain
    \begin{align}
        \begin{split}
            \sum_{\bm{\mu},\bm{\nu}}\zeta_{\bm{\mu}}\zeta_{\bm{\nu}}\,&\mathbb{E}_{\bm{U}}\left[\,\mathcal{U}\left(\bm{j}_\Lambda\oplus\sigma(\bm{j}_\Lambda),\bm{\bar{\mu}}\,\vert\,\Omega_{\bm{k}}(\bm{j}_\Lambda)\oplus\Omega_{\bm{l}}\circ\sigma(\bm{j}_\Lambda),\bm{\bar{\nu}}\right)\right]\\
            &=\sum_{\varrho \in \bar{S}_\Lambda}\sum_{\tau\in S_{4N}}\mathrm{Wg}_{4N}(\varrho^{-1}\circ\tau;M)f(\bm{\zeta},\tau)
            \, ,
        \end{split}
        \label{eq:polynomial_to_integrate_1}
    \end{align}
    where, $\bar{S}_\Lambda\subseteq S_{4N}$ depends on $\Lambda$, $\sigma$, $\bm{k}$ and $\bm{l}$, and is defined as
    \begin{equation}
        \bar{S}_\Lambda = \left\{\varrho\in S_{4N}\,\vert\,\varrho\left(\bm{j}_\Lambda\oplus\sigma(\bm{j}_\Lambda)\right)=\Omega_{\bm{k}}(\bm{j}_\Lambda)\oplus\Omega_{\bm{l}}\circ\sigma(\bm{j}_\Lambda)\right\}.
        \label{eq:set_bar_S}
    \end{equation}
    Let us note that $\bar{S}_\Lambda$ is non-empty, since $\Omega_{\bm{k}}\oplus\Omega_{\bm{l}} \in \bar{S}_\Lambda$. On the other hand,
    \begin{equation}
        f(\bm{\zeta},\tau)=\sum_{\bm{\mu}}\sum_{\substack{\bm{\nu} \text{ s.t. }\\\bm{\bar{\nu}}=\tau(\bm{\bar{\mu}})}}\!\!\zeta_{\bm{\mu}}\,\zeta_{\bm{\nu}}.
        \label{eq:f_function}
    \end{equation}
    
    Notice that $\mathrm{Wg}_{4N}(\varrho^{-1}\circ\tau;M)$ does not depend on the specific values of the indices in $\bm{j}_\Lambda$. Rather, through its dependence on the permutations $\varrho\in \bar{S}_\Lambda$, it is determined by the \textit{structure} of $\bm{j}_\Lambda$, i.e., by the number of different $\{j_\lambda\}$ and their positions within the sequence.

    Since we are interested in the expected value of the linear cross-entropy as $M\rightarrow\infty$, we can focus on the asymptotic behavior of the Weingarten function, which is the subject of the second key result that we will use in this section. According to Corollary 2.7 in Ref.~\cite{collins2006integration}, as $M\rightarrow \infty$ we can write
    \begin{align}
        \begin{split}
            \!\!\mathrm{Wg}_{4N}(\varrho^{-1}\circ\tau;M)&=\mathrm{Moeb}(\varrho^{-1}\circ\tau)\frac{1}{M^{4N+\|\varrho^{-1}\circ\tau\|}}\\
            &\quad+\mathcal{O}\left(M^{-4N-\|\varrho^{-1}\circ\tau\|-2}\right),
        \end{split}
        \label{eq:asymptotic_weingarten}
    \end{align}
    where $\mathrm{Moeb}(\sigma)$ is the Möbius function (see Appendix~\ref{app:weingarten} for its definition) and $\|\sigma\|$ denotes the minimum number of transpositions in which we can write $\sigma$.

    Since $\|\varrho^{-1}\circ\tau\|\geq 0$, we can recognize
    that the leading order term in the asymptotic expansion in Eq.~\eqref{eq:asymptotic_weingarten} decays at least as fast as $M^{-4N}$. Taking into account that when $\tau=\varrho$, $\|\varrho^{-1}\circ\tau\|=\|e_{4N}\|=0$, and $\mathrm{Moeb}(\varrho^{-1}\circ\tau)=\mathrm{Moeb}(e_{4N})=1$, with $e_{4N}$ the identity permutation in $S_{4N}$, we can write the asymptotic form of Eq.~\eqref{eq:polynomial_to_integrate_1} as 
    \begin{align}
        \begin{split}
            \sum_{\bm{\mu},\bm{\nu}}\zeta_{\bm{\mu}}\zeta_{\bm{\nu}}\,&\mathbb{E}_{\bm{U}}\left[\,\mathcal{U}\left(\bm{j}_\Lambda\oplus\sigma(\bm{j}_\Lambda),\bm{\bar{\mu}}\,\vert\,\Omega_{\bm{k}}(\bm{j}_\Lambda)\oplus\Omega_{\bm{l}}\circ\sigma(\bm{j}_\Lambda),\bm{\bar{\nu}}\right)\right]\\
            &=\sum_{\varrho \in \bar{S}_\Lambda}f(\bm{\zeta},\varrho)M^{-4N}+\mathcal{O}\left(M^{-4N-1}\right).
        \end{split}
        \label{eq:polynomial_to_integrate_2}
    \end{align}

    Bringing together the results in Eqs.~\eqref{eq:polynomial_to_integrate_2} and~\eqref{eq:lxe_sqz_after_integral} we obtain the following expression for the average value of $\mathrm{LXE}\left(\bm{A}_{\mathrm{sqz}},\bm{A}_{\mathrm{sqz}};2N\right)$:
    \begin{align}
        \begin{split}
            &\mathbb{E}_{\bm{U}}\left[\mathrm{LXE}\left(\bm{A}_{\mathrm{sqz}},\bm{A}_{\mathrm{sqz}};2N\right)\right]=\binom{\frac{R}{2}+N-1}{N}^{-2}\\
            &\times\sum_{\bm{k}^{(c)}, \bm{l}^{(c)}}\,\prod_{a=1}^N\frac{1}{k_a!l_a!(2a)^{k_a+l_a}}\sum_{\Lambda\in \mathcal{Q}[\bm{j}]}\sum_{\sigma\in S_{2N}}\frac{1}{\Lambda!}\\
            &\times\sum_{\mathrm{diff.}\{j_\lambda\}}\left(\sum_{\varrho \in \bar{S}_\Lambda}f(\bm{\zeta},\varrho)M^{-4N}+\mathcal{O}\left(M^{-4N-1}\right)\right),
        \end{split}
        \label{eq:lxe_sqz_after_average}
    \end{align}
    where we were able to move the sum over different $\{j_\lambda\}$ past the sum over $\sigma \in S_{2N}$ and the term $1/\Lambda!$ because neither of them depends on the specific values that the $\{j_\lambda\}$ take.
    
    In fact, there are no longer any terms in Eq.~\eqref{eq:lxe_sqz_after_average} that depend on the specific values of these indices. Indeed, as we argued before, the sum over permutations in $\bar{S}_\Lambda$, as well as the terms of order $\mathcal{O}(M^{-4N-1})$, will depend only on the structure of $\bm{j}_\Lambda$, which is determined by $\Lambda$. Consequently, we can make the replacement
    \begin{equation}
        \sum_{\text{diff.}\,\{r_\lambda\}}\rightarrow \frac{M!}{(M-N_\Lambda)!} = M^{N_\Lambda}+\mathcal{O}(M^{N_\Lambda-1}),
        \label{eq:sum_different_representatives}
    \end{equation}
    where $N_\Lambda$ is the number of different $\{j_\lambda\}$, i.e., the number of blocks in $\Lambda$. This allows us to write   
    \begin{align}
        \begin{split}
            \sum_{\mathrm{diff.}\{j_\lambda\}}&\left(\sum_{\varrho \in \bar{S}_\Lambda}f(\bm{\zeta},\varrho)M^{-4N}+\mathcal{O}\left(M^{-4N-1}\right)\right)=\\
            &\sum_{\varrho \in \bar{S}_\Lambda}f(\bm{\zeta},\varrho)M^{-4N+N_\Lambda}+\mathcal{O}\left(M^{-4N+N_\Lambda-1}\right).
        \end{split}
        \label{eq:asymptotic_sum_different_indices}
    \end{align}

    Given that $1\leq N_\Lambda \leq 2N$, with $N_\Lambda = 2N$ for $\Lambda=\{\{j_1\},\dots,\{j_{2N}\}\}$, we can see that the leading term in the asymptotic expansion of the average LXE decays as $M^{-2N}$. On this account, and keeping in mind that $\Lambda!=1$ for the leading term, we may write Eq.~\eqref{eq:lxe_sqz_after_average} as
    \begin{align}
        \begin{split}
            &\mathbb{E}_{\bm{U}}\left[\mathrm{LXE}\left(\bm{A}_{\mathrm{sqz}},\bm{A}_{\mathrm{sqz}};2N\right)\right]=\\
            &\binom{\frac{R}{2}+N-1}{N}^{-2}
            \sum_{\bm{k}^{(c)}, \bm{l}^{(c)}}\,\prod_{a=1}^N\frac{1}{k_a!l_a!(2a)^{k_a+l_a}}\\
            &\times\left(\sum_{\sigma\in S_{2N}}\sum_{\varrho \in \bar{S}}f(\bm{\zeta},\varrho)M^{-2N}+\mathcal{O}\left(M^{-2N-1}\right)\right),
        \end{split}
        \label{eq:average_lxe_sqz_asymptotic}
    \end{align}
    where
    \begin{equation}
        \bar{S} = \left\{\varrho\in S_{4N}\,\vert\,\varrho\left(\bm{j}\oplus\sigma(\bm{j})\right)=\Omega_{\bm{k}}(\bm{j})\oplus\Omega_{\bm{l}}\circ\sigma(\bm{j})\right\}.
        \label{eq:set_bar_S_final}
    \end{equation}

\subsection{Final expression of the ideal LXE score}\label{sec:final_form_lxe_score}

    Recall that in the definition of the LXE score given in Eq.~\eqref{eq:lxe_score_definition} there is one extra term that depends on the number of modes $M$, namely, the normalization factor
    \[\binom{M+2N-1}{2N}=\frac{\Gamma(M+2N)}{\Gamma(2N+1)\Gamma(M)}.\]
    For a fixed $2N$, $\Gamma(M+2N)/\Gamma(M) \sim M^{2N}$ as $M\rightarrow \infty$, so we can write the following asymptotic expression for the binomial coefficient:
    \begin{equation}
        \binom{M+2N-1}{2N}\sim \frac{M^{2N}}{(2N)!}.
        \label{eq:asymptotic_binomial_coeff}
    \end{equation}

    On these grounds, and noticing that 
    \[\binom{\frac{R}{2}+N-1}{N}=\frac{1}{2^NN!}\frac{(R+2N-2)!!}{(R-2)!!},\]
    we are now able to write
    \begin{align}
        \begin{split}
            &s(\bm{A}_{\mathrm{sqz}};2N)=\lim_{M\rightarrow\infty} \binom{M+2N-1}{2N}\\
            &\qquad\qquad\qquad\qquad\times\mathbb{E}_{\bm{U}}\left[\mathrm{LXE}\left(\bm{A}_{\mathrm{sqz}},\bm{A}_{\mathrm{sqz}};2N\right)\right]\\
            &=\frac{4^{N}(N!)^2}{(2N)!}\left[\frac{(R-2)!!}{(R+2N-2)!!}\right]^2
            \\
            &\quad\times\sum_{\bm{k}^{(c)}, \bm{l}^{(c)}}\,\prod_{a=1}^N\frac{1}{k_a!l_a!(2a)^{k_a+l_a}}\sum_{\sigma\in S_{2N}}\sum_{\varrho \in \bar{S}}f(\bm{\zeta},\varrho).
        \end{split}
        \label{eq:lxe_sqz_score_gen}
    \end{align}

    At this point we may take into account that the entries of the diagonal matrix $\bm{\zeta}$ satisfy $\zeta_a=1$ for $1\leq a\leq R$, and $\zeta_a = 0$ otherwise. This allows us to see that
    \begin{equation}
        f(\bm{\zeta},\varrho)=\sum_{\bm{\mu}\in [R]^{2N}} \!\!\!\sum_{\substack{\bm{\nu}\in [R]^{2N}\\ \text{s.t }\bm{\bar{\nu}}=\varrho(\bm{\bar{\mu}})}}\!\!\!1,
        \label{eq:series_in_zeta}
    \end{equation}
    where the symbol $\bm{\mu}\in[R]^{2N}$ indicates that each $\mu_a$, $a=1,\dots, 2N$, takes values in the set $[R]=\{1,\dots,R\}$. According to Lemma 6 in Ref.~\cite{matsumoto2012general}, the sum in the last expression satisfies the relation
    \begin{equation}
        \sum_{\bm{\mu}\in [R]^{2N}} \!\!\!\sum_{\substack{\bm{\nu}\in [R]^{2N}\\ \text{s.t }\bm{\bar{\nu}}=\varrho(\bm{\bar{\mu}})}}\!\!\!1=R^{\,\ell(\varrho)},
        \label{eq:coset_type_sum}
    \end{equation}
    where $\ell(\varrho)$ is the \textit{length of the coset-type} of $\varrho$~\cite{matsumoto2012general, collins2022weingarten}. 

    In order to understand the meaning of $\ell(\varrho)$, let us state the definition of the \textit{coset-type} of a permutation given in Ref.~\cite{collins2022weingarten} (see Fig.~\ref{fig:definition_of_coset_type} for an example). Let $\sigma\in S_{2m}$. We can assign to this permutation an \textit{undirected graph} $\bm{\Gamma}(\sigma)$, whose vertices are $1,2,\dots,2m$, and whose edges are defined by  
    $\{(2k-1,2k)\,\vert\,k\in\{1,\dots,m\}\}$ and $\{(\sigma(2k-1),\sigma(2k))\,\vert\,k\in\{1,\dots,m\}\}$. Note that there are a total of $2m$ edges, and each vertex lies in exactly two edges. This implies that the \textit{connected components} of the graph have an even number of edges~\cite{collins2022weingarten}. Call the lengths of such connected components $2\eta_1,2\eta_2\dots,2\eta_l$ and arrange them so that $\eta_1\geq\eta_2\geq\cdots\geq\eta_l\geq 1$. Then, $\eta(\sigma)=(\eta_1,\eta_2,\dots,\eta_l)$ is an \textit{integer partition} of $m$, and is called the coset-type of $\sigma$. The length of the coset-type of $\sigma$ will then be the length of the partition $\eta(\sigma)$, or equivalently, the number of connected components in the graph $\bm{\Gamma}(\sigma)$.

    \begin{figure}[!t]
        {
          \includegraphics[width=0.8\columnwidth]{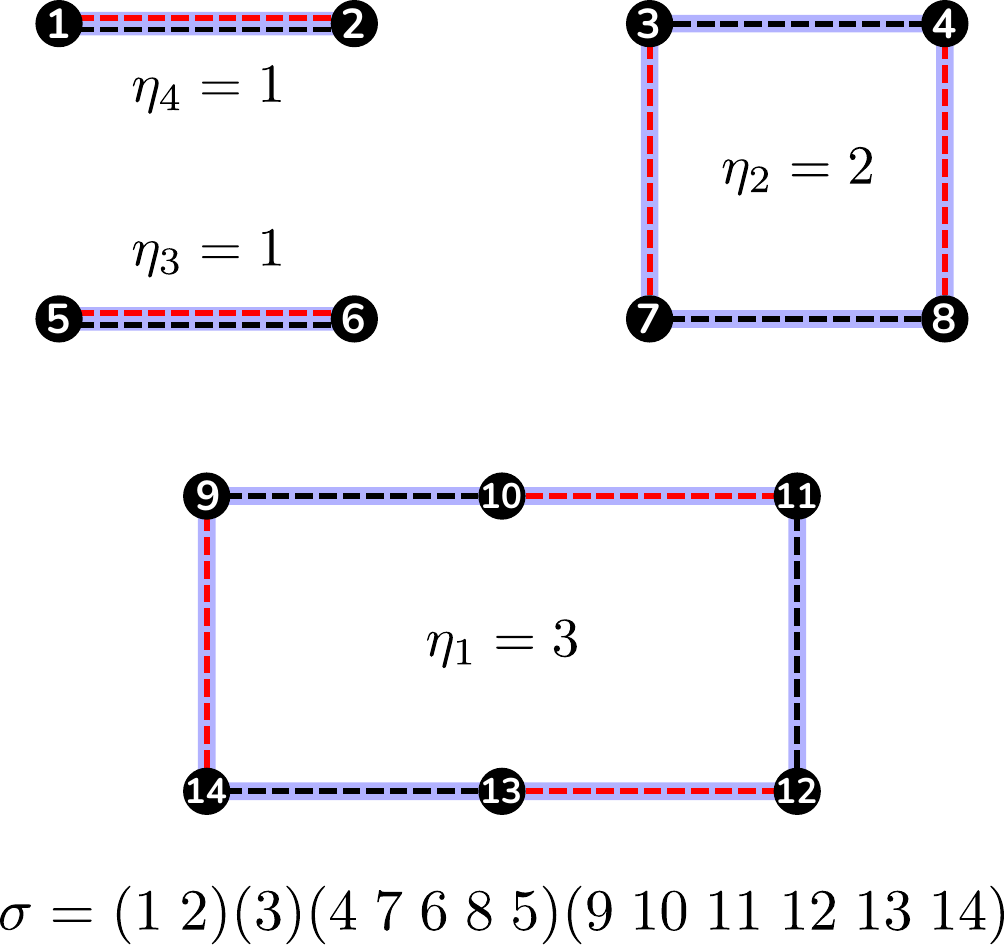}%
        }
        \caption{Illustration of the definition of the coset-type of $\sigma$ for $\sigma = (1\,2)(3)(4\,7\,6\,8\,5)(9\,10\,11\,12\,13\,14)$. Notice that $\sigma$ is an element of $S_{2m}$ with $m=7$.
        The vertices of the undirected graph $\bm{\Gamma}(\sigma)$ are represented by numbered black circles. The edges of the form $(2k-1,2k)$, that correspond to the set $\{(1,2), (3,4), (5,6), (7,8), (9,10), (11,12), (12,14)\}$, are drawn with black dashed lines. The edges of the form $(\sigma(2k-1),\sigma(2k))$ are obtained after the application of $\sigma$ over the pairs $(2k,2k-1)$ and correspond to the set $\{(2,1),(3,7),(4,8),(6,5),(10,11),(12,13),(14,9)\}$. These edges are shown as red dashed lines. 
        Each connected component of the graph is highlighted with a light blue, thick line. As can be seen, there are a total of four connected components, all of them cycles, with lengths $2\eta_b$, $b=1,2,3,4$. This implies that the coset-type of $\sigma$ is $\eta(\sigma)=(3,2,1,1)$, and its length is $\ell(\sigma)=4$. We can readily check that $\eta(\sigma)$ is an integer partition of $m=7$. }
        \label{fig:definition_of_coset_type}
    \end{figure}
    
    According to the previous definition, we can see that $1\leq \ell(\varrho) \leq 2N$, which in turn allows us to write
    \begin{equation}
        \sum_{\varrho \in \bar{S}}f(\bm{\zeta},\varrho)=\sum_{\varrho \in \bar{S}}R^{\,\ell(\varrho)}=\sum_{\ell = 1}^{2N} b_\ell R^\ell,
        \label{eq:coset_type_sum_2}
    \end{equation}
    where 
    \begin{align} b_\ell \equiv b_\ell\left(\bm{j}\oplus\sigma(\bm{j}),\Omega_{\bm{k}}(\bm{j})\oplus\Omega_{\bm{l}}\circ\sigma(\bm{j})\right)\label{eq:definebell}
    \end{align} is the number of permutations $\varrho\in S_{4N}$ that take $\bm{j}\oplus\sigma(\bm{j})$ into $\Omega_{\bm{k}}(\bm{j})\oplus\Omega_{\bm{l}}\circ\sigma(\bm{j})$, and whose coset-type has length $\ell$ (or, equivalently, whose associated graphs $\bm{\Gamma}(\varrho)$ have $\ell$ connected components).
    
    Combining Eqs.~\eqref{eq:coset_type_sum_2} and~\eqref{eq:lxe_sqz_score_gen}, and defining (see Eq.~\eqref{eq:coefficients_sqz_lxe_score})
    \begin{align}
        \begin{split}
            c_\ell = \!\!\!\sum_{\bm{k}^{(c)}, \bm{l}^{(c)}}&\,\prod_{a=1}^N\frac{1}{k_a!l_a!(2a)^{k_a+l_a}}\\
            &\!\times\!\!\!\sum_{\sigma\in S_{2N}}b_\ell\left(\bm{j}\oplus\sigma(\bm{j}),\Omega_{\bm{k}}(\bm{j})\oplus\Omega_{\bm{l}}\circ\sigma(\bm{j})\right),
        \end{split}
    \end{align}
    we finally obtain the result stated in Eq.~\eqref{eq:main_result_lxe_sqz}, at the very beginning of this section:  
    \begin{equation*}
        s(\bm{A}_\mathrm{sqz};2N)=\frac{4^{N}(N!)^2}{(2N)!}\left[\frac{(R-2)!!}{(R+2N-2)!!}\right]^2\sum_{\ell=1}^{2N}c_\ell R^{\ell}.
    \end{equation*}

    \section{\label{sec:different_squeezing} LXE score for the ideal squeezed state model with different squeezing parameters}
    
    We can readily generalize the methods presented in the last section to determine the ideal score for GBS setups that use input squeezed states with different squeezing parameters.
    
    Suppose that the first $R$ modes of the interferometer receive single-mode squeezed states with squeezing parameters $\{r_k\}$, $k\in\{1,\dots, R\}$, while the remaining $M-R$ modes receive the vacuum state. The matrix $\bm{A}_{\mathrm{sqz}}'$ describing this setup can be written as (see Appendix~\ref{app:misc})
    \begin{equation}
        \bm{A}_{\mathrm{sqz}}'=\bm{V}\oplus\bm{V}^*,
        \label{eq:sqz_model_general}
    \end{equation}
    where $\bm{V}=\bm{U}\bm{\zeta}'\bm{U}^{\mathrm{T}}$, with $\bm{\zeta}'=\tanh(\bm{r})\oplus\bm{0}_{M-R}$ and $\tanh(\bm{r})=\mathrm{diag}[\tanh(r_1),\dots,\tanh(r_R)]$.

    Models $\bm{A}_{\mathrm{sqz}}'$ and $\bm{A}_{\mathrm{sqz}}$ (which was defined in Eq.~\eqref{eq:sqz_model_definition}) differ only in the definition of the diagonal matrix $\bm{\zeta}'$. This suggests that we can compute the ideal score associated to $\bm{A}_{\mathrm{sqz}}'$ by making the replacement 
    \[\tanh(r)\bm{\zeta} \longrightarrow \bm{\zeta}'=\tanh(\bm{r})\oplus\bm{0}_{M-R},\] 
    while keeping all the procedures shown in Sec.~\ref{sec:lxe_score_sqz} mostly the same. Following this strategy, we can see that the change from $\bm{A}_{\mathrm{sqz}}$ to $\bm{A}_{\mathrm{sqz}}'$ will become manifest in two terms: the factor $\bar{\mathcal{D}}(\bm{A}_{\mathrm{sqz}}', \bm{A}_{\mathrm{sqz}}', 2N)$ defined
    in Eq.~\eqref{eq:calDdefine}, and the function $f(\bm{\zeta}',\tau)$ defined in Eq.~\eqref{eq:f_function}.

    To compute $\bar{\mathcal{D}}(\bm{A}_{\mathrm{sqz}}', \bm{A}_{\mathrm{sqz}}', 2N)$, we note that, according to Eqs.~\eqref{eq:derivative_q_function} and~\eqref{eq:der_q_function_sqz},
    \begin{equation}
        \frac{1}{(2N)!}\left.\frac{\partial^{2N}q(\alpha,\bm{\phi},\bm{A}_{\mathrm{sqz}}')}{\partial\alpha^{2N}}\right|_{\alpha=0}\!\!\!=\sum_{\bm{k}^{(c)}}\,\prod_{a=1}^N\frac{1}{k_a!a^{k_a}}\prod_{a=1}^Nu_a^{k_a},
        \label{eq:der_q_function_sqz_ds}
    \end{equation}
    where $u_a=\frac{1}{2}\mathrm{tr}\left[(\bm{D}^2\bm{V}\bm{D}^2\bm{V}^*)^a\right]$. When $\bm{\phi}=\bm{0}$, we have $\bm{D}=\mathbb{I}_M$, and so
    \begin{equation*}
        u_a=\frac{1}{2}\mathrm{tr}\left[(\bm{V}\bm{V}^*)^a\right]=\frac{1}{2}\mathrm{tr}\left[(\bm{\zeta}')^{2a}\right]=\frac{1}{2}\sum_{k=1}^R\tanh^{2a}(r_k).
    \end{equation*}
    By defining 
    \begin{align}
        \varepsilon_{a}=\frac{1}{R}\sum_{k=1}^R\tanh^{a}(r_k),
        \label{eq:moment_tanh_sqz}
    \end{align}
    where we note that $0<|\varepsilon_a|<1$ for all $a$ and $R$, we can write $u_a = \frac{1}{2}\varepsilon_{2a}R$, and
    \begin{align}
        \begin{split}
            \frac{1}{(2N)!}\left.\frac{\partial^{2N}q(\alpha,\bm{0},\bm{A}_{\mathrm{sqz}}')}{\partial\alpha^{2N}}\right|_{\alpha=0}
            &=\sum_{\bm{k}^{(c)}}\,\prod_{a=1}^N\frac{\varepsilon_{2a}^{k_a}R^{k_a}}{k_a!(2a)^{k_a}}\\
            &=\sum_{\ell=1}^Nd'_\ell R^\ell.
        \end{split}
        \label{eq:der_q_function_sqz_ds_2}
    \end{align}
    The last equality above is obtained by noticing
    that the non-negative integer components of
    $\bm{k}=(k_1,\dots,k_N)$ satisfy the
    constraint $k_1+2k_2+\cdots+Nk_N=N$. For each $\bm{k}$, the sum $\sum_{\bm{k}^{(c)}}$ will include a term proportional to $R^{k_1+\dots+k_N}$, which can be at least $R$ (for $k_N=1$ and $k_a = 0\;\forall a\neq N$) and at most $R^N$ (for $k_1=N$ and $k_a = 0\;\forall a\neq 1$).

    From Eq.~\eqref{eq:der_q_function_sqz_ds_2}, 
    and using Eq.~\eqref{eq:coefficient_lxe_2},
    we may write
    \begin{equation}
        \bar{\mathcal{D}}(\bm{A}_{\mathrm{sqz}}', \bm{A}_{\mathrm{sqz}}', 2N)=\left(\sum_{\ell=1}^Nd'_\ell R^\ell\right)^{-2}.
        \label{eq:coeff_lxe_sqz_ds}
    \end{equation}
    
    This expression leads to the following modified form of Eq.~\eqref{eq:lxe_sqz_score_gen}:
    \begin{align}
        \begin{split}
            &s(\bm{A}_{\mathrm{sqz}}';2N)=\frac{1}{(2N)!}\left(\sum_{\ell=1}^Nd'_\ell R^\ell\right)^{-2}
            \\
            &\quad\times\sum_{\bm{k}^{(c)}, \bm{l}^{(c)}}\,\prod_{a=1}^N\frac{1}{k_a!l_a!(2a)^{k_a+l_a}}\sum_{\sigma\in S_{2N}}\sum_{\varrho \in \bar{S}}f(\bm{\zeta}',\varrho),
        \end{split}
        \label{eq:lxe_sqz_score_gen_ds}
    \end{align}
    where $f(\bm{\zeta}',\varrho)$, is now computed as 
    \begin{align}
        \begin{split}
            f(\bm{\zeta}',\varrho)&=\sum_{\bm{\mu}}\sum_{\substack{\bm{\nu} \text{ s.t. }\\\bm{\bar{\nu}}=\varrho(\bm{\bar{\mu}})}}\!\!\zeta'_{\bm{\mu}}\,\zeta'_{\bm{\nu}}\\
            &=\sum_{\bm{\mu}\in [R]^{2N}} \!\!\!\sum_{\substack{\bm{\nu}\in [R]^{2N}\\ \text{s.t }\bm{\bar{\nu}}=\varrho(\bm{\bar{\mu}})}}\!\prod_{a=1}^{2N}\tanh(r_{\mu_a})\tanh(r_{\nu_a}).
        \end{split}
        \label{eq:f_function_ds}
    \end{align}

    In Appendix~\ref{app:general_sqz_score}, we prove that Eq.~\eqref{eq:f_function_ds} reduces to 
    \begin{equation}
        f(\bm{\zeta}',\varrho)=\left(\prod_{b\in\eta(\varrho)}\varepsilon_{2b}\right)R^{\ell(\varrho)},
        \label{eq:f_function_ds_2}
    \end{equation}
    where, let us remind the reader, $\eta(\varrho)$ stands for the coset-type of $\varrho$, and $\ell(\varrho)$ is the length of $\eta(\varrho)$ (or, equivalently, the number of connected components in the undirected graph $\bm{\Gamma}(\varrho)$). 

    If we recall that $1\leq\ell(\varrho)\leq 2N$, we can define the subset $\bar{S}_{\ell}\subset S_{4N}$ as
    \begin{align}
        \begin{split}
            \bar{S}_{\ell}=&\{\varrho\in S_{4N}\,\vert\,\varrho\left(\bm{j}\oplus\sigma(\bm{j})\right)=\Omega_{\bm{k}}(\bm{j})\oplus\Omega_{\bm{l}}\circ\sigma(\bm{j})\\
            &\quad\text{ and }\bm{\Gamma}(\varrho)\text{ has }\ell\text{ connected components}\},
        \end{split}
        \label{eq:s_tilde_l_subset}
    \end{align}
    and write
    \begin{equation}
        \sum_{\varrho\in \bar{S}}f(\bm{\zeta}',\varrho)=\sum_{\ell=1}^{2N}\left(\sum_{\varrho\in \bar{S}_{\ell}}\prod_{b\in\eta(\varrho)}\varepsilon_{2b}\right)R^{\ell},
        \label{eq:reorganizing_sums_ds}
    \end{equation}
    thus verifying that $\sum_{\varrho\in \bar{S}}f(\bm{\zeta}',\varrho)$ remains a polynomial of degree $2N$ in $R$.

    Combining Eqs.~\eqref{eq:lxe_sqz_score_gen_ds} and~\eqref{eq:reorganizing_sums_ds} we obtain the final expression for the ideal LXE score for the model $\bm{A}_{\mathrm{sqz}}'$:
    \begin{align}
        \begin{split}
            &s(\bm{A}_{\mathrm{sqz}}';2N)=\frac{1}{(2N)!}\left(\sum_{\ell=1}^Nd'_\ell R^\ell\right)^{-2}\sum_{\ell=1}^{2N}c'_{\ell}R^{\ell},
        \end{split}
        \label{eq:lxe_sqz_score_different_squeezing}
    \end{align}
    where 
    \begin{equation}
        c'_\ell =\!\!\! \sum_{\bm{k}^{(c)}, \bm{l}^{(c)}}\,\prod_{a=1}^N\frac{1}{k_a!l_a!(2a)^{k_a+l_a}}\!\!\sum_{\sigma\in S_{2N}}\sum_{\varrho\in \bar{S}_{\ell}}\prod_{b\in\eta(\varrho)}\varepsilon_{2b}.
        \label{eq:c_tilde_coeffs}
    \end{equation}

    Again, in Appendix~\ref{app:general_sqz_score}, we
    prove that if $r_k=r$ for all $k\in\{1,\dots, R\}$, $s(\bm{A}_{\mathrm{sqz}}';2N)$ becomes 
    independent of $r$ and we recover
    $s(\bm{A}_{\mathrm{sqz}};2N)$ identically,
    showing the consistency between
    Eqs.~\eqref{eq:main_result_lxe_sqz} and~\eqref{eq:lxe_sqz_score_different_squeezing}.
    
    The main difference between the scores
    $s(\bm{A}_{\mathrm{sqz}};2N)$ and
    $s(\bm{A}_{\mathrm{sqz}}';2N)$ is that the
    latter shows an explicit dependence on the
    input squeezing parameters. In order to properly account for this dependence, we need not only determine the number of connected components in $\bm{\Gamma}(\varrho)$, but also the lengths of all connected components in the graph.
    
    Fig.~\ref{fig:lxe_ideal_b}, shows the scores $s(\bm{A}_{\mathrm{sqz}};2N)$ and $s(\bm{A}'_{\mathrm{sqz}};2N)$ as functions of $2N$ for $R\in\{10, 20, 50, 150\}$, and for different choices of the input squeezing parameters. Brute force computation of the coefficients $\{c_\ell\}$ and $\{c_\ell'\}$ for values of $2N>8$ proves to be difficult due to the increasing number of elements in the set $\bar{S}$. In view of this, we only show the exact computation of the scores for $2\leq2N\leq8$. For each value of $R$, we randomly selected four sets of input squeezing parameters, making sure that every set had a different mean number of photons $\bar{n}=\frac{1}{R}\sum_{k=1}^{R}\sinh^2(r_k)$. For $R=10$, we obtained $\bar{n}\in \{0.4,0.8,2.2,6.2\}$; for $R=20$, $\bar{n}\in \{0.4,1.1,2.3,5.2\}$; for $R=50$, $\bar{n}\in \{0.4,1.0,2.4,5.4\}$; and for $R=150$, $\bar{n}\in \{0.4,1.1,2.5,5.6\}$. We computed the coefficients $\{c_\ell\}$ and $\{c_\ell'\}$ using the methods in the libraries 
    \texttt{SymPy}~\cite{meurer2017sympy} and \texttt{graph-theory}~\cite{madsen2024graph}.
    
    As can be seen, when we increase the value of $\bar{n}$, $s(\bm{A}_{\mathrm{sqz}}';2N)$ approaches the score that would be obtained by a model that uses input states with the same squeezing parameter. Interestingly, $s(\bm{A}_{\mathrm{sqz}}';2N)$ appears to approach $s(\bm{A}_{\mathrm{sqz}};2N)$ more rapidly with increasing $R$. This suggests that, in the limit of large $R$, $s(\bm{A}_{\mathrm{sqz}}';2N)$ will be approximately equal to the score of an ideal model with uniform squeezing, no matter the choice of the input squeezing parameters. This behavior nicely connects with the results of the next section.

    \begin{figure*}[t!]
        \centering
        \subfloat[]{\label{fig:lxe_ideal_a}\includegraphics[width=0.49\textwidth]{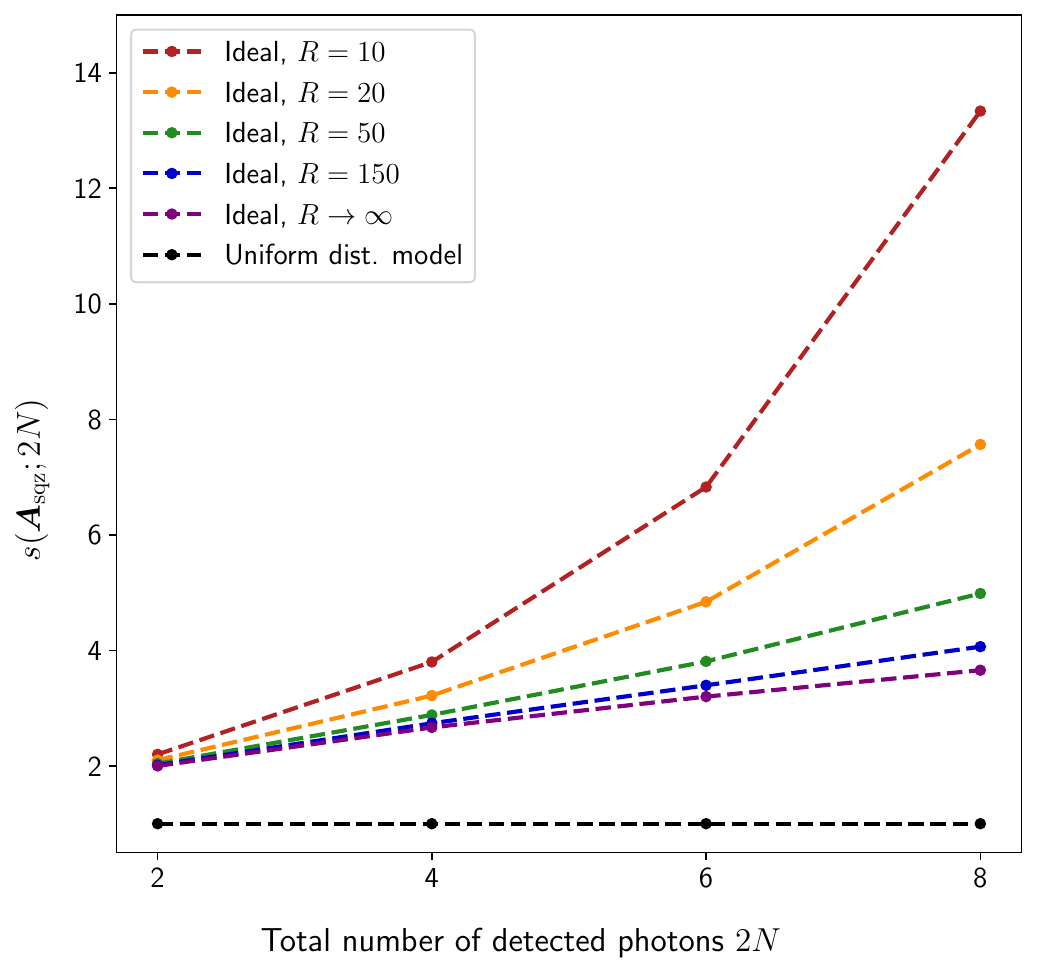}}\hfill
        \subfloat[]{\label{fig:lxe_ideal_b}\includegraphics[width=0.49\textwidth]{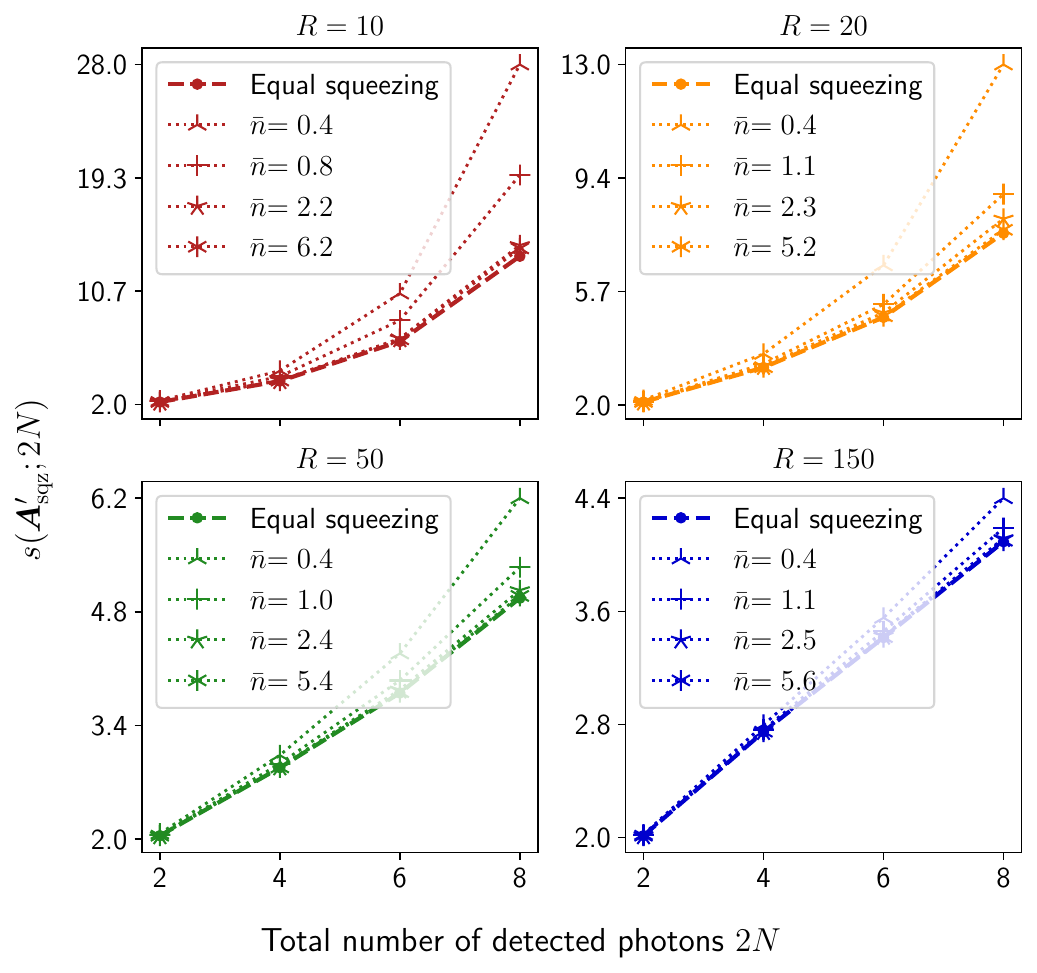}}
        \caption{(a) Ideal LXE score as a function of the total number of detected photons, $2N$. The purple dashed line shows the score for an ideal squeezed state model with no vacuum input modes (as 
        indicated by the label $R\rightarrow \infty$). The black dashed line shows the score for the model $\bm{A}_{\mathrm{uni}}$, which leads to a uniform distribution over the sample space. The score for this model is identically 1 for every $N$. The remaining lines correspond to the ideal squeezed state model $\bm{A}_{\mathrm{sqz}}$, where the first $R$ modes receive identical single-mode squeezed states and the remaining $M-R$ modes receive the vacuum state. The red dashed line corresponds to $R=10$, the orange dashed line to $R=20$, the green dashed line to $R=50$, and the blue dashed line to $R=150$. (b) Ideal LXE score as a function of the total number of detected photons, $2N$, for GBS setups using input states with different squeezing parameters. The dotted lines were obtained by randomly selecting the squeezing parameters of the $R$ input squeezed states. For each value of $R$ there are four sets of squeezing parameters, each of them identified by their corresponding mean number of photons $\bar{n}=\frac{1}{R}\sum_{k=1}^{R}\sinh^2(r_k)$. The dashed lines show the scores for the case of identical squeezing in the input modes. The red lines correspond to $R=10$, the orange lines to $R=20$, the green lines to $R=50$, and the blue lines to $R=150$.}
        \label{fig:lxe_score_2N}
    \end{figure*}

    \section{\label{sec:no_vacuum}LXE score for the ideal squeezed state model without vacuum input modes}

    An interesting instance of the ideal squeezed state model is obtained when $R=M$, i.e., when all the input modes in the interferometer receive single-mode squeezed states. The computation of the ideal score for this type of setup can be done by taking the limit as $R\rightarrow \infty$ of Eqs.~\eqref{eq:main_result_lxe_sqz} and~\eqref{eq:lxe_sqz_score_different_squeezing}.

    Let us consider first the case of $s(\bm{A}_{\mathrm{sqz}};2N)$, given in Eq.~\eqref{eq:main_result_lxe_sqz}. For large $R$,
    \[\left[\frac{(R-2)!!}{(R+2N-2)!!}\right]^2\sim R^{-2N}.\]
    Thus, when taking the limit, the only non-vanishing contribution from the polynomial $\sum_{\ell}c_\ell R^\ell$ will be associated to the term $c_{2N}R^{2N}$. This allows us to write the ideal score as
    \begin{equation}
        s(\tilde{\bm{A}}_{\mathrm{sqz}};2N)=\lim_{R\rightarrow\infty}s(\bm{A}_{\mathrm{sqz}};2N) = 
        \frac{4^{N}(N!)^2}{(2N)!}\,c_{2N}.
        \label{eq:ideal_score_no_vacuum_ident}
    \end{equation}

    In the case of $s(\bm{A}_{\mathrm{sqz}}';2N)$, we have from Eq.~\eqref{eq:lxe_sqz_score_different_squeezing}
    \[\left(\sum_{\ell=1}^Nd'_\ell R^\ell\right)^{-2}\sim (d_{N}')^{-2}R^{-2N},\]
    when $R\rightarrow\infty$, which implies that the only non-vanishing contribution from the polynomial $\sum_{\ell}c_\ell' R^\ell$ will be associated to the term $c_{2N}'R^{2N}$. Then,
    \begin{equation}
        s(\bm{A}'_{\mathrm{sqz}};2N)\sim 
        \frac{1}{(2N)!}\,\frac{c'_{2N}}{(d'_{N})^2}.
        \label{eq:ideal_score_no_vacuum_diff}
    \end{equation}

    We can see in Eq.\eqref{eq:der_q_function_sqz_ds_2} that the only contribution to $d'_N$ comes from a vector $\bm{k}$ satisfying $k_1=N$ and $k_a = 0$ otherwise. This means that 
    \begin{equation}
        d_N'=\frac{\varepsilon_2^{k_1}}{k_1!2^{k_1}}=\frac{\varepsilon_2^N}{2^NN!}.
        \label{eq:coefficient_d_N'}
    \end{equation}

    Following Eq.~\eqref{eq:c_tilde_coeffs}, we see that $c_{2N}'$ is determined by computing the coset-type of the permutations $\varrho\in \bar{S}_{2N}$. By definition, every $\varrho\in \bar{S}_{2N}$ has an associated graph $\bm{\Gamma}(\varrho)$ with $2N$ connected components, which implies that every connected component in $\bm{\Gamma}(\varrho)$ has length two. This means that all $\varrho\in \bar{S}_{2N}$ have coset-type $\eta(\varrho)=(1,\dots,1)$, where 1 appears a total of $2N$ times. Consequently,
    \begin{equation}
        \sum_{\varrho\in \bar{S}_{2N}}\prod_{b\in\eta(\varrho)}\varepsilon_{2b}=\sum_{\varrho\in \bar{S}_{2N}}\prod_{a=1}^{2N}\varepsilon_{2}=\varepsilon_2^{2N}|\bar{S}_{2N}|.
        \label{eq:coefficient_c_N_cont}
    \end{equation}
    Recalling Eq.~\eqref{eq:s_tilde_l_subset} and the definition of the coefficients $b_\ell$ 
    given 
    in Eq.~\eqref{eq:definebell},
    we can verify that
    \begin{equation}
        |\bar{S}_{2N}| = b_{2N}\left(\bm{j}\oplus\sigma(\bm{j}),\Omega_{\bm{k}}(\bm{j})\oplus\Omega_{\bm{l}}\circ\sigma(\bm{j})\right),
        \label{eq:set_length_rel}
    \end{equation}
    and thus
    \begin{align}
        \begin{split}
            c'_{2N} &=\!\!\! \sum_{\bm{k}^{(c)}, \bm{l}^{(c)}}\,\prod_{a=1}^N\frac{1}{k_a!l_a!(2a)^{k_a+l_a}}\!\!\sum_{\sigma\in S_{2N}}\varepsilon_2^{2N}|\bar{S}_{2N}|\\
            &=\varepsilon_2^{2N}\!\!\! \sum_{\bm{k}^{(c)}, \bm{l}^{(c)}}\,\prod_{a=1}^N\frac{1}{k_a!l_a!(2a)^{k_a+l_a}}\!\!\sum_{\sigma\in S_{2N}}b_{2N}\\
            &=\varepsilon_2^{2N}c_{2N}.
        \end{split}
        \label{eq:c_coeffs_relation}
    \end{align}

    Combining Eqs.~\eqref{eq:ideal_score_no_vacuum_diff},~\eqref{eq:coefficient_d_N'} and~\eqref{eq:c_coeffs_relation}, we conclude that
    \begin{equation}
        s(\tilde{\bm{A}}'_{\mathrm{sqz}};2N)=\lim_{R\rightarrow\infty}s(\bm{A}'_{\mathrm{sqz}};2N) = 
        \frac{4^{N}(N!)^2}{(2N)!}\,c_{2N},
        \label{eq:ideal_score_no_vacuum_rel}
    \end{equation}
    i.e., the ideal score for a setup with no vacuum input modes will have the same value, whether we use input squeezed states with the same squeezing parameter or not.
    
    The permutations $\varrho \in S_{4N}$ whose associated undirected graphs $\bm{\Gamma}(\varrho)$ have $2N$ connected components constitute the \textit{hyperoctahedral group} of degree $2N$,
    $H_{2N}$~\cite{matsumoto2012general} (see Definition~\ref{def:hyperoctahedral_definition}). We may therefore say that the coefficient $b_{2N}$ is the number of permutations in $H_{2N}$ that take $\bm{j}\oplus\sigma(\bm{j})$ into $\Omega_{\bm{k}}(\bm{j})\oplus\Omega_{\bm{l}}\circ\sigma(\bm{j})$ for given $\bm{k}$, $\bm{l}$ and $\sigma$.

    In Appendix~\ref{app:non_vacuum_proofs} we prove that 
    \begin{equation}
        c_{2N}=\sum_{\bm{k}^{(c)}, \bm{l}^{(c)}}\,\prod_{a=1}^N\frac{1}{k_a!l_a!(2a)^{k_a+l_a}}\sum_{\sigma\in S_{2N}}b_{2N}=1
        \label{eq:counting_part_res}
    \end{equation}
    for every $N$, which implies that
    \begin{equation}
        s(\tilde{\bm{A}}'_{\mathrm{sqz}};2N)=s(\tilde{\bm{A}}_{\mathrm{sqz}};2N) = 
        \frac{4^{N}(N!)^2}{(2N)!}.
    \label{eq:ideal_score_no_vacuum_2}
    \end{equation}

    Fig.~\ref{fig:lxe_ideal_a} shows the values of the scores $s(\tilde{\bm{A}}_{\mathrm{sqz}};2N)$ and $s(\bm{A}_{\mathrm{sqz}};2N)$ as functions of $2N$ for $R\in\{10,20,50,150\}$. The computation of the coefficients $\{c_\ell\}$ becomes increasingly difficult for values of $2N>8$ due to the sharp increase in the number of elements in $\bar{S}$. For this reason, we only show the scores for $2\leq2N\leq8$. As can be seen, the value of the score for finite $R$ is greater than the score of an ideal model with no vacuum input modes. For small values of $R$, $s(\bm{A}_{\mathrm{sqz}};2N)$ quickly diverges from $s(\tilde{\bm{A}}_{\mathrm{sqz}};2N)$ when we increase $2N$. In contrast, for values of $R\sim 150$, $s(\bm{A}_{\mathrm{sqz}};2N)$ seems to closely follow $s(\tilde{\bm{A}}_{\mathrm{sqz}};2N)$. As shown in Fig.~\ref{fig:lxe_ideal_b}, this behavior can also be expected for models that use input squeezed states with different squeezing parameters.    

    Eq.~\eqref{eq:ideal_score_no_vacuum_2} is consistent with the ideal score found in a recent work by Ehrenberg et al.~\cite{ehrenberg2023transition, ehrenberg2024second}. Their result was obtained by computing the first and second moments of the modulus squared of hafnians of random Gaussian matrices, which, according to the hiding conjecture~\cite{deshpande2022quantum, arkhipov2012bosonic, grier2022complexity}, approximate the GBS distribution in the photon-collision-free limit. We did not use the hiding conjecture in our derivation of Eq.~\eqref{eq:ideal_score_no_vacuum_2}. However, defining the score in the limit $M\rightarrow\infty$ implies that our results are only valid in the photon-collision-free regime.

    The analytical computation of the LXE score for setups that are not in the limit of $M\rightarrow \infty$ is beyond the scope of this work. However, we can study the behavior of the ideal score in this regime by numerically computing $\bar{s}_{\mathrm{sqz}}(2N)$ as indicated in Eq.~\eqref{eq:lxe_score_estimation}. Let us recall that $\bar{s}(N)$ is an estimator of the LXE score determined by computing the probabilities of a given set of samples with respect to an ideal squeezed state model. In the case of $\bar{s}_{\mathrm{sqz}}(2N)$, the samples correspond to the ideal model itself. The procedure to determine $\bar{s}_{\mathrm{sqz}}(2N)$ is the following: For a given value of $M$, generate a Haar-random unitary matrix. Using this unitary, generate a set of $L$ samples from the probability distribution of the ideal squeezed state model, and compute their corresponding probabilities. All the samples must have $2N$ detected photons. Finally, compute the estimator 
    $\bar{s}_{\mathrm{sqz}}(2N)$ using Eq.~\eqref{eq:lxe_score_estimation}.

    It is important to remember that sampling from the probability distribution of an ideal squeezed state model, as well as computing the corresponding probabilities, is a computationally hard task whose cost grows exponentially with the total number of detected photons in the samples. This means that the numerical computation of $\bar{s}_{\mathrm{sqz}}(2N)$ is restricted, in practice, to low values of $2N$.

    We numerically estimated $\bar{s}_{\mathrm{sqz}}(2N)$ for GBS setups with $M\in\{50,100,200\}$. We chose these values of $M$ because they are close to those used in recent experimental implementations of GBS~\cite{zhong2020quantum, zhong2021phase, madsen2022quantum, deng2023gaussian}. We focused on ideal models with no vacuum input modes so we can compare the estimated scores with Eq.~\eqref{eq:ideal_score_no_vacuum_2}, which can be easily computed for $2N>8$. We also considered input states with the same squeezing parameter, which was set so that the mean number of photons were $\bar{n}=20$ for all three values of $M$. Therefore, $r=0.60$ for $M=50$, $r=0.43$ for $M=100$, and $r=0.31$ for $M=200$. Notice, however, that the definition of $\bm{A}_{\mathrm{sqz}}$ makes $\bar{s}_{\mathrm{sqz}}(2N)$ independent of our choice of the squeezing parameters.   
    For each value of $M$, we generated 10 Haar-random unitaries, and for each unitary we generated $L=1000$ samples per value of $2N$. We computed the scores for $10\leq 2N \leq 26$. The generation of all Haar-random unitaries and samples, as well as the computation of the probabilities of each individual sample, were done using the methods in the library \texttt{thewalrus}~\cite{gupt2019thewalrus}.

    Fig.~\ref{fig:lxe_score_finite_M} shows the results of the computation of $\bar{s}_{\mathrm{sqz}}(2N)$ as a function of $2N$. The colored circles in the small figures at the bottom represent the values of $\bar{s}_{\mathrm{sqz}}(2N)$ for each of the 10 Haar-random unitaries. The dashed lines with error bars correspond to the average of these values, which we denote $\langle\bar{s}_{\mathrm{sqz}}(2N)\rangle$. As was mentioned below Eq.~\eqref{eq:lxe_score_estimation}, the estimated score can be interpreted as the average of the probabilities of each individual sample (with respect to the ideal model) multiplied by a constant. Consequently, the error of each $\bar{s}_{\mathrm{sqz}}(2N)$ corresponds to the standard error of the mean. The uncertainty of $\langle\bar{s}_{\mathrm{sqz}}(2N)\rangle$, which leads to the error bars shown in Fig.~\ref{fig:lxe_score_finite_M}, is computed through error propagation.

    \begin{figure}[!t]
        \centering
        \includegraphics[width=\columnwidth]{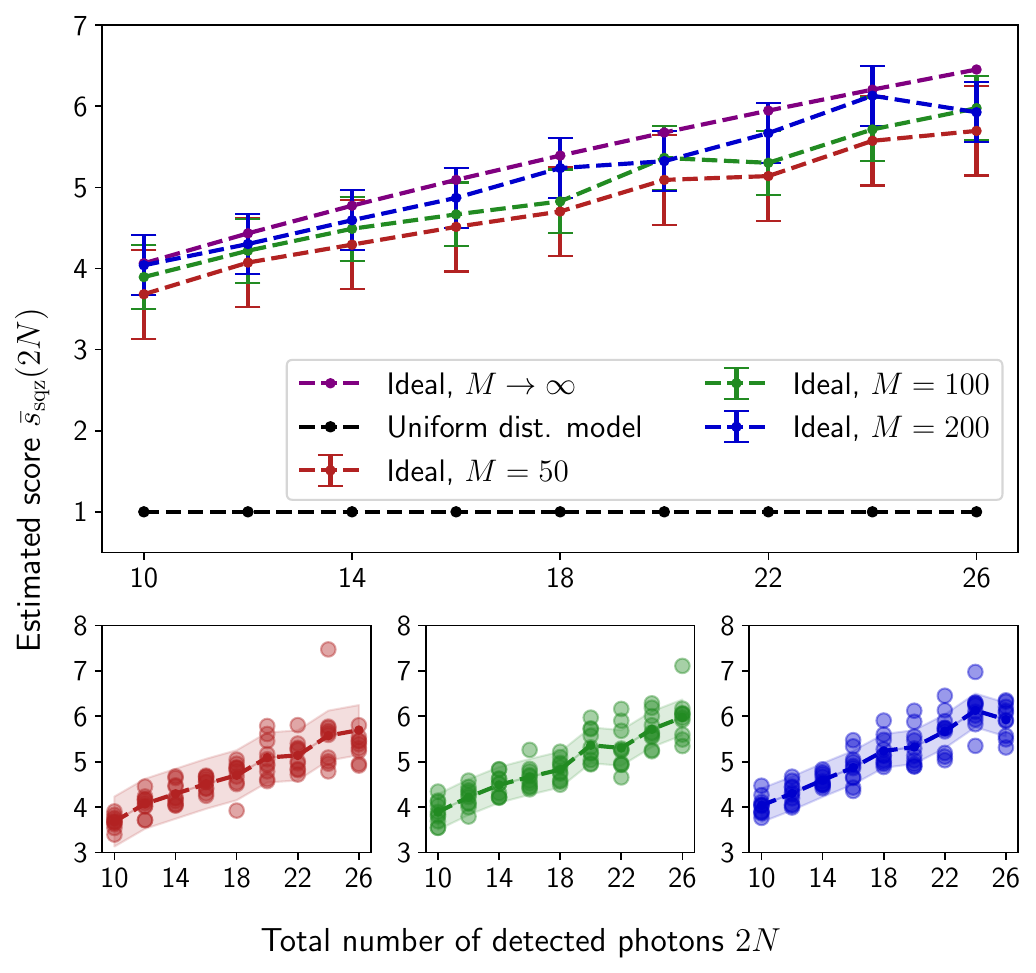}
        \caption{Estimated score $\bar{s}_{\mathrm{sqz}}(2N)$ as a function of $2N$. The violet dashed line corresponds to the score of an ideal model with $M\rightarrow\infty$. The black dashed line corresponds to a model that leads to a uniform distribution over the sample space. The remaining lines show the average estimated score $\langle \bar{s}_{\mathrm{sqz}}(2N)\rangle$: the red dashed line corresponds to $M=50$, the green dashed line to $M=100$, and the blue dashed line to $M=200$. Error bars show the uncertainty in the estimation of the average score. The colored circles shown in the small plots at the bottom indicate the values of $\bar{s}_{\mathrm{sqz}}(2N)$ for each of the 10 Haar-random unitaries used in the calculation of the corresponding average. The estimation of each of these values was done using $L=1000$ samples. The shaded regions indicate how many unitaries obtain values of the score within the error of the estimated average value.}
        \label{fig:lxe_score_finite_M}
    \end{figure}

    We can see that the average estimated score, for all the values of $M$ considered, closely resembles the analytical result in the limit of $M\rightarrow \infty$.
    Indeed, the relative difference between $\langle\bar{s}_{\mathrm{sqz}}(2N)\rangle$ and $s(\tilde{\bm{A}}_{\mathrm{sqz}};2N)$, 
    \begin{equation}
        \Delta(2N) = \frac{s(\tilde{\bm{A}}_{\mathrm{sqz}};2N)-\langle\bar{s}_{\mathrm{sqz}}(2N)\rangle}{s(\tilde{\bm{A}}_{\mathrm{sqz}};2N)}\times100\%,
        \label{eq:relative_difference_scores}
    \end{equation}
    satisfies $8.12\%\leq\Delta(2N)\leq13.57\%$ for $M=50$, $4.17\%\leq\Delta(2N)\leq10.82\%$ for $M=100$, and $0.58\%\leq\Delta(2N)\leq8.16\%$ for $M=200$. 
    
    Interestingly, $\langle\bar{s}_{\mathrm{sqz}}(2N)\rangle$ approaches $s(\tilde{\bm{A}}_{\mathrm{sqz}};2N)$ more rapidly for small values of $2N$, while $\langle\bar{s}_{\mathrm{sqz}}(2N)\rangle$ and $s(\tilde{\bm{A}}_{\mathrm{sqz}};2N)$ seem to diverge for large $2N$. The divergence is faster the smaller is $M$. This suggest that the range of total detected photons for which $s(\tilde{\bm{A}}_{\mathrm{sqz}};2N)$ will be a good approximation of $\bar{s}_{\mathrm{sqz}}(2N)$ depends on the number of modes considered. One condition to ensure that the GBS distribution is in the photon-collision-free regime is that the mean number photons, $\bar{n}$, satisfies $\bar{n}\in o(\sqrt{M})$~\cite{deshpande2022quantum, arkhipov2012bosonic, grier2022complexity}. Since $\bar{n}$ determines the range of $2N$ in which it is more likely to find experimental samples, and considering that Eq.~\eqref{eq:ideal_score_no_vacuum_2} is only valid in the collision-free limit, we may say that $s(\tilde{\bm{A}}_{\mathrm{sqz}};2N)$ will be a good approximation of $\bar{s}_{\mathrm{sqz}}(2N)$ for a range of total detected photons satisfying $2N\in o(\sqrt{M})$.
    
    We can also see that most of the estimated $\bar{s}_{\mathrm{sqz}}(2N)$ take values within the uncertainty of the average estimated score, even for a number of modes as low as $M=50$. This confirms that typical GBS implementations obtain ideal scores that are close to the average over Haar-random unitaries.  

    \begin{figure*}[!t]
        \centering
        \includegraphics[width=\textwidth]{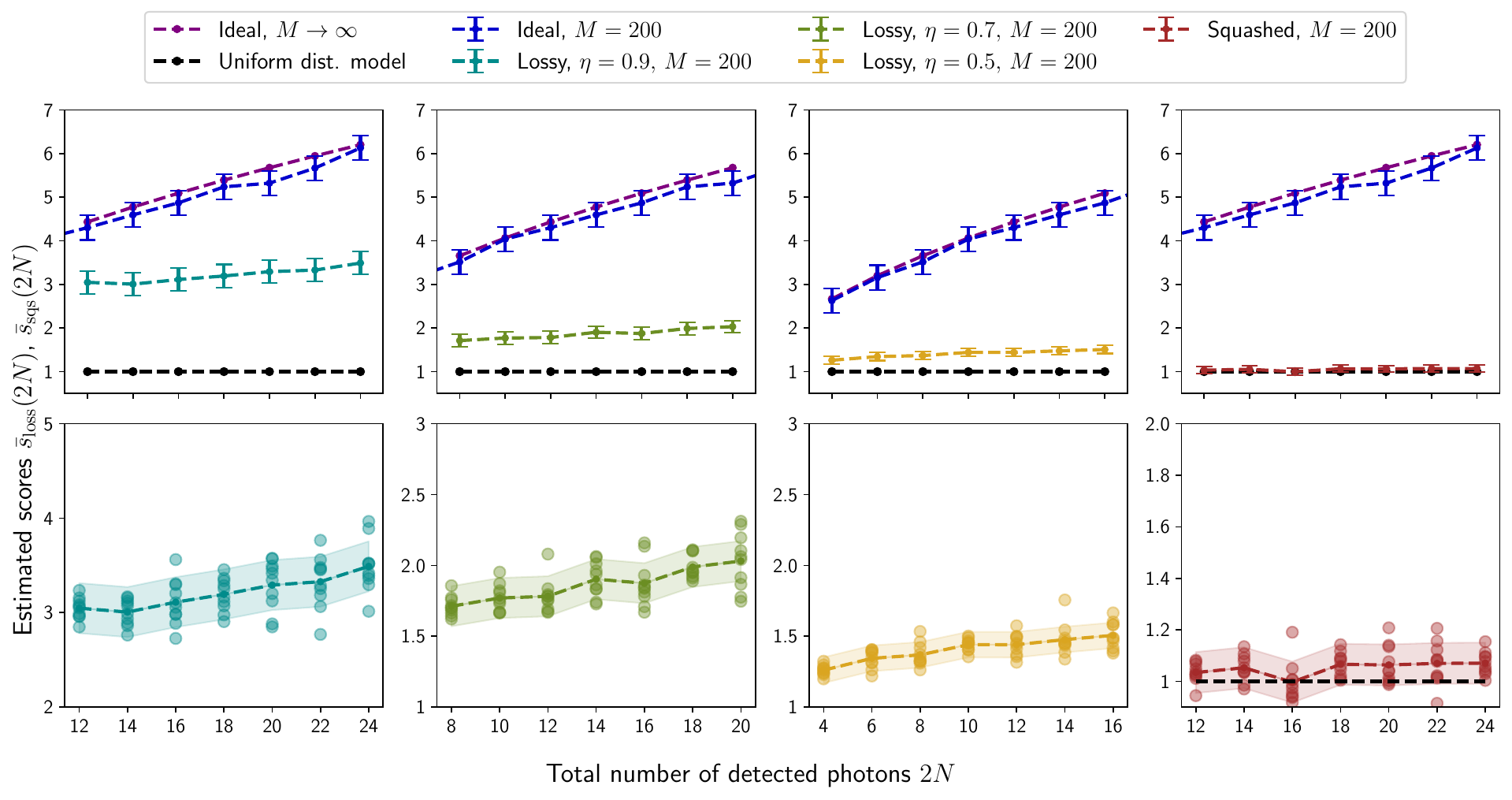}
        \caption{Estimated scores $\bar{s}_{\mathrm{loss}}(2N)$ and $\bar{s}_{\mathrm{sqs}}(2N)$ as functions of $2N$. The violet dashed line corresponds to the score of an ideal model with $M\rightarrow\infty$. The black dashed line corresponds to a model that leads to a uniform distribution over the sample space. The dark blue dashed line with error bars corresponds to the average estimated score of an ideal model with $M=200$. The remaining lines show the average estimated scores $\langle \bar{s}_{\mathrm{loss}}(2N)\rangle$ and $\langle \bar{s}_{\mathrm{sqs}}(2N)\rangle$: the teal dashed line corresponds to a lossy squeezed state model with $\eta=0.9$, the olive green dashed line to $\eta=0.7$, the dark yellow dashed line to $\eta=0.5$, and the dark red line to the squashed state model. Error bars show the uncertainty in the estimation of the average score. The colored circles shown in the plots at the bottom indicate the values of $\bar{s}_{\mathrm{loss}}(2N)$ and $\bar{s}_{\mathrm{sqs}}(2N)$ for each of the 10 Haar-random unitaries used in the calculation of the corresponding average. The estimation of each of these values was done using $L=1000$ samples. The shaded regions indicate how many unitaries obtain values of the score within the error of the estimated average value.}
        \label{fig:lxe_score_loss}
    \end{figure*}

    To finish this section, we would like to numerically investigate the 
    score that would be obtained by GBS setups with transmission losses. This
    study would give us an idea of how sensitive is the LXE score to the
    presence of noise. To do this, we consider a simple model in which the
    input squeezed states are sent through single-mode loss channels before
    entering a Haar-random unitary interferometer. We will consider that there are no vacuum input modes, that all input states have the same
    squeezing, and that all loss channels have the same transmission
    parameter, $\eta$. When the transmission losses are high, the state of
    the light after the single-mode loss channel can be approximated by a
    squashed state~\cite{martinez2023classical}. We will use a squashed state
    model where all input states have the same mean number of photons to
    investigate the case of extreme losses. The details on how obtain the
    matrices $\bm{A}_{\mathrm{loss}}$ and $\bm{A}_{\mathrm{sqs}}$ associated
    to the lossy squeezed state and squashed state models can be found in
    Appendix~\ref{app:misc}.
    
    It is worth mentioning that the squashed states are classical Gaussian states~\cite{martinez2023classical}, and thus sampling from the probability distribution of model $\bm{A}_{\mathrm{sqs}}$ can be done efficiently. The lossy squeezed states used in model $\bm{A}_{\mathrm{loss}}$ are not classical states. However, it has been shown that sampling from $\boldsymbol{A}_{\mathrm{loss}}$ when transmission losses are high can be done with reasonable classical hardware resources (for models with hundreds of modes, millions of samples can be generated in the order of an hour using contemporary supercomputers)~\cite{oh2024classical}. For both models, the calculation of probabilities with respect to model $\bm{A}_{\mathrm{sqz}}$ is still computationally hard. Therefore, the numerical computation of the estimators will also be restricted to low values of $2N$.

    We numerically estimated $\bar{s}_{\mathrm{loss}}(2N)$ and $\bar{s}_{\mathrm{sqs}}(2N)$ for a GBS setup with $M=200$. We considered transmission parameters $\eta\in\{0.5,0.7,0.9\}$, and set the squeezing parameter of the input states to $r=0.31$. The mean number of photons for the squashed state model was $\bar{n}=20$. Here, we also find that the definition of model $\bm{A}_{\mathrm{sqs}}$ makes $\bar{s}_{\mathrm{sqs}}(2N)$ independent of our choice of $\bar{n}$. For each of these cases, we generated 10 Haar-random unitaries, and for each unitary we generated $L=1000$ samples per value of $2N$. We computed the scores for $2N$ between 4 and 24.
    
    Fig.~\ref{fig:lxe_score_loss} shows the computed values of $\bar{s}_{\mathrm{loss}}(2N)$ and $\bar{s}_{\mathrm{sqs}}(2N)$ as functions of $2N$. Just as before, the colored circles in the figures at the bottom represent the values of estimators for each of the 10 Haar-random unitaries. The dashed lines with error bars correspond to $\langle\bar{s}_{\mathrm{loss}}(2N)\rangle$ and $\langle\bar{s}_{\mathrm{sqs}}(2N)\rangle$. The error of each $\bar{s}_{\mathrm{loss}}(2N)$ and $\bar{s}_{\mathrm{sqs}}(2N)$ corresponds to the standard error of the mean, and the uncertainty of the corresponding averages is computed through error propagation.

    As can be seen, the values of $\langle\bar{s}_{\mathrm{loss}}(2N)\rangle$ are significantly lower than those corresponding to the estimated ideal score. We find that the relative difference between $\langle\bar{s}_{\mathrm{loss}}(2N)\rangle$ and $\langle\bar{s}_{\mathrm{sqz}}(2N)\rangle$ satisfies $29.15\%\leq\Delta(2N)\leq43.06\%$ for $\eta=0.9$, $51.31\%\leq\Delta(2N)\leq62.02\%$ for $\eta=0.7$, and $52.03\%\leq\Delta(2N)\leq69.07\%$ for $\eta=0.5$. For the squashed state model, the difference is even more striking. The relative difference between $\langle\bar{s}_{\mathrm{sqs}}(2N)\rangle$ and $\langle\bar{s}_{\mathrm{sqz}}(2N)\rangle$ satisfies $75.96\%\leq\Delta(2N)\leq82.53\%$, and we can readily notice that the estimated $\bar{s}_{\mathrm{sqs}}(2N)$ take values close to 1 for all the unitaries used in the computation. Finally, we confirm that most of the estimated $\bar{s}_{\mathrm{loss}}(2N)$ and $\bar{s}_{\mathrm{sqs}}(2N)$ take values within the uncertainty of the average estimated score. This demonstrates that typical GBS implementations with transmission losses also obtain LXE scores that are close to the average over Haar-random unitaries.  

    These results show that the presence of transmission losses in GBS setups do not necessarily bring the LXE score to values close to those obtained by models that lead to a uniform distribution over the sample space. Indeed, only for situations in which the transmission losses are higher than $50\%$ does the score of a lossy squeezed state model approach the value of 1. However, we notice that losses as low as $10\%$ seem to define a reference value that is easily distinguishable from that of the ideal model. This will prove useful at the moment of verifying real-world GBS implementations, as we can define clearly delimited regions associated to certain values of the transmission loss and then determine in which of them lies the score associated to a set of experimental samples. A more thorough study of the LXE score for models with transmission losses, as well as other types of Gaussian and non-Gaussian noise, is needed in order to completely determine the reference values delimiting these regions. 

    We also bring attention to the fact that the squashed state model obtains estimated scores very close to 1. These models have proven to be good classical adversaries of recent GBS implementations, performing as good as, and sometimes better than, the ground truth at a number of validation tests~\cite{martinez2023classical}. Therefore, it would be valuable to determine the analytical value of the LXE score for the squashed model, and verify if it is identically 1.

    The definition of analytical reference values corresponding to noisy models will be of central importance for the use of the LXE score as a GBS validation metric. However, we need also determine if obtaining a score greater than a given reference value is a computationally hard task or not. In this way, we would be certain that the LXE score is a good witness of quantum computational advantage. Future studies on the LXE score need also focus on this subject.
    
    \section{\label{sec:discussion}Discussion}
    
    In this work, we proposed to use the linear-cross entropy score (LXE score) as a tool for validating quantum advantage claims in the context of Gaussian Boson Sampling (GBS). Taking inspiration from the definition of the linear-cross entropy benchmark (XEB) used for the validation or Random Circuit Sampling (RCS) implementations, we defined the LXE score as a normalized version of the linear cross-entropy between two GBS models, one of them being the ideal squeezed state model, averaged over the Haar measure of the unitary group, and evaluated at the limit of a large number of modes. The key idea of the verification strategy consists in finding reference values of the LXE score (corresponding to the ideal squeezed state model, known classical models, and adversarial samplers) and compare them with the estimated score obtained by the outcomes of a given GBS experiment. Using this comparison, we can assess how far a GBS implementation is from its corresponding typical, ideal model.
    
    Following the definition of the LXE score, we identified two of its reference values: the first one corresponding to a GBS setup leading to a uniform probability distribution for each sector of the total number $N$ of detected photons in the experimental samples. The second one is associated to a GBS implementation using a lossless interferometer that receives single-mode squeezed states in the first $R$ of its $M$ input modes, and the vacuum in the rest of them. The first reference value follows directly from the normalization of the LXE score and is equal to 1. In addition to this, we found two expressions for the second reference value as a function of $R$ and $N$: Eq.~\eqref{eq:main_result_lxe_sqz} for the case of input states having the same squeezing parameter, and Eq.~\eqref{eq:lxe_sqz_score_different_squeezing} for input sates with different squeezing parameters. It is worth mentioning that even if the ideal value of the LXE score is not constant with $N$ (unlike its RCS counterpart) the validation strategy remains sound, we need only compare the estimated score of the experimental outcomes with respect to the reference curve set by the ideal squeezed state model.

    An important feature of the expression we found for the ideal LXE score is that part of its dependence on $R$ can be written as a polynomial of degree $2N$ in this variable. The corresponding coefficients can be computed by counting the number of undirected graphs, associated to certain permutations, with a given number of connected components. When we consider setups that use input states with different squeezing parameters, the coefficients also depend on the lengths of the connected components of the graph. This property of the ideal score closely resembles a recent result by Ehrenberg et al.~\cite{ehrenberg2023transition, ehrenberg2024second} concerning a study on anticoncentration in GBS. In their work, they developed a graph-theoretical technique  to compute the first and second moments of the output GBS distribution in the photon-collision-free regime. They found that the second moment of the distribution 
    can be expressed as a polynomial of degree $2N$ in $R$, with coefficients computed by counting how many graphs have a specific number of connected components. Moreover, they find an expression for the LXE score of an ideal squeezed state model that uses input states with the same squeezing parameter in terms of the first and second moments~\cite{ehrenberg2024second}. 
    
    Our expression for the ideal LXE score, in the case of equal squeezing, and that of Ehrenberg et al. are very similar. The main difference is that their results rely on the GBS hiding conjecture~\cite{deshpande2022quantum, arkhipov2012bosonic, grier2022complexity}, which states that, in the photon-collision-free regime, the distribution of the symmetric product of Haar-random unitary matrices closely approximates the distribution of the symmetric product of complex Gaussian random matrices. In contradistinction, our results rely on the distribution of Haar-random unitaries in the asymptotic limit. Additionally, the definition and origin of the graphs pertinent to our computation of the ideal score differ from those used in Refs.~\cite{ehrenberg2023transition, ehrenberg2024second}. At this time, more work is required in order to find the relation between the graph- theoretical technique of Ehrenberg et al. and the computation of the coefficients defined in Eq.~\eqref{eq:coefficients_sqz_lxe_score}. While our results provide a complete description to determine the ideal LXE score, we must also acknowledge that a more detailed analysis of  Eqs.~\eqref{eq:coefficients_sqz_lxe_score},~\eqref{eq:b_coefficients_definition} and~\eqref{eq:c_tilde_coeffs} is needed in order to make the computation of the score better-suited for the validation real-world GBS implementations.

    When considering setups with no vacuum input modes, i.e., setups with $R=M$, we found a simple expression for the ideal score as a function of $2N$ (see Eq.~\eqref{eq:ideal_score_no_vacuum_2}). We also proved that this equation holds whether we use input squeezed states with different squeezing parameters or not. We compared this result, which was obtained taking the limit as $M\rightarrow\infty$, with numerical estimations of the ideal score for setups with $M\in\{50,100,200\}$. We found that, for all $M$, the analytical expression closely approximates the estimated scores in a range of $2N$ between 10 and 26. However, we noticed that this approximation worsens for increasing $2N$. We argue that the reason behind this is that, being defined for $M\rightarrow \infty$, our expression for the ideal score is only valid in the photon-collision-free regime. Since in this limit the mean number of photons satisfies $\bar{n}\in o(\sqrt{M})$, we expect Eq.~\eqref{eq:ideal_score_no_vacuum_2} to be a good approximation of the estimated score, for a given $M$, for a range of $2N$ satisfying $2N\in o(\sqrt{M})$. 

    We also computed the estimated score for a squashed state model, and for a simple GBS model that includes transmission losses. We found that the presence of these type of loss does not necessarily lead to scores that rapidly approach the value of 1. However, transmission losses as low as $10\%$ lead to an estimated score that is notably different from the reference value set by an ideal model. On the other hand, we found that the estimated score of the squashed state model is surprisingly close to 1. These findings suggest that computing the score for different values of the transmission loss, as well as for classical models such as the squashed state model, is a viable way of defining reference values that will allow us to assess how far a given GBS implementation is from its corresponding ideal model.  
    
    Additional work is required to analytically determine the score of noisy and classical GBS models, as well as to prove whether the score of the squashed state model is identically 1 or not. It would also be highly relevant to provide sufficient evidence to decide if obtaining a score sufficiently close to that of the ideal model (or sufficiently different from those of classical models) is a computationally hard task or not. Our future studies of the LXE score will follow these directions.
    
    The use of the LXE score will greatly benefit the field of GBS verification. Since we have defined the score using an ideal model instead of the ground truth, the validation of GBS implementations using this metric will not rely on providing evidence that the outcomes of the experiment follow its expected theoretical distribution (of which, additionally, we do not know if it is computationally hard to sample). Verifying directly against the ideal model will also shield the validation procedure from classical techniques that directly intend to simulate the ground truth. In view of this, we believe that obtaining a high LXE score constitutes a stricter test that current and future GBS implementations should consider when providing evidence in support of quantum advantage claims. 
    
    In addition, since the computation of probability amplitudes for pure Gaussian states requires the calculation of hafnians of matrices with half the size of those used in the computation of probabilities of mixed states, the estimation of the LXE score can be done for a range of $N$ that is approximately twice as large as those considered in the validation of recent GBS implementations~\cite{zhong2020quantum,zhong2021phase,madsen2022quantum,deng2023gaussian}. This represents a considerable improvement on the verification of GBS experiments.

    Moreover, given that estimating the LXE score relies on the computation of probability amplitudes of pure states, our validation technique is protected against ``spoofing'' attacks that rely on efficiently simulating lossy GBS implementations, postselecting samples with high probabilities in the ground truth (i.e., with heavy outcomes), and then computing cross-entropy measures that are also defined with respect to the ground truth~\cite{oh2023spoofing} (thus obtaining higher scores than the experimental samples). A direct simulation of an ideal GBS experiment is computationally as hard as the computation of the LXE score~\cite{quesada2022quadratic}, so ``spoofing'' our validation strategy using heavy outputs in the ideal distribution is unlikely.

    \textit{Note added.} After the publication of a first version of this work on the arXiv, we became aware of the second article of Ehrenberg et al.~\cite{ehrenberg2024second}, where there is an expression for the ideal LXE score in terms of moments of the GBS distribution. We now include references to this work in the main text.

    \section*{Data and code availability}
    All the data used in the computation of the results shown in Figs.~\ref{fig:lxe_score_2N} to~\ref{fig:lxe_score_loss} can be found at \url{https://doi.org/10.5281/zenodo.13327744}. The corresponding code is available at \url{https://github.com/polyquantique/lxe_score_for_gbs}.

    \section*{Acknowledgements}
    
	We acknowledge the support from the Ministère de l’Économie et de l’Innovation du Québec and the Natural Sciences and Engineering Research Council of Canada (NSERC).  HdG would like to thank NSERC. NQ and HdG would also like to thank ACFAS for partial support of this project. This research was enabled in part by support provided by the Digital Research Alliance of Canada.
    
    \bibliography{lxe_revtex}

    \onecolumngrid
    
    \appendix
    
    \section{\label{app:misc}Compilation of results for the computation of the integral form of the LXE}

    In this appendix we present the proofs of some of the results used in the computation of the integral form of the LXE. We also show that a GBS setup using identical thermal states as input of all modes of a lossless interferometer leads to a probability distribution that is uniform for each sector of the total number of detected photons. 
    
    We begin by proving, in a slightly more general fashion, the factorization $\left(\bm{W}\bm{A}\bm{W}\right)_{\bm{n}}=\bm{\bar{W}}_{\bm{n}}\bm{A}_{\bm{n}}\bm{\bar{W}}_{\bm{n}}$ used in the derivation of Eq.~\eqref{eq:integral_lxe_3}.
    
    \begin{proposition}
        \label{prop:patt_matrix}
        Let $\bm{M}=(M_{j,k})_{j,k=1}^P$ be a complex $P\times P$ matrix, and let $\bm{n}=(n_1,\dots,n_P)$ be a vector whose entries are non-negative integers. Define the matrix $\bm{M}_{(\bm{n})}$ as the matrix obtained from taking the $k$-th row and column of $\bm{M}$ and repeating it $n_k$ times. This matrix is a $N\times N$ block matrix (with $N=\sum_{k=1}^P n_k$) of the form $\bm{M}_{(\bm{n})} = (\bm{B}_{p,q})_{q,p=1}^P$, where the blocks $\bm{B}_{p,q}$ have size $n_p\times n_q$ and their corresponding entries are all equal to $N_{q,p}$. Furthermore, let $\bm{R}$, $\bm{L}$ be $P\times P$ complex diagonal matrices of the form $\bm{R} = \mathrm{diag}(r_1,\dots,r_P)$, $\bm{L} = \mathrm{diag}(l_1,\dots,l_P)$. Then $(\bm{L}\bm{M}\bm{R})_{(\bm{n})} = \bm{\bar{L}}_{(\bm{n})}\bm{M}_{(\bm{n})}\bm{\bar{R}}_{(\bm{n})}$, where $\bm{\bar{R}}_{(\bm{n})}=\bigoplus_{k=1}^Pr_k\mathbb{I}_{n_k}$ and $\bm{\bar{L}}_{(\bm{n})}=\bigoplus_{k=1}^Pl_k\mathbb{I}_{n_k}$ with $\mathbb{I}_{n_k}$ the identity matrix of size $n_k \times n_k$. 
    \end{proposition}
    
    \begin{proof}
        The entries of $\bm{L}\bm{M}\bm{R}$ can be written as $(\bm{L}\bm{M}\bm{R})_{j,k}=l_jM_{j,k}r_k$. Also, suppose that $\bm{D}$ is a $S\times T$ matrix whose entries are all equal $xCy$, where $x$, $y$, $C$ are complex numbers. We can see that such a matrix can be written in the form $\bm{D}=(x\,\mathbb{I}_S)\bm{C}(y\,\mathbb{I}_T)$, where $\bm{C}$ is a $S\times T$ matrix whose entries are all equal to $C$. According to the definition, $(\bm{L}\bm{M}\bm{R})_{(\bm{n})}=(\bm{K}_{p,q})_{p,q=1}^P$, where the entries of the $n_p\times n_q$ blocks are all equal to $l_pM_{p,q}r_q$, which implies $\bm{K}_{p,q}=(l_p\mathbb{I}_{n_p})\bm{B}_{p,q}(r_q\mathbb{I}_{n_q})$. From this expression, we can define the block diagonal matrices $\bm{\bar{R}}_{(\bm{n})}=\bigoplus_{k=1}^Pr_k\mathbb{I}_{n_k}$ and $\bm{\bar{L}}_{(\bm{n})}=\bigoplus_{k=1}^Pl_k\mathbb{I}_{n_k}$ and check that $(\bm{L}\bm{M}\bm{R})_{(\bm{n})} = \bm{\bar{L}}_{(\bm{n})}\bm{M}_{(\bm{n})}\bm{\bar{R}}_{(\bm{n})}$.
    \end{proof}
    
    The relation $\left(\bm{W}\bm{A}\bm{W}\right)_{\bm{n}}=\bm{\bar{W}}_{\bm{n}}\bm{A}_{\bm{n}}\bm{\bar{W}}_{\bm{n}}$ can be recovered by observing that $\bm{A}_{\bm{n}}=\bm{A}_{(\bm{n}\oplus\bm{n})}$.
    
    We now prove that $\Pr(N|\bm{A})$ can be computed by repeated differentiation of $q(\alpha,\bm{0},\bm{A})$.
    
    \begin{proposition}
        \label{prop:total_photons_dist}
        For any GBS model $\bm{A}$
        \begin{equation}
            \Pr(N|\bm{A})=\frac{1}{N!}\Pr(\bm{0}|\bm{A})\left.\frac{\partial^N q(\alpha, \bm{0},\bm{A})}{\partial \alpha^N}\right|_{\alpha=0}\, ,
            \label{eq:total_photon_prob}
        \end{equation}
        where $q(\alpha,\bm{0},\bm{A})=[\mathrm{det}(\mathbb{I}_{2M}-\alpha\bm{X}\bm{A})]^{-1/2}$ is a specialization of Eq.~(\ref{eq:q_function_def}).
    \end{proposition}

    \begin{proof}
        Let us begin by the relation $1=\sum_{\bm{n}}\Pr(\bm{n}|\bm{A})$, where the sum extends over all possible detection patterns. In terms of the hafnian, this relation can be written as
        \begin{equation}
            \frac{1}{\Pr(\bm{0}|\bm{A})}=\sum_{N=0}^\infty\sum_{\bm{n}\in K(N)}\frac{1}{\bm{n}!}\mathrm{haf}\left[\bm{A}_{\bm{n}}\right].    
        \end{equation}
        Making the replacement $\bm{A}\rightarrow\alpha\bm{A}$, we can use the property  $\mathrm{haf}[\alpha\bm{A}_{\bm{n}}]=\alpha^{|\bm{n}|}\mathrm{haf}[\bm{A}_{\bm{n}}]$ to write 
        \begin{equation}
            q(\alpha,\bm{0},\bm{A})=\sum_{N=0}^\infty\alpha^N \left(\sum_{\bm{n}\in K(N)}\frac{1}{\bm{n}!}\mathrm{haf}\left[\bm{A}_{\bm{n}}\right]\right)\, .    
        \end{equation}
       By repeatedly differentiating with respect to $\alpha$ we find
        \begin{equation}
            \left.\frac{1}{N!}\frac{\partial^N q(\alpha,\bm{0},\bm{A})}{\partial \alpha^N}\right|_{\alpha=0}= \sum_{\bm{n}\in K(N)}\frac{1}{\bm{n}!}\mathrm{haf}\left[\bm{A}_{\bm{n}}\right],   
        \end{equation}
        which implies
        \begin{equation}
            \Pr(N|\bm{A})=\frac{1}{N!}\Pr(\bm{0}|\bm{A})\left.\frac{\partial^N q(\alpha, \bm{0},\bm{A})}{\partial \alpha^N}\right|_{\alpha=0}.
        \end{equation}    
    \end{proof}
    
    From this result we can readily see that
    \begin{align}
        \mathcal{D}(\bm{A},\bm{B};N)=\frac{1}{(N!)^2}\frac{\Pr(\bm{0}|\bm{A})\Pr(\bm{0}|\bm{B})}{\Pr(N|\bm{A})\Pr(N|\bm{B})}
        =\left[ 
        \frac{\partial^N q(\alpha, \bm{0},\bm{A})}{\partial \alpha^N} 
        \left.\frac{\partial^N q(\beta, \bm{0},\bm{B})}{\partial \beta^N}\right|_{\substack{\alpha=0\\\beta=0}}\right]^{-1}.
    \end{align}

    For the remaining demonstrations of this section, we need to give a more detailed description of the GBS models introduced in Sec.~\ref{sec:gbs_setup}.

    When we consider input single-mode Gaussian states and a lossless interferometer represented by a Haar-random unitary matrix, the matrix $\bm{A}$ will have the general structure
    \begin{equation}
        \bm{A}=
        \begin{pmatrix}
            \bm{V}&\bm{Y}\\
            \bm{Y}^*&\bm{V}^*
        \end{pmatrix},
        \label{eq:gbs_model_general_structure}
    \end{equation}
    where $\bm{V}=\bm{U}\bm{\lambda}\bm{U}^{\mathrm{T}}$ is a symmetric matrix, and $\bm{Y}=\bm{U}\bm{\mu}\bm{U}^{\dagger}$ is Hermitian. The matrices $\bm{\lambda}$, $\bm{\mu}$ are defined as $\bm{\lambda}=\mathrm{diag}(\lambda_1,\dots,\lambda_M)$, $\bm{\mu}=\mathrm{diag}(\mu_1,\dots,\mu_M)$, where the $\{\lambda_k\}$ and $\{\mu_k\}$ are real parameters that can be written in terms of the entries of the \textit{real covariance matrix} of the input single-mode Gaussian states as
    \begin{align}
        \lambda_k &= \frac{1}{1+2\sigma_p^{(k)}/\hbar}-\frac{1}{1+2\sigma_x^{(k)}/\hbar},\label{eq:real_parameter_general_moments_1}\\
        \mu_k &= 1-\left(\frac{1}{1+2\sigma_x^{(k)}/\hbar}+\frac{1}{1+2\sigma_p^{(k)}/\hbar}\right).
        \label{eq:real_parameter_general_moments_2}
    \end{align}
    
    Recall that the entries of the real covariance matrix, $\bm{\sigma}$, of a $M$-mode, non-displaced Gaussian state are computed as $\sigma_{j,k}=\frac{1}{2}\left\langle\{\hat{r}_j,\hat{r}_k\}\right\rangle$, where the $\{\hat{r}_k\}$ are the components of the operator vector $\bm{\hat{r}}=(\hat{x}_1,\dots,\hat{x}_M,\hat{p}_1\dots,\hat{p}_M)$ with $\hat{x}_k$, $\hat{p}_k$ the quadrature operators of mode $k$. For a single-mode Gaussian state, we can write $\sigma_x^{(k)}=\langle \hat{x}^2 \rangle$ and $\sigma_p^{(k)}=\langle \hat{p}^2 \rangle$. Notice that we can neglect the non-diagonal entries $\sigma_{x,p}^{(k)}=\sigma_{p,x}^{(k)}=\frac{1}{2}\left\langle\{\hat{x},\hat{p}\}\right\rangle$ because these can be obtained from a diagonal covariance matrix via a local rotation (i.e., a local phase-shift), which can be absorbed into the unitary operator describing interferometer~\cite{rahimi2015can}. 
    Using $\bm{\sigma}>0$, it follows that $\sigma_x^{(k)}, \sigma_p^{(k)}>0$ for all $k$. 

    For a GBS setup that uses identical thermal states with mean number of photons $\bar{n}$ as input of all the modes of the interferometer, we have $\sigma_x^{(k)}=\sigma_p^{(k)}=\hbar (2\bar{n}+1)/2$ for all $k$, which implies $\lambda_k=0$, $\mu_k=\bar{n}/(1+\bar{n})$ for all $k$. This leads to the following form of the matrix $\bm{A}_{\mathrm{thm}}$:
    \begin{equation}
        \bm{A}_{\mathrm{thm}}=\frac{\bar{n}}{\bar{n}+1}
        \begin{pmatrix}
            \bm{0}&\mathbb{I}_M\\
            \mathbb{I}_M&\bm{0}
        \end{pmatrix}
        =\frac{\bar{n}}{\bar{n}+1}\bm{X}.
        \label{eq:gbs_thermal_model}
    \end{equation}
    
    We may now prove that this model leads to a uniform probability distribution for each sector of the total number of detected photons.
    
    \begin{proposition}
        \label{prop:thermal_probability}
        \begin{equation}
            \Pr(\bm{n}|\bm{A}_{\mathrm{thm}})=\Pr(\bm{0}|\bm{A}_{\mathrm{thm}})\left(\frac{\bar{n}}{\bar{n}+1}\right)^{|\bm{n}|}.
            \label{eq:thermal_probability}
        \end{equation}
    \end{proposition}

    \begin{proof}
        Notice that
        \begin{equation}
            (\bm{A}_{\mathrm{thm}})_{\bm{n}}=\frac{\bar{n}}{\bar{n}+1}
        \begin{pmatrix}
            \bm{0}&(\mathbb{I}_M)_{(\bm{n})}\\
            (\mathbb{I}_M)_{(\bm{n})}&\bm{0}
        \end{pmatrix},    \label{eq:thermalA}
        \end{equation}
        where $\bm{M}_{(\bm{n})}$ is the matrix obtained by taking the $k$-th row and column of $\bm{M}$ and repeating them $n_k$ times. In the special case of $\mathbb{I}_M$ it can be seen that $(\mathbb{I}_M)_{(\bm{n})}=\bigoplus_{k=1}^M\bm{1}_{n_k}$, where $\bm{1}_m$ is a $m\times m$ matrix whose entries are all equal to one.

        Taking into account Eq.~(\ref{eq:thermalA}), we can write
        \begin{align}
            \mathrm{haf}[(\bm{A}_{\mathrm{thm}})_{\bm{n}}]&=\left(\frac{\bar{n}}{\bar{n}+1}\right)^{|\bm{n}|}
            \mathrm{haf}\left[
            \begin{pmatrix}
                \bm{0}&(\mathbb{I}_M)_{(\bm{n})}\nonumber\\
                (\mathbb{I}_M)_{(\bm{n})}&\bm{0}
            \end{pmatrix}
            \right]\\
            &=\left(\frac{\bar{n}}{\bar{n}+1}\right)^{|\bm{n}|}
            \mathrm{per}\left[(\mathbb{I}_M)_{(\bm{n})}\right]\\
            &=\left(\frac{\bar{n}}{\bar{n}+1}\right)^{|\bm{n}|}\bm{n}!,\nonumber
        \end{align}
        where $\mathrm{per}[\cdot]$ stands for the permanent of a matrix, and we have used the relation
        \begin{equation}
            \mathrm{haf}\left[
        \begin{pmatrix}
            \bm{0}&\bm{G}\\
            \bm{G}^{\mathrm{T}}&\bm{0}
        \end{pmatrix}
        \right] = \mathrm{per}[\bm{G}]    
        \end{equation}
        for any $m\times m$ matrix $\bm{G}$. Moreover, we have used the relation $\mathrm{per}\left[\bigoplus_{k=1}^M\bm{1}_{n_k}\right]=\prod_{k=1}^Mn_k!=\bm{n}!$~\cite{barvinok2016combinatorics}. It directly follows that
        \begin{equation}
            \Pr(\bm{n}|\bm{A}_{\mathrm{thm}})=\Pr(\bm{0}|\bm{A}_{\mathrm{thm}})\left(\frac{\bar{n}}{\bar{n}+1}\right)^{|\bm{n}|}.    
        \end{equation}        
    \end{proof}

    We finish this section by proving the convergence of the series expansion that leads to Eq.~\eqref{eq:q_function_expansion_1}. Recall that $q(\alpha,\bm{\phi},\bm{A})=[\mathrm{det}(\mathbb{I}_{2M}-\alpha\bm{\Omega}\bm{A})]^{-1/2}$. By noticing that $\mathbb{I}_{2M}-\alpha\bm{\Omega}\bm{A}=\exp(\log(\mathbb{I}_{2M}-\alpha\bm{\Omega}\bm{A}))$, and using the relation $\mathrm{det}[\exp(\bm{A})]=\exp[\mathrm{tr}(\bm{A})]$, we can write
    \begin{equation}
        q(\alpha,\bm{\phi},\bm{A})=\exp\left[-\frac{1}{2}\mathrm{tr}\left(\log(\mathbb{I}_{2M}-\alpha\bm{\Omega}\bm{A})\right)\right].
        \label{eq:q_function_log_form}
    \end{equation}
    We now use $\log(1+x)=\sum_{l=1}^\infty(-1)^{l+1}x^l/l$, which converges for $|x|<1$, to write
    \begin{equation}
        q(\alpha,\bm{\phi},\bm{A})=\exp\left[-\frac{1}{2}\mathrm{tr}\left(\sum_{l=1}^{\infty}-\frac{\alpha^l}{l}(\bm{\Omega}\bm{A})^{l}\right)\right].      
        \label{eq:q_function_log_form_2}
    \end{equation}
    This expression directly yields to Eq.~\eqref{eq:q_function_expansion_1}. We will focus now on the series inside the trace.
    \begin{proposition}
        \label{eq:convergence}
        The power series 
        \begin{equation}
            \sum_{l=1}^{\infty}\frac{\alpha^l}{l}(\bm{\Omega}\bm{A})^{l}
            \label{eq:series_convergence}
        \end{equation}
        converges for every GBS model $\bm{A}$ of the form of Eq.~\eqref{eq:gbs_model_general_structure} provided that $|\alpha|<1$.
    \end{proposition}

    \begin{proof}
        Let us take the spectral norm of the power series and write
        \begin{align}
            \left\|\sum_{l=1}^{\infty}\frac{\alpha^l}{l}(\bm{\Omega\bm{A}})^{l}\right\|_2&\leq\sum_{l=1}^{\infty}\left\|\frac{\alpha^l}{l}(\bm{\Omega}\bm{A})^{l}\right\|_2
            =\sum_{l=1}^{\infty}\frac{|\alpha|^l}{l}\left\|(\bm{\Omega}\bm{A})^{l}\right\|_2
            \leq\sum_{l=1}^{\infty}\frac{|\alpha|^l}{l}\left\|\bm{\Omega}\bm{A}\right\|_2^{l}=\sum_{l=1}^{\infty}\frac{|\alpha|^l}{l}\left\|\bm{W}\bm{X}\bm{W}\bm{A}\right\|_2^{l}.
        \end{align}
        Notice that $\bm{W}\bm{X}=\bm{X}\bm{W}$. Moreover $\|\bm{W}^2\|_2=1$ since $\bm{W}$ is unitary. We therefore have
        \begin{align}
            \left\|\sum_{l=1}^{\infty}\frac{\alpha^l}{l}(\bm{\Omega\bm{A}})^{l}\right\|_2&\leq\sum_{l=1}^{\infty}\frac{|\alpha|^l}{l}\left\|\bm{W}^2\bm{X}\bm{A}\right\|_2^{l}
            \leq\sum_{l=1}^{\infty}\frac{|\alpha|^l}{l}\left\|\bm{X}\bm{A}\right\|_2^{l}.
        \end{align}

        According to Eq.~\eqref{eq:gbs_model_general_structure}, we can recast $\bm{X}\bm{A}$ as
        \begin{equation}
            \bm{X}\bm{A}=\bm{F}\bm{G}\bm{F}^{\dagger},
        \quad
        \bm{F}=
        \begin{pmatrix}
            \bm{U}^*&\bm{0}\\
            \bm{0}&\bm{U}
        \end{pmatrix},\quad
        \bm{G}=\begin{pmatrix}
            \bm{\mu}&\bm{\lambda}\\
            \bm{\lambda}&\bm{\mu}
        \end{pmatrix}\, , 
        \end{equation}
        where $\bm{F}$ is also unitary. This allows us to see that $\|\bm{X}\bm{A}\|_2=\|\bm{G}\|_2$. Since $\bm{G}$ is real and symmetric, its singular values are the absolute values of its eigenvalues. These, in turn, can be proven to be
        \begin{equation}
            \left\{1-\frac{2}{1+2\sigma_x^{(k)}/\hbar},\,\,1-\frac{2}{1+2\sigma_p^{(k)}/\hbar}\right\}.    
        \end{equation}
        Because $\sigma_x^{(k)},\sigma_p^{(k)}>0$ for all $k$, these eigenvalues lie in the interval $(-1,1)$, which implies that $\|\bm{G}\|_2<1$. We conclude that
        \begin{equation}
            \left\|\sum_{l=1}^{\infty}\frac{\alpha^l}{l}(\bm{\Omega\bm{A}})^{l}\right\|_2<\sum_{l=1}^{\infty}\frac{|\alpha|^l}{l}.    
        \end{equation}

        For $|\alpha|<1$, $\sum_{l=1}^{\infty}|\alpha|^l/l$ converges to $-\log(1-|\alpha|)$. Therefore, the series in Eq.~\eqref{eq:series_convergence} converges for $|\alpha|<1$.
    \end{proof}

    The specific case of the ideal squeezed state model is obtained for $2\sigma_x^{(k)}/\hbar=e^{2r_k}$, $2\sigma_p^{(k)}/\hbar=e^{-2r_k}$, with $r_k$ the squeezing parameter at mode $k$, for $1\leq k\leq R$ and $\sigma_x^{(k)}=\sigma_p^{(k)}=\hbar/2$ otherwise. This implies that $\mu_k=0$ for all $k$, while $\lambda_k=\tanh(r_k)$ for $1\leq k\leq R$ and zero otherwise.

    For single-mode squeezed states passing through single-mode loss channels before entering the interferometer (like those used in the transmission loss model of Sec.~\ref{sec:no_vacuum}) we have $2\sigma_x^{(k)}/\hbar=\eta_k e^{2r_k}+(1-\eta_k)$ and $2\sigma_p^{(k)}/\hbar=\eta_ke^{-2r_k}+(1-\eta_k)$, with $\eta_k$ the transmission parameter of the loss channel at mode $k$, for $1\leq k\leq R$ and $\sigma_x^{(k)}=\sigma_p^{(k)}=\hbar/2$ otherwise. This implies that 
    \begin{equation}
        \mu_k=\frac{\eta_k(1-\eta_k)\sinh^2(r_k)}{1+\eta_k(2-\eta_k)\sinh^2(r_k)},\quad\lambda_k=\frac{\eta_k\sinh(r_k)\cosh(r_k)}{1+\eta_k(2-\eta_k)\sinh^2(r_k)}
        \label{eq:lossy_sqz_model_especs}
    \end{equation}
    for $1\leq k\leq R$, while $\mu_k=\lambda_k=0$ otherwise. Notice that we can recover the ideal squeezed state model from these equations by setting $\eta_k=1$ for all $k$.

    Finally, the squashed state model (see also Sec.~\ref{sec:no_vacuum}) is obtained for $2\sigma_x^{(k)}/\hbar=1+4\bar{n}_k$, $2\sigma_p^{(k)}/\hbar=1$, with $\bar{n}_k$ the mean number of photons at mode $k$, for $1\leq k\leq R$ and $\sigma_x^{(k)}=\sigma_p^{(k)}=\hbar/2$ otherwise. This implies that $\mu_k=\lambda_k=\bar{n}_k/(1+2\bar{n}_k)$ for $1\leq k\leq R$, while $\mu_k=\lambda_k=0$ otherwise.
    
    \section{\label{app:index_structure}Index structure of the LXE}

    In this appendix we will show how to obtain Eq.~\eqref{eq:lxe_sqz_cycle_U_polynomial} and we will motivate the definition of the permutations $\Omega_{\bm{k}}$. Let us start by recalling that $\mathrm{LXE}\left(\bm{A}_{\mathrm{sqz}},\bm{A}_{\mathrm{sqz}};2N\right)$ is expressed as an integral of $Z_N[\bm{u}(\bm{\phi},\bm{U})]Z_N[\bm{u}(-\bm{\phi},\bm{U})]$ with respect to $d\bm{\phi}$, where
    \begin{equation}
        Z_N[\bm{u}(\bm{\phi},\bm{U})]=\sum_{\bm{k}^{(c)}}\,\prod_{a=1}^N\frac{1}{k_a!a^{k_a}}\prod_{a=1}^Nu_a^{k_a},
        \label{eq:Z_term}
    \end{equation}
    $u_a=\frac{1}{2}\mathrm{tr}\left[(\bm{D}^2\bm{V}\bm{D}^2\bm{V}^*)^a\right]$, $\bm{V}=\bm{U}\bm{\zeta}\bm{U}^{\mathrm{T}}$, 
    and $u_a$ depends on $\bm{\phi}$ through $\bm{D}=\bm{D(\phi)}$. In what follows, we will express $Z_N[\bm{u}(\bm{\phi},\bm{U})]$ in terms of the entries of $\bm{V}$.

    Notice that the entries of $\bm{D}^2\bm{V}\bm{D}^2\bm{V}^*$ can be written as
    \begin{equation}
        (\bm{D}^2\bm{V}\bm{D}^2\bm{V}^*)_{j_1,j_3}=\sum_{j_2=1}^Me^{2i(\phi_{j_1}+\phi_{j_2})}V_{j_1,j_2}V_{j_2,j_3}^*,    
    \end{equation}
    with $\{j_k\}$ denoting a set of dummy indices. Using this expression, we can readily see that
    \begin{align}
        [(\bm{D}^2\bm{V}\bm{D}^2\bm{V}^*)^2]_{j_1,j_5} =\sum_{j_2,j_3,j_4}e^{2i(\phi_{j_1}+\phi_{j_2}+\phi_{j_3}+\phi_{j_4})}V_{j_1,j_2}V_{j_3,j_4}V_{j_2,j_3}^*V_{j_4,j_5}^*,
    \end{align}
    or, for a general power $l$,
    \begin{align}
        [(\bm{D}^2\bm{V}\bm{D}^2\bm{V}^*&)^l]_{j_1,j_{2l+1}}=\sum_{j_2,\dots,j_{2l}}e^{2i(\phi_{j_1}+\cdots+\phi_{j_{2l}})}V_{j_1,j_2} \cdots V_{j_{2l-1},j_{2l}}V_{j_2,j_3}^* \cdots V_{j_{2l},j_{2l+1}}^*.
        \label{eq:sqz_index_structure_mat}
    \end{align}
    Taking the trace of $(\bm{D}^2\bm{V}\bm{D}^2\bm{V}^*)^l$ we obtain
    \begin{align}
        \mathrm{tr}[(\bm{D}^2\bm{V}\bm{D}^2\bm{V}^*&)^l]=\sum_{j_1,\dots,j_{2l}}e^{2i(\phi_{j_1}+\cdots+\phi_{j_{2l}})}V_{j_1,j_2} \cdots V_{j_{2l-1},j_{2l}}V_{j_2,j_3}^* \cdots V_{j_{2l},j_{1}}^*.\label{eq:sqz_index_structure_mat_trace}
    \end{align}
    
    Let us gather all the dummy indices in the sequence $\bm{j}_{1,2l}=(j_1,\dots,j_{2l})$. Notice that the subscripts in $\bm{j}$ indicate the \textit{labels} of the first and last dummy indices. We define $\omega_l \in S_{2l}$ as the permutation that transforms $(j_1,j_2,\dots,j_{2l-1}, j_{2l})$ into $(j_2,j_3,\dots,j_{2l},j_1)$:
    $\omega_l(\bm{j}_{1,2l})=(j_2,j_3,\dots,j_{2l},j_1)$. Furthermore, let us define $V\left[\bm{j}_{1,2l}\right]\equiv V_{j_1,j_2} \cdots V_{j_{2l-1},j_{2l}}$ and $V^*\left[\omega_l(\bm{j}_{1,2l})\right]\equiv V_{j_2,j_3}^* \cdots V_{j_{2l},j_{1}}^*$. Then, we can recast Eq.~\eqref{eq:sqz_index_structure_mat_trace} into 
    \begin{align}
        \mathrm{tr}[(\bm{D}^2\bm{V}\bm{D}^2\bm{V}^*)^l]=\!\!&\sum_{j_1,\dots,j_{2l}}\!\!E(\bm{j}_{1,2l})V\left[\bm{j}_{1,2l}\right]V^*\left[\omega_l(\bm{j}_{1,2l})\right],\label{eq:sqz_index_structure_mat_trace_recast}
    \end{align}
    where $E(\bm{j}_{1,2l})=\exp[2i(\phi_{j_1}+\cdots+\phi_{j_{2l}})]$.

    Defining a second sequence of dummy indices $\bm{j}'_{1,2l}=(j_1',\dots,j_{2l}')$, we can readily see that $V[\bm{j}_{1,2l}]V[\bm{j}'_{1,2l}]=V[\bm{j}_{1,2l}\oplus\bm{j}'_{1,2l}]$, with analogous relations holding for $V^*[\cdot]$ and $E[\cdot]$. These properties come in handy when computing powers of $\mathrm{tr}[(\bm{D}^2\bm{V}\bm{D}^2\bm{V}^*)^l]$. Indeed, consider the expression
    \begin{align}
        &\left(\mathrm{tr}[(\bm{D}^2\bm{V}\bm{D}^2\bm{V}^*)^l]\right)^2=\!\!\sum_{j_1,\dots,j_{2l}}\sum_{j_1',\dots,j_{2l}'}\!\!E(\bm{j}_{1,2l})E(\bm{j}'_{1,2l})V[\bm{j}_{1,2l}]V[\bm{j}'_{1,2l}]V^*[\omega_l(\bm{j}_{1,2l})]V^*[\omega_l(\bm{j}'_{1,2l})].
    \end{align}
    Renaming each dummy index $j_k'$ as $j_k'\rightarrow j_{2l+k}$, and using the direct sum properties of $V$, $V^*$ and $E$, we can write 
    \begin{align}
        \left(\mathrm{tr}[(\bm{D}^2\bm{V}\bm{D}^2\bm{V}^*)^l]\right)^2=\!\!&\sum_{j_1,\dots,j_{4l}}\!\!E(\bm{j}_{1,4l})V[\bm{j}_{1,4l}]V^*[\omega_l(\bm{j}_{1,2l})\oplus \omega_l(\bm{j}_{2l+1,4l})].
    \end{align}
    Applying this same procedure a given number of times, say $k_l$ times, we obtain  
    \begin{align}
            &\left(\mathrm{tr}[(\bm{D}^2\bm{V}\bm{D}^2\bm{V}^*)^l]\right)^{k_l}=\!\!\!\!\sum_{j_1,\dots,j_{2lk_l}}\!\!\!\!E(\bm{j}_{1,2lk_l})V[\bm{j}_{1,2lk_l}]V^*\!\!\left[\bigoplus_{p=1}^{k_l}\omega_l(\bm{j}_{2l(p-1)+1,2lp})\right].
        \label{eq:power_of_power_traces}
    \end{align}

    From the definition of $u_a$, we can see that Eq.~(\ref{eq:power_of_power_traces}) allows us
    to directly express $u_a^{k_a}$ in terms of
    the entries of $\bm{V}$ for an arbitrary
    value of $a$. The next step is to use
    Eq.~\eqref{eq:power_of_power_traces} to
    compute the product $\prod_{a=1}^Nu_a^{k_a}$.
    The strategy is completely analogous to the
    one we used to obtain
    Eq.~\eqref{eq:power_of_power_traces}: we
    define a set of primed dummy indices and make
    the product of two different sums, then rename the primed indices, and
    finally make the direct sum of the
    sequences of indices involved.
    
    Consider, for example, the product $u_1^{k_1}u_2^{k_2}$:
    \begin{align}
        u_1^{k_1}u_2^{k_2}&=\frac{1}{2^{k_1+k_2}}\sum_{j_1,\dots,j_{2k_1}}\!\!\!E(\bm{j}_{1,2k_1})V[\bm{j}_{1,2k_1}] V^*\!\!\left[\bigoplus_{p=1}^{k_1}\omega_1(\bm{j}_{2(p-1)+1,2p})\right]\sum_{j_1',\dots,j_{4k_2}'}\!\!\!E(\bm{j}'_{1,4k_2})V[\bm{j}'_{1,4k_2}V^*\!\!\left[\bigoplus_{p=1}^{k_2}\omega_2(\bm{j}'_{4(p-1)+1,4p})\right].
        \label{eq:u1k1ku2k}
    \end{align}
    Changing the dummy indices $j_b'$ as $j_b'\rightarrow j_{2k_1+b}$, we can rewrite Eq.~(\ref{eq:u1k1ku2k}) as
    \begin{align}
        &u_1^{k_1}u_2^{k_2}=\frac{1}{2^{k_1+k_2}}\sum_{j_1,\dots,j_{2k_1+4k_2}}\!\!\!E(\bm{j}_{1,2k_1+4k_2})V[\bm{j}_{1,2k_1+4k_2}]V^*\!\!\left[\bigoplus_{p=1}^{k_1}\omega_1(\bm{j}_{2(p-1)+1,2p})\bigoplus_{p=1}^{k_2}\omega_2(\bm{j}_{2k_1 + 4(p-1)+1,2k_1+4p})\right].
    \end{align}

    Applying this same process for the remaining $\{u_a\}$'s, we obtain the following general expression: 
    \begin{align}
        \prod_{a=1}^Nu_a^{k_a}=\frac{1}{2^{k_1+\dots+k_N}}\sum_{j_1,\dots,j_{2N}}\!\!\!E(\bm{j}_{1,2N})V[\bm{j}_{1,2N}]V^*\!\!\left[\bigoplus_{a=1}^{N}\bigoplus_{p=1}^{k_a}\omega_a(\bm{j}_{2v_{a-1} + 2a(p-1)+1,2v_{a-1}+2ap})\right],
        \label{eq:index_u_prod}
    \end{align}
    where we have used $k_1 + 2k_2+\dots+Nk_N=N$ and defined $v_a=\sum_{p=1}^apk_p$, $v_0\equiv0$.

    Write $\bm{j}\equiv\bm{j}_{1,2N}=(j_1,\dots,j_{2N})$. Moreover, let $\Omega_{\bm{k}}\in S_{2N}$ act on $\bm{j}$ as
    \begin{align}
            \Omega_{\bm{k}}(\bm{j})=\Omega_{\bm{k}}[(j_1,\dots,j_{2N})]&=\bigoplus_{a=1}^{N}\bigoplus_{p=1}^{k_a}\omega_a(\bm{j}_{2v_{a-1} + 2a(p-1)+1,2v_{a-1}+2ap})\nonumber\\
            &=\bigoplus_{a=1}^N\bigoplus_{p=1}^{k_a}\omega_{a}[(j_{2v_{a-1}+2a(p-1)+1},\dots,j_{2v_{a-1}+2ap})].
    \end{align}
    Then, we can express Eq.~\eqref{eq:index_u_prod} in the form
    \begin{align}
        \begin{split}
            &\prod_{a=1}^Nu_a^{k_a}=\frac{1}{2^{k_1+\dots+k_N}}\sum_{\bm{j}}E(\bm{j})V[\bm{j}]V^*[\Omega_{\bm{k}}(\bm{j})],
        \end{split}
        \label{eq:index_u_prod_2}
    \end{align}
    where $\sum_{\bm{j}}\equiv\sum_{j_1=1}^M\cdots\sum_{j_{2N}=1}^M$. This result allows us to readily write $Z_N[\bm{u}(\bm{\phi},\bm{U})]$ in the form
    \begin{equation}
        Z_N[\bm{u}(\bm{\phi},\bm{U})]=\!\sum_{\bm{k}^{(c)}}\prod_{a=1}^N\frac{1}{k_a!(2a)^{k_a}}\sum_{\bm{j}}E(\bm{j})V[\bm{j}]V^*[\Omega_{\bm{k}}(\bm{j})].
        \label{eq:Z_term_2}
    \end{equation}
    Going one step further, we can obtain the expansion of $Z_N[\bm{u}(\bm{\phi},\bm{U})]Z_N[\bm{u}(-\bm{\phi},\bm{U})]$ in terms of $V$ and $V^*$:
    \begin{align}
            Z_N[\bm{u}(\bm{\phi},\bm{U})]Z_N[\bm{u}(-\bm{\phi},\bm{U})]=\sum_{\bm{k}^{(c)}, \bm{l}^{(c)}}\prod_{a=1}^N\frac{1}{k_a!l_a!(2a)^{k_a+l_a}}\sum_{\bm{j},\bm{j}'}E(\bm{j})E^*(\bm{j}')V[\bm{j}\oplus\bm{j}']V^*[\Omega_{\bm{k}}(\bm{j})\oplus\Omega_{\bm{l}}(\bm{j}')].
        \label{eq:ZZ_term}
    \end{align}

    Recall now that Eq.~\eqref{eq:monomial_V_function_def} allows us to write monomials in the entries of $\bm{V}$ and $\bm{V}^*$ as polynomials in the entries of $\bm{U}$ and $\bm{U}^*$: 
    \begin{align}
        V[\bm{g}]V^*[\bm{h}]=V_{g_1,g_2}\cdots V_{g_{2l-1},g_{2l}}V^*_{h_1,h_2}\cdots V^*_{h_{2m-1},h_{2m}}
        =\sum_{\bm{\mu},\bm{\nu}}\zeta_{\bm{\mu}}\zeta_{\bm{\nu}}\,\mathcal{U}\left(\bm{g},\bm{{\bar{\mu}}}\,\vert\,\bm{h},\bm{\bar{\nu}}\right),
    \end{align}
    where the dummy indices in $\bm{\mu}=(\mu_1,\dots,\mu_l)$ and $\bm{\nu}=(\nu_1,\dots,\nu_m)$ take values in $\{1,\dots, M\}$, and $\bm{\bar{\mu}}=(\mu_1,\mu_1,\dots,\mu_l,\mu_l)$, $\bm{\bar{\nu}}=(\nu_1,\nu_1,\dots,\nu_m,\nu_m)$. As per Eq.~(\ref{eq:zetadefine}), $\zeta_{\bm{\mu}}=\zeta_{\mu_1}\cdots\zeta_{\mu_l}$, with $\{\zeta_{k}\}$ the diagonal entries of $\bm{\zeta}$; and
    \begin{align}
        \mathcal{U}\left(\bm{g},\bm{{\bar{\mu}}}\,\vert\,\bm{h},\bm{\bar{\nu}}\right)&=U_{g_1,\mu_1}U_{g_2,\mu_1}\cdots U_{g_{2l-1},\mu_l}U_{g_{2l},\mu_l}U^*_{h_1,\nu_1}U^*_{h_2,\nu_1}\cdots U^*_{h_{2m-1},\nu_m}U^*_{h_{2m},\nu_m}.   
    \end{align}

    We can now recast Eq.~\eqref{eq:ZZ_term} as 
    \begin{align}
            Z_N[\bm{u}(\bm{\phi},\bm{U})]Z_N[\bm{u}(-\bm{\phi},\bm{U})]
            =\sum_{\bm{k}^{(c)}, \bm{l}^{(c)}}\prod_{a=1}^N\frac{1}{k_a!l_a!(2a)^{k_a+l_a}}\sum_{\bm{j},\bm{j}'}E(\bm{j})E^*(\bm{j}')
            \sum_{\bm{\mu},\bm{\nu}}\zeta_{\bm{\mu}}\zeta_{\bm{\nu}}\,\mathcal{U}\left(\bm{j}\oplus\bm{j}',\bm{\bar{\mu}}\,\vert\,\Omega_{\bm{k}}(\bm{j})\oplus\Omega_{\bm{l}}(\bm{j}'),\bm{\bar{\nu}}\right).
        \label{eq:ZZ_term_2}
    \end{align}
    Integrating the previous equation with respect to $d\bm{\phi}$, multiplying by $\binom{\frac{R}{2}+N-1}{N}^{-2}$, and defining 
    \begin{align}
        I(\bm{j},\bm{j}')=\frac{1}{(2\pi)^M}\int_0^{2\pi}E(\bm{j})E^*(\bm{j}')\,d\bm{\phi}=\frac{1}{(2\pi)^M}\int_0^{2\pi}\exp\left[\sum_{m\in \bm{j},n\in \bm{j}'}\!\!\!\!2i\,(\phi_m-\phi_n)\right]d\bm{\phi},
    \end{align}
    we obtain the expression of $\mathrm{LXE}\left(\bm{A}_{\mathrm{sqz}},\bm{A}_{\mathrm{sqz}};2N\right)$ presented in Eq.~\eqref{eq:lxe_sqz_cycle_U_polynomial}.
    
    \section{\label{app:integral}Integral over phases}
    Consider the integral
    \begin{align}
        I(\bm{j},\bm{j}')&=\frac{1}{(2\pi)^M}\int_0^{2\pi}\exp\left[\sum_{m\in \bm{j},n\in \bm{j}'}\!\!\!\!2i\,(\phi_m-\phi_n)\right]d\bm{\phi},
    \end{align}
    where we see $\bm{j}$, $\bm{j}'$ as sequences of variables taking values in $\{1,\dots M\}$. Recall that $\phi_k\in [0,2\pi]$ for all $k$ and $d\bm{\phi}=d\phi_1\cdots d\phi_M$. 

    Let us turn our attention to the term inside the exponential:
    \begin{equation}
        2i\sum_{m\in\bm{j}}\phi_m-2i\sum_{n\in\bm{j}'}\phi_n.    
    \end{equation}
    As mentioned in the main text, if this sum is non-zero, there must be at least one $p\in\bm{j}$ or $p\in\bm{j}'$ such that
    \begin{equation}
        2i\sum_{m\in\bm{j}}\phi_m-2i\sum_{n\in\bm{j}'}\phi_n=2i\, z \phi_p + (\text{other terms}),    
    \end{equation}
    where $z$ is a non-zero integer. This expression represents the fact that the sum is unbalanced. Since the exponential can be factorized, we can focus on the integral involving $\phi_p$ only:
    \begin{align}
        \int_0^{2\pi}e^{2iz\phi_p}d\phi_p&=\left.\frac{1}{2iz}e^{2iz\phi_p}\right|_0^{2\pi}=\frac{1}{2iz}\left[(e^{2i\pi})^{2z}-1\right]=\frac{1}{2iz}(1-1)=0.
    \end{align}
    This means that the entirety of the integral is zero whenever the sum inside the exponential is different from zero. If the exponent is zero, we can readily see that the integral is equal to one:
    \begin{equation}
        \frac{1}{(2\pi)^M}\int_0^{2\pi}\exp(0)d\bm{\phi}=\frac{1}{(2\pi)^M}(2\pi)^M=1.   
    \end{equation}
    
    When considering the sum inside the exponential as a function of the variables $\bm{j}$ and $\bm{j}'$, we can recognize that we can make it vanish whenever $\bm{j}'$ is a permutation of $\bm{j}$ (this ensures the balance in the summations). This is the case even if we find that some of the $\{j_k\}$ have repeated values, which is allowed given the fact that all of them take values in the same set. The task now is to find a function of $\bm{j}$ and $\bm{j}'$ that is unity for any event that makes the sum in the exponential equal to zero, and that vanishes identically otherwise. The integral will then be equal to this function.

    Let us start by considering that all the indices in $\bm{j}$ take different values. Notice that this makes the indices in $\bm{j}'$ take different values as well. Consider a permutation $\sigma \in S_{2N}$. We can see that the function $F_\sigma(\bm{j}',\bm{j})=\delta_{j_1',\sigma(j_1)}\cdots\delta_{j_{2N}',\sigma(j_{2N})}$, with $\delta_{j,k}$ the usual Kronecker delta, vanishes whenever $\bm{j}'\neq \sigma(\bm{j})$, and is equal to one otherwise. Summing over all the permutations in $S_{2N}$, we obtain a function that is equal to one whenever $\bm{j}'$ is a permutation, any permutation, of $\bm{j}$, and vanishes otherwise:
    \begin{equation}
        F(\bm{j}',\bm{j})=\sum_{\sigma\in S_{2N}}F_{\sigma}(\bm{j}',\bm{j})=\sum_{\sigma\in S_{2N}}\prod_{a=1}^{2N}\delta_{j'_a,\sigma(j_a)}.
        \label{eq:F_delta_function}
    \end{equation}
    This is the value of $I(\bm{j},\bm{j}')$ \textit{when all the indices in $\bm{j}$ take different values}.

    When some of the $\{j_k\}$ have the same values as others, the situation becomes slightly more involved. To address this problem, it is convenient to first translate the phrase \textit{``there are some indices having the same value as others''} into a sequence of repeated dummy indices taken from $\bm{j}$. We can do this by considering the following procedure. First, identify all the indices in $\bm{j}$ that have the same value and gather them in individual sets, one for each different value that the indices take. The result of this step is a set of subsets of $\bm{j}$. For instance, consider $2N=6$, so $\bm{j}=(j_1,j_2,j_3,j_4,j_5,j_6)$, and let us suppose that we are in a situation where $j_1=j_4=j_3$, $j_2=j_5$ and $j_1\neq j_2\neq j_6$. After gathering the indices that have the same values, we obtain the set $\{\{j_1,j_3,j_4\},\{j_2,j_5\},\{j_6\}\}$. We may recognize that this procedure is equivalent to finding a \textit{partition} of $\bm{j}$ (i.e., a collection of non-empty, mutually disjoint subsets of $\bm{j}$, usually called blocks, whose union is equal to $\bm{j}$), where all the indices within each block of the partition take the same value. Next, we select one index within each block and replace all the indices in $\bm{j}$ that belong to the same block by this \textit{representative index}. In our example, let us choose $j_1$, $j_2$, and $j_6$ as representatives. After the replacement of indices belonging to the same block, we obtain the sequence $\bm{g}=(j_1,j_2,j_1,j_1,j_2,j_6)$. This procedure can be applied for any situation where we have repeated indices. Notice that selecting different representative indices leads to different $\bm{g}$. However, no matter the choice of representatives, the resulting $\bm{g}$ represent the same \textit{situation}.

    Having found a sequence $\bm{g}$ associated to a specific case of repeated indices, let us return to finding a function that vanishes whenever $\bm{j}'$ is not a \textit{unique permutation} of $\bm{g}$, and is unity otherwise. We could consider the function $F(\bm{j}',\bm{g})$, but we can readily observe that the sum $\sum_{\sigma\in S_{2N}}$ is overcounting the different permutations of $\bm{g}$, i.e., there are multiple permutations in $S_{2N}$ that when applied to $\bm{g}$ lead to the same result. This implies that if $\bm{j}'\neq\sigma(\bm{j})$, we obtain $F(\bm{j}',\bm{g})=0$, but when $\bm{j}'=\sigma(\bm{j})$, we generally do not obtain $F(\bm{j}',\bm{g})=1$. Fortunately, we can solve this issue in a simple manner: we normalize $F(\bm{j}',\bm{g})$ using the number of times the unique permutations of $\bm{g}$ are being overcounted. This number is the same for every distinct permutation, and depends only on the \textit{multiplicities} of the indices appearing in $\bm{g}$.
    Let $\{j_k\}_{k\in \Theta}$, where $\Theta\subset \{1,\dots,2N\}$, be the different dummy indices that appear in $\bm{g}$, and let each $j_k$ appear $t_k$ times within the sequence. Then the number of times each unique permutation of $\bm{g}$ is being overcounted by $\sum_{\sigma\in S_{2N}}$ is $\prod_{k\in\Theta}t_k!$. 
    
    On this account, we can see that the function
    \begin{equation}
        \left(\prod_{k\in\Theta}\frac{1}{t_k!}\right) F(\bm{j}',\bm{g})=\prod_{k\in\Theta}\frac{1}{t_k!}\sum_{\sigma\in S_{2N}}\prod_{a=1}^{2N}\delta_{j'_a,\sigma(g_a)}
        \label{eq:F_delta_function_repeated}
    \end{equation}
    vanishes whenever $\bm{j}'$ is not a unique permutation of $\bm{g}$, and is equal to one otherwise. This is the result of $I(\bm{j},\bm{j}')$ \textit{when we have a situation of repeated indices represented by $\bm{g}$}.

    We are now two steps away from finding a general expression for $I(\bm{j},\bm{j}')$. The first of these consists in finding a systematic way of computing all the possible $\bm{g}$, i.e., a systematic way of identifying all the events with repeated indices. We already gave a hint of how to do this when we explained how to construct $\bm{g}$; the key is to use the partitions of $\bm{j}$. Indeed, since the procedure of grouping indices in $\bm{j}$ that have the same values naturally leads to a partition of $\bm{j}$, we can invert the process and assign to each possible partition a sequence $\bm{g}$ representing an event with repeated indices. The procedure is the following: consider a partition $\Lambda$ of $\bm{j}$. We will think of all the indices within each block $\lambda\in \Lambda$ as having the same value. The corresponding $\bm{g}$ is constructed by choosing a representative index $j_\lambda$ for each $\lambda$, and then replacing all the $j_k$ in $\bm{j}$ that belong to the same $\lambda$ by the corresponding $j_\lambda$. For clarity, we will introduce the notation $\bm{g}\equiv\bm{j}[\Lambda,\{j_\lambda\}]$.

    Now, note that the degeneracy of each $\{j_\lambda\}$ is equal to the length, $|\lambda|$, of each block (i.e. the number of elements in each block). This allows us to write the normalization factor of $F(\bm{j}',\bm{j}[\Lambda,\{j_\lambda\}])$
    as $\Lambda!=\prod_{\lambda \in \Lambda}|\lambda|!$. Moreover, the event where all the indices in $\bm{j}$ take different values is represented by the partition $\Lambda_0=\{\{j_1\},\dots,\{j_{2N}\}\}$ (one index per block). For this partition $\Lambda_0!=1$.

    Since each partition $\Lambda$ of $\bm{j}$ gives us a unique way of grouping the indices $\{j_k\}$ in different blocks, each $\Lambda$ leads to a $\bm{j}[\Lambda,\{j_\lambda\}]$ representing a unique event with repeated indices. Note, however, that each $\Lambda$ is associated to several $\bm{j}[\Lambda,\{j_\lambda\}]$ differing only on the choice of representative indices $\{j_\lambda\}$. Nevertheless, all of these sequences represent the same unique situation where there are indices taking the same value as others.

    The final step for finding an expression for $I(\bm{j},\bm{j}')$ consists in constructing a function of $\bm{j}$ and $\bm{j}[\Lambda,\{j_\lambda\}]$ that is equal to one only \textit{when we have a situation of repeated indices represented by $\bm{j}[\Lambda,\{j_\lambda\}]$}, and vanishes otherwise. We need this in order to single out the contributions of each unique $\Lambda$ to the integral. The event that we need to identify can be equivalently written as \textit{``when all the indices in $\lambda\in \Lambda$ are equal to $j_\lambda$, and all the $\{j_\lambda\}$ are different from each other''}. This phrase can be readily written in terms of Kronecker deltas as:
    \begin{equation}
        \left[\prod_{\lambda\in\Lambda}\prod_{f\in\lambda}\delta_{j_\lambda,f}\right]\prod_{(\lambda\neq\mu)\in\Lambda}(1-\delta_{j_\lambda,j_\mu}).
        \label{eq:indentifying_partitions}
    \end{equation}

    We may now bring together Eqs.~\eqref{eq:F_delta_function}to~\eqref{eq:indentifying_partitions}; sum over all the possible ways of identifying events with repeated indices, i.e. sum over the set of all partitions of $\bm{j}$, $\mathcal{Q}[\bm{j}]$; and finally write $I(\bm{j},\bm{j}')$ as (see Eq.~\eqref{eq:phases_integral_result})
    \begin{align}
        I(\bm{j},\bm{j'})=\sum_{\Lambda\in\mathcal{Q}[\bm{j}]}\frac{1}{\Lambda!}F(\bm{j}',\bm{j}[\Lambda,\{j_\lambda\}])\left[\prod_{\lambda\in\Lambda}\prod_{f\in\lambda}\delta_{j_\lambda,f}\right]\prod_{(\lambda\neq\mu)\in\Lambda}(1-\delta_{j_\lambda,j_\mu}).
    \end{align}

    It is worth mentioning that in the subscript $\Lambda \in \mathcal{Q}[\bm{j}]$, $\bm{j}$ should be viewed as a collection of fixed indices, or objects, whose sole purpose is to determine all the possible partitions of a set with $2N$ elements. Thus, in this subscript, we must not replace any $j_k$ by one of its possible values $\{1,\dots,M\}$. The reason behind this is that the sum over partitions was included only as a way to index a series of events concerning the variables $\{j_k\}$. 

    To finish this section, let us turn our attention to the reorganization of sums that leads to Eq.~\eqref{eq:lxe_sqz_after_integral}. Let $h(\bm{j},\bm{j}')$ be an arbitrary function of $\bm{j}$ and $\bm{j}'$, and consider the expression
    \begin{align}
        \sum_{\bm{j},\bm{j}'}h(\bm{j},\bm{j}')I(\bm{j},\bm{j}')=\sum_{\bm{j},\bm{j}'}\sum_{\Lambda\in\mathcal{Q}[\bm{j}]}\frac{1}{\Lambda!}h(\bm{j},\bm{j}')F(\bm{j}',\bm{j}[\Lambda,\{j_\lambda\}])\left[\prod_{\lambda\in\Lambda}\prod_{f\in\lambda}\delta_{j_\lambda,f}\right]\prod_{(\lambda\neq\mu)\in\Lambda}(1-\delta_{j_\lambda,j_\mu}),
    \end{align}
    where we recall that $\sum_{\bm{j}}=\sum_{j_1=1}^M\cdots\sum_{j_{2N}=1}^M$.

    Choose an arbitrary partition $\Lambda\in\mathcal{Q}[\bm{j}]$, and let us focus on the term 
    \begin{align}
        T=\sum_{\bm{j},\bm{j}'}h(\bm{j},\bm{j}')F(\bm{j}',\bm{j}[\Lambda,\{j_\lambda\}]) \left[\prod_{\lambda\in\Lambda}\prod_{f\in\lambda}\delta_{j_\lambda,f}\right]\prod_{(\lambda\neq\mu)\in\Lambda}(1-\delta_{j_\lambda,j_\mu}).
    \end{align}
    Following the definition of $F(\bm{j}',\bm{j}[\Lambda,\{j_\lambda\}])$ in terms of Kronecker deltas, we can replace the sum over the indices $\bm{j}'$ by a sum over permutations $\sigma\in S_{2N}$ and, moreover, we can make the transformation $\bm{j}'\rightarrow \sigma(\bm{j}[\Lambda,\{j_\lambda\}])$:
    \begin{align}
        T=\sum_{\bm{j}}\sum_{\sigma \in S_{2N}}h(\bm{j},\sigma(\bm{j}[\Lambda,\{j_\lambda\}])) \left[\prod_{\lambda\in\Lambda}\prod_{f\in\lambda}\delta_{j_\lambda,f}\right]\prod_{(\lambda\neq\mu)\in\Lambda}(1-\delta_{j_\lambda,j_\mu}).
    \end{align}

    The term $\prod_{\lambda\in\Lambda}\prod_{f\in\lambda}\delta_{j_\lambda,f}$ allows us to make the transformation $\bm{j}\rightarrow \bm{j}[\Lambda,\{j_\lambda\}]$, and to turn the sum over $\bm{j}$ into a sum only over the values of the representative indices $\{j_\lambda\}$. We can write
    \begin{align}
        T=\sum_{\{j_\lambda\}}\sum_{\sigma \in S_{2N}}&h(\bm{j}[\Lambda,\{j_\lambda\}],\sigma(\bm{j}[\Lambda,\{j_\lambda\}]))\prod_{(\lambda\neq\mu)\in\Lambda}(1-\delta_{j_\lambda,j_\mu}).
    \end{align}

    Finally, the product of deltas $\prod_{(\lambda\neq\mu)\in\Lambda}(1-\delta_{j_\lambda,j_\mu})$ ensures that we focus only on the terms where the representative indices take different values. Then, $T$ will read
    \begin{equation}
        T=\sum_{\text{diff.}\{j_\lambda\}}\sum_{\sigma \in S_{2N}}h(\bm{j}[\Lambda,\{j_\lambda\}],\sigma(\bm{j}[\Lambda,\{j_\lambda\}])).
    \end{equation}

    Summing all the contributions from different partitions, we reach the result
    \begin{align}
        \sum_{\bm{j},\bm{j}'}h(\bm{j},\bm{j}')I(\bm{j},\bm{j}')=\sum_{\Lambda\in\mathcal{Q}[\bm{j}]}\sum_{\text{diff.}\{j_\lambda\}}\sum_{\sigma \in S_{2N}}\frac{1}{\Lambda!}h(\bm{j}[\Lambda,\{j_\lambda\}],\sigma(\bm{j}[\Lambda,\{j_\lambda\}])).
        \label{eq:reordering_sums}
    \end{align}

    Eq.~\eqref{eq:lxe_sqz_after_integral} is obtained by applying this result to 
    \begin{align}
        h(\bm{j},\bm{j}')=\sum_{\bm{\mu},\bm{\nu}}\zeta_{\bm{\mu}}\zeta_{\bm{\nu}}\,\mathcal{U}\left(\bm{j}\oplus\bm{j}',\bm{\bar{\mu}}\,\vert\,\Omega_{\bm{k}}(\bm{j})\oplus\Omega_{\bm{l}}(\bm{j}'),\bm{\bar{\nu}}\right).
    \end{align}
    
    \section{\label{app:weingarten}Weingarten calculus}

    In this appendix we gather the two key theorems regarding the Weingarten Calculus that we used in Sec.~\ref{sec:average_haar}.

    \begin{proposition}[Lemma 3 from Ref.~\cite{matsumoto2012general}]
        \label{prop:average_unitary_weingarten}
        Let $\bm{U}$ be a $M\times M$ Haar-distributed unitary matrix and let $\bm{j}=(j_1,\dots,j_n)$, $\bm{\mu}=(\mu_1,\dots,\mu_n)$, $\bm{j}'=(j'_1,\dots,j'_m)$ and $\bm{\mu}'=(\mu'_1,\dots,\mu'_m)$ be four sequences of indices in $[M]=\{1,\dots,M\}$ (i.e., each index can take values from 1 to $M$). If $m=n$
        \begin{equation}
            \mathbb{E}_{\bm{U}}\left[\,\mathcal{U}\!\left(\bm{j},\bm{\mu}\,|\,\bm{j}',\bm{\mu}'\right)\right] = \sum_{\substack{\varrho\in S_n\\\bm{j}'=\varrho(\bm{j})}}\sum_{\substack{\tau\in S_n\\\bm{\mu}'=\tau(\bm{\mu})}} \mathrm{Wg}_n(\varrho^{-1}\circ\tau;M),
            \label{eq:average_unitary_weingarten}
        \end{equation}
        and it vanishes otherwise. Here,
        \begin{equation}
            \mathcal{U}\!\left(\bm{j},\bm{\mu}\,|\,\bm{j}',\bm{\mu}'\right)=U_{j_1,\mu_1}\cdots U_{j_n,\mu_n}U^*_{j'_1,\mu'_1}\cdots U^*_{j'_n,\mu'_n},    
        \end{equation}
        and $\mathrm{Wg}_n(\sigma;M)$ is the Weingarten function for the unitary group $U(M)$~\cite{collins2022weingarten,matsumoto2012general}. The sums extend over all permutations $\varrho$, $\tau$ in the symmetric group of degree $n$, $S_n$, such that $\bm{j}'=\varrho(\bm{j})$ and $\bm{\mu}'=\tau(\bm{\mu})$.
    \end{proposition}
    Notice that we can write 
    \begin{align}
        \sum_{\substack{\varrho\in S_n\\\bm{j}'=\varrho(\bm{j})}}\sum_{\substack{\tau\in S_n\\\bm{\mu}'=\tau(\bm{\mu})}} \mathrm{Wg}_n(\varrho^{-1}\circ\tau;M)=\sum_{\varrho\in S_n}\sum_{\tau\in S_n}\Delta[\bm{j}'\,\vert\,\varrho(\bm{j})]\Delta[\bm{\mu}'\,\vert\,\tau(\bm{\mu})]\mathrm{Wg}_n(\varrho^{-1}\circ\tau;M),
    \end{align}
    with $\Delta[\bm{j}'\,\vert\,\varrho(\bm{j})]=\prod_{a=1}^n\delta_{j'_a,\varrho(j_a)}$, which corresponds to the form used in Eq.~\eqref{eq:average_monomial}.

    \begin{proposition}[Corollary 2.7 from Ref.~\cite{collins2006integration}]
    \label{prop:asymptotic_weingarten}
        \begin{equation}
            M^{n+||\sigma||}\,\mathrm{Wg}_n(\sigma;M)=\mathrm{Moeb}(\sigma) + \mathcal{O}(M^{-2}),
            \label{eq:asymptotic_weingarten_app}
        \end{equation}
        where $||\sigma||$ is the minimum number of transpositions in which we can write $\sigma$, and
        \begin{equation}
            \mathrm{Moeb}(\sigma)=\prod_{k=1}^l\mathrm{Cat}(|\chi_k|-1)(-1)^{|\chi_k|-1}    
        \end{equation}
        is the Möbius function. In this expression, the $\{\chi_k\}$ are the disjoint cycles in the cycle decomposition of $\sigma$ (i.e., we write $\sigma=\chi_1\chi_2\cdots\chi_l$), and $|\chi_k|$ denotes the length of the cycle $\chi_k$. $\mathrm{Cat}(n)=(2n)!/(n!(n+1)!)$ stands for the $n$-th Catalan number. 
    \end{proposition}

    \section{\label{app:general_sqz_score}Computation of the LXE score for the ideal squeezed state model with different squeezing parameters}

    In this appendix we complete the details of the computation of the ideal score for GBS setups using different squeezing parameters, and show that resulting expression is consistent with the findings of Sec.~\ref{sec:lxe_score_sqz}.

    Our main goal in this section is to prove that
    \begin{equation}
        \sum_{\bm{\mu}\in [R]^{2N}} \!\!\!\sum_{\substack{\bm{\nu}\in [R]^{2N}\\ \text{s.t }\bm{\bar{\nu}}=\varrho(\bm{\bar{\mu}})}}\!\prod_{a=1}^{2N}\tanh(r_{\mu_a})\tanh(r_{\nu_a}) = \left(\prod_{b\in\eta(\varrho)}\varepsilon_{2b}\right)R^{\ell(\varrho)}.
        \label{eq:rel_f_1}
    \end{equation}
    In order to do so, let us first briefly recall the definition of all the terms involved in this relation. $\bm{\mu}=(\mu_1,\dots,\mu_{2N})$, $\bm{\nu}=(\nu_1,\dots,\nu_{2N})$ are vectors whose integer components take values in the set $[R]=\{1,\dots,R\}$ (this is indicated by the symbols $\bm{\mu}\in [R]^{2N}$, $\bm{\nu}\in [R]^{2N}$). $\bar{\bm{\mu}}=(\bar{\mu}_1,\dots,\bar{\mu}_{4N})=(\mu_1,\mu_1,\dots,\mu_{2N},\mu_{2N})$, and $\bar{\bm{\nu}}$ has an analogous definition. Let us remind the reader that $R$ is the number of modes with a squeezed state at the input. The squeezing parameters of the input single-mode squeezed states are $\{r_k\}_{k=1}^R$. $\eta(\varrho)$ stands for the coset-type of $\varrho\in S_{4N}$, and $\ell(\varrho)$ is the length of $\eta(\varrho)$. Finally, $\varepsilon_a$ is defined as
    \begin{equation}
        \varepsilon_a = \frac{1}{R}\sum_{k=1}^{R}\tanh^a(r_k).    
    \end{equation}

    Now, notice that the second sum in the left-hand side of Eq.~\eqref{eq:rel_f_1} runs over values of $\bm{\nu}$ that satisfy the condition $\bar{\bm{\nu}}=\varrho(\bar{\bm{\mu}})$. Along with the definitions of $\bar{\bm{\nu}}$ and $\bar{\bm{\mu}}$, this condition allows us to write the following relations:
    \begin{align}
        &\bar{\mu}_{2k-1}=\bar{\mu}_{2k}=\mu_k,\;k\in\{1,\dots,2N\},\label{eq:mu_rel_1}\\
        &\bar{\mu}_{\varrho(2k-1)}=\bar{\mu}_{\varrho(2k)}=\nu_k,\;k\in\{1,\dots,2N\}\label{eq:mu_rel_2}.
    \end{align}
    The second of these relations leads to two types of constraints: (i) a constraint that defines each $\nu_k$ in terms of $\mu_k$ and $\varrho$; and (ii) a constraint between different $\{\mu_k\}$, which comes from the relation $\bar{\mu}_{\rho(2k-1)}=\bar{\mu}_{\varrho(2k)}$. This means that the final sum in Eq.~\eqref{eq:rel_f_1} will run only over a number of ``free'' $\{\tilde{\mu}_l\}$. In what follows, we will determine how many of these free indices there are.

    Let us recall that each $\varrho\in S_{4N}$ has an associated undirected graph, $\bm{\Gamma}(\varrho)$, whose vertices are $\{1,\dots,4N\}$, and whose edges are defined by the sets $\{(2k-1, 2k),\,k\in\{1,\dots,2N\}\}$ and $\{(\varrho(2k-1), \varrho(2k)),\,k\in\{1,\dots,2N\}\}$. Combined with relations~\eqref{eq:mu_rel_1} and~\eqref{eq:mu_rel_2}, this definition implies that for every edge in $\bm{\Gamma}(\varrho)$ there will be an equality of the form $\bar{\mu}_{2k}=\bar{\mu}_{2k-1}$ or $\bar{\mu}_{\varrho(2k)}=\bar{\mu}_{\varrho(2k-1)}$. This, in turn, leads to the following result:
    \begin{proposition}
        \label{prop:connected_equiv_equal}
        Let $a,b\in\{1,\dots,4N\}$ and $a\neq b$, then $\bar{\mu}_a=\bar{\mu}_b$ if and only if the vertices $a$ and $b$ are connected in $\Gamma(\varrho)$.
    \end{proposition}
    \begin{proof}
        If the vertices $a$ and $b$ are connected in $\bm{\Gamma}(\varrho)$, then there exists a sequence of vertices $\{v_1,v_2,\dots,v_p\}$, with $v_1=a$, $v_p=b$ and $2\leq p\leq 4N$, where $v_{k}$ is adjacent to $v_{k+1}$, i.e. there is an edge connecting them, for all $k\in\{1,\dots,p-1\}$. The presence of an edge between $v_{k}$ and $v_{k+1}$, implies that $\bar{\mu}_{v_k}=\bar{\mu}_{v_{k+1}}$. Since these relations hold for every $k$, we can conclude that $\bar{\mu}_{v_1}=\bar{\mu}_{v_p}$, so $\bar{\mu}_a=\bar{\mu}_b$. 

        Suppose now that $\bar{\mu}_{a}=\bar{\mu}_b$. If there exists some $k\in\{1,\dots,2N\}$ such that $a=2k-1$, $b=2k$ or $a=\varrho(2k-1),b=\varrho(2k)$ (or the same relations but inverting the places of $a$ and $b$), then $a$ and $b$ are connected. In any other case, we can find a new vertex $v_2$ adjacent to $a$ that satisfies $\bar{\mu}_a=\bar{\mu}_{v_2}$, either by Eq.~\eqref{eq:mu_rel_1} or Eq.~\eqref{eq:mu_rel_2}, and check if it is connected to $b$. If not, we can find a second vertex $v_3$ adjacent to $v_2$ but not to $a$ (because every vertex lies in exactly two edges) that will satisfy $\bar{\mu}_a=\bar{\mu}_{v_2}=\bar{\mu}_{v_3}$, then we can verify whether it is adjacent to $b$ or not. We can repeat this process until we find a sequence of adjacent vertices $\{v_2,\dots,v_{p-1}\}$ that satisfy $\bar{\mu}_a=\bar{\mu}_{v_2}=\cdots=\bar{\mu}_b$. The existence of this sequence implies that $a$ and $b$ are connected in $\bm{\Gamma}(\varrho)$.
    \end{proof}

    Following this proposition, we can see that if $\bar{\mu}_a\neq\bar{\mu}_b$, then the vertices $a$ and $b$ will belong to different connected components in $\bm{\Gamma}(\varrho)$. Since every $\bar{\mu}_a$ is equal to some $\mu_k$, with $k\in\{1,\dots,2N\}$, this implies that there will be as many different (or free) $\{\tilde{\mu}_{l}\}$ as there are connected components in $\bm{\Gamma}(\varrho)$. Therefore, we may say that the number of different $\{\tilde{\mu}_l\}$ is equal to the length of the coset-type of $\varrho$, $\ell(\varrho)$.

    In each connected component, an edge of the form $(2k-1,2k)$ determines a single $\mu_k$ via the relation $\bar{\mu}_{2k-1}=\bar{\mu}_{2k}=\mu_k$. The edges of the form $(\varrho(2k-1),\varrho(2k))$ determine which $\{\mu_k\}$ are equal to others. This implies that the number of repetitions, i.e., the degeneracy, of every free $\tilde{\mu}_l$ will be equal to the number of edges of the form $(2k-1, 2k)$ in its corresponding connected component. This number is equal to half the length of the connected component.

    Let us recall the definition of the coset-type of $\varrho\in S_{4N}$. Let $2\eta_1,2\eta_2,\dots,2\eta_{\ell}$, with $\eta_1\geq\eta_2\geq\cdots\geq\eta_\ell\geq1$, be the lengths of the different connected components in $\bm{\Gamma}(\varrho)$. The coset-type of $\varrho$ is defined as $\eta(\varrho)=(\eta_1,\dots,\eta_\ell)$. As can be seen, $\eta_l$ is half the length of the $l$-th connected component of $\bm{\Gamma}(\varrho)$, which, according to the argument given above, will be equal to the degeneracy of the free $\tilde{\mu}_l$ corresponding to that component. This means that we can determine the degeneracies of all $\{\tilde{\mu}_l\}$ by computing $\eta(\varrho)$. 

    We are now in the position to prove Eq.~\eqref{eq:rel_f_1}. Taking into account that each $\{\nu_k\}$ will be equal to some free $\tilde{\mu}_l$, we may write 
    \begin{equation}
        \sum_{\bm{\mu}\in [R]^{2N}} \!\!\!\sum_{\substack{\bm{\nu}\in [R]^{2N}\\ \text{s.t }\bm{\bar{\nu}}=\varrho(\bm{\bar{\mu}})}}\!\prod_{a=1}^{2N}\tanh(r_{\mu_a})\tanh(r_{\nu_a}) = \sum_{\{\tilde{\mu}_l\}}\prod_{a=1}^{\ell(\varrho)}\tanh^{2\eta_a}(r_{\tilde{\mu}_a})=\prod_{a=1}^{\ell(\varrho)}\sum_{\tilde{\mu}_a=1}^R\tanh^{2\eta_a}(r_{\tilde{\mu}_a}).
        \label{eq:rel_f_2}
    \end{equation}
    Recalling the definition of $\varepsilon_a$, we obtain
    \begin{equation}
        \prod_{a=1}^{\ell(\varrho)}\sum_{\tilde{\mu}_a=1}^R\tanh^{2\eta_a}(r_{\tilde{\mu}_a})=\prod_{a=1}^{\ell(\varrho)}\varepsilon_{2\eta_a}R= \left(\prod_{b\in\eta(\varrho)}\varepsilon_{2b}\right)R^{\ell(\varrho)}.
        \label{eq:rel_f_3}
    \end{equation}

    To conclude this appendix, we will show that the result in Eq.~\eqref{eq:lxe_sqz_score_different_squeezing} is consistent with Eq.~\eqref{eq:main_result_lxe_sqz}. Suppose that $r_k=r$ for all $k\in\{1,\dots,R\}$, then $\varepsilon_a = \tanh^a(r)$ for every $a$. This allows us to see that
    \begin{equation}
        \prod_{b\in\eta(\varrho)}\varepsilon_{2b}=\tanh^{2\eta_1+\dots+2\eta_\ell}(r)=\tanh^{4N}(r),
        \label{eq:epsilon_rel}
    \end{equation}
    where we took into account that, according to its definition, $\eta(\varrho)$ is an integer partition of $2N$ for all $\varrho$ (therefore $\eta_1 +\dots +\eta_\ell=2N$). This allows us to recast Eq.~\eqref{eq:c_tilde_coeffs} as 
    \begin{align}
        \begin{split}
            c'_\ell &=\!\!\! \sum_{\bm{k}^{(c)}, \bm{l}^{(c)}}\,\prod_{a=1}^N\frac{1}{k_a!l_a!(2a)^{k_a+l_a}}\!\!\sum_{\sigma\in S_{2N}}\sum_{\varrho\in \bar{S}_{\ell}}\tanh^{4N}(r)\\
            &=\tanh^{4N}(r) \sum_{\bm{k}^{(c)}, \bm{l}^{(c)}}\,\prod_{a=1}^N\frac{1}{k_a!l_a!(2a)^{k_a+l_a}}\!\!\sum_{\sigma\in S_{2N}}|\bar{S}_\ell|.
        \end{split}
        \label{eq:c_tilde_coeffs_app}
    \end{align}

    Now recalling the definition of $|\bar{S}_{\ell}|$ from Eq.~\eqref{eq:set_length_rel},
    \begin{equation}
        |\bar{S}_{\ell}| = b_{\ell}\left(\bm{j}\oplus\sigma(\bm{j}),\Omega_{\bm{k}}(\bm{j})\oplus\Omega_{\bm{l}}\circ\sigma(\bm{j})\right),
        \label{eq:set_length_rel_l}
    \end{equation}
    and $b_\ell$ from Eq.~(\ref{eq:definebell}), we can then write
    \begin{align}
        c'_\ell &=\tanh^{4N}(r) \sum_{\bm{k}^{(c)}, \bm{l}^{(c)}}\,\prod_{a=1}^N\frac{1}{k_a!l_a!(2a)^{k_a+l_a}}\!\!\sum_{\sigma\in S_{2N}}b_{\ell}\left(\bm{j}\oplus\sigma(\bm{j}),\Omega_{\bm{k}}(\bm{j})\oplus\Omega_{\bm{l}}\circ\sigma(\bm{j})\right)=\tanh^{4N}(r)\,c_\ell.
        \label{eq:c_tilde_coeffs_app_2}
    \end{align}

    On the other hand, we have 
    \begin{equation}
        \sum_{\ell=1}^Nd'_\ell R^\ell=\sum_{\bm{k}^{(c)}}\,\prod_{a=1}^N\frac{\varepsilon_{2a}^{k_a}R^{k_a}}{k_a!(2a)^{k_a}}=\sum_{\bm{k}^{(c)}}\,\prod_{a=1}^N\frac{\tanh^{2ak_a}(r)R^{k_a}}{k_a!(2a)^{k_a}}=\tanh^{2N}(r)\sum_{\bm{k}^{(c)}}\,\prod_{a=1}^N\frac{R^{k_a}}{k_a!(2a)^{k_a}},
        \label{eq:d_coeff_app}
    \end{equation}
    where we used the condition $k_1+2k_2+\cdots+Nk_N=N$. In order to continue, we need the following proposition:
    \begin{proposition}
        \label{prop:sum_Z_n_R}
        \begin{equation}
            \sum_{\bm{k}^{(c)}}\,\prod_{a=1}^N\frac{R^{k_a}}{k_a!(2a)^{k_a}}=\binom{\frac{R}{2}+N-1}{N}.
            \label{eq:sum_Z_n_R}
        \end{equation}
    \end{proposition}
    \begin{proof}
        Notice that
        \begin{equation}
            \sum_{\bm{k}^{(c)}}\,\prod_{a=1}^N\frac{R^{k_a}}{k_a!(2a)^{k_a}}=Z_N\left(\frac{R}{2},\frac{R}{2},\cdots,\frac{R}{2}\right),   
        \end{equation}
        where $Z_N$ is the cycle index of $S_{N}$ (see Eq.~\eqref{eq:cycle_index}):
        \begin{equation}
            Z_N(y_1,\dots,y_N)=\sum_{\bm{k}^{(c)}}\,\prod_{a=1}^N\frac{1}{k_a!a^{k_a}}\prod_{a=1}^Ny_a^{k_a}.
        \end{equation}

        Using the generating function of $Z_n(y_1,\dots,y_n)$~\cite{comtet1974advanced,wilf2005generatingfunctionology}
        \begin{equation}
            \exp\left[\sum_{l=1}^\infty y_l\frac{\alpha^l}{l}\right]=\sum_{n=0}^\infty Z_n(y_1,\dots,y_n)\alpha^n,    
        \end{equation}
        we can write
        \begin{equation}
            \exp\left[\frac{R}{2}\sum_{l=1}^\infty\frac{\alpha^l}{l}\right]=\sum_{n=0}^\infty Z_n\left(\frac{R}{2},\frac{R}{2},\dots,\frac{R}{2}\right)\alpha^n.    
        \end{equation}
        Assuming that $|\alpha|<1$, we can write $\sum_{l=1}^\infty \alpha^l/l=-\log(1-\alpha)$, and obtain
        \begin{equation}
            \exp\left[-\frac{R}{2}\log(1-\alpha)\right]=(1-\alpha)^{-R/2}=\sum_{n=0}^\infty Z_n\left(\frac{R}{2},\frac{R}{2},\dots,\frac{R}{2}\right)\alpha^n.    
        \end{equation}
        
        Taking into account that
        \begin{equation}
            (1-\alpha)^{-R/2}=\sum_{n=0}^\infty\binom{\frac{R}{2}+n-1}{n}\alpha^n,   
        \end{equation}
        we see that Eq.~\eqref{eq:sum_Z_n_R} holds.
    \end{proof}

    From this result, it follows that
    \begin{equation}
        \sum_{\ell=1}^Nd'_\ell R^\ell=\tanh^{2N}(r)\binom{\frac{R}{2}+N-1}{N}.
        \label{eq:d_coeff_app_2}
    \end{equation}
    Combining this relation with Eqs.~\eqref{eq:c_tilde_coeffs_app_2} and~\eqref{eq:lxe_sqz_score_different_squeezing}, we obtain
    \begin{align}
        \begin{split}
            s(\bm{A}_{\mathrm{sqz}}';2N)&=\frac{1}{(2N)!}\left(\sum_{\ell=1}^Nd'_\ell R^\ell\right)^{-2}\sum_{\ell=1}^{2N}c'_{\ell}R^{\ell}=\frac{1}{(2N)!}\tanh^{-4N}(r)\binom{\frac{R}{2}+N-1}{N}^{-2}\tanh^{4N}(r)\sum_{\ell=1}^{2N}c_{\ell}R^{\ell}\\
            &=\frac{4^{N}(N!)^2}{(2N)!}\left[\frac{(R-2)!!}{(R+2N-2)!!}\right]^2\sum_{\ell=1}^{2N}c_\ell R^{\ell}=s(\bm{A}_\mathrm{sqz};2N).
        \end{split}
        \label{eq:lxe_sqz_score_different_squeezing_app}
    \end{align}

    \section{\label{app:non_vacuum_proofs} Computation of the LXE score for the ideal squeezed state model without vacuum input modes}

    In this appendix we present the details of the computation of the score for the ideal squeezed state model without vacuum input modes. Specifically, we prove that 
    \begin{align}
        \begin{split}
            \sum_{\bm{k}^{(c)}, \bm{l}^{(c)}}\,\prod_{a=1}^N\frac{1}{k_a!l_a!(2a)^{k_a+l_a}}\sum_{\sigma\in S_{2N}}b_{2N}\left(\bm{j}\oplus\sigma(\bm{j}),\Omega_{\bm{k}}(\bm{j})\oplus\Omega_{\bm{l}}\circ\sigma(\bm{j})\right)=1,
        \end{split}
        \label{eq:counting_part}
    \end{align}
    with $b_{2N}$ defined in Eq.~(\ref{eq:b_coefficients_definition}).
    This leads to the result
    \begin{equation}
        s(\tilde{\bm{A}}'_\mathrm{sqz};2N)=s(\tilde{\bm{A}}_\mathrm{sqz};2N)=\frac{4^{N}(N!)^2}{(2N)!}.
    \label{eq:ideal_score_no_vacuum_app}
    \end{equation}

    Let us recall that $b_{2N}\left(\bm{j}\oplus\sigma(\bm{j}),\Omega_{\bm{k}}(\bm{j})\oplus\Omega_{\bm{l}}\circ\sigma(\bm{j})\right)$ is the number of permutations in the hyperoctahedral group of degree $2N$, $H_{2N}\subset S_{4N}$, that transform the sequence $\bm{j}\oplus\sigma(\bm{j})$ into $\Omega_{\bm{k}}(\bm{j})\oplus\Omega_{\bm{l}}\circ\sigma(\bm{j})$. Recall also that the indices $\bm{j}=(j_1,\dots,j_{2N})$ are all different. 
    
    We begin by recasting $b_{2N}\left(\bm{j}\oplus\sigma(\bm{j}),\Omega_{\bm{k}}(\bm{j})\oplus\Omega_{\bm{l}}\circ\sigma(\bm{j})\right)$ into an expression that is independent of $\bm{j}$. 
    Given two permutations $\sigma \in S_{n}$ and $\sigma'\in S_m$, their \textit{direct sum} $\sigma\oplus\sigma' \in S_{m+n}$ is defined by its action on the sequence $(1,\dots,m+n)$ as
    \begin{equation}
        (\sigma\oplus\sigma')(k)=
        \begin{cases}
            \sigma(k) & k\in\{1,\dots, n\},\\
            \sigma'(k-n) + n & k\in\{n+1,\dots, n+m\},
        \end{cases}
        \label{eq:direct_sum_perm_def}
    \end{equation}
    i.e., $\sigma\oplus\sigma'$ acts as $\sigma$ in the first $n$ elements of a sequence, and as $\sigma'$ in the remaining elements. From this definition, we can write the relations
    \begin{equation}
        \bm{j}\oplus\sigma(\bm{j})=(e_{2N}\oplus\sigma)(\bm{j}\oplus\bm{j}),   
    \end{equation}
    \begin{equation}
        \Omega_{\bm{k}}(\bm{j})\oplus\Omega_{\bm{l}}\circ\sigma(\bm{j})=(\Omega_{\bm{k}}\oplus\Omega_{\bm{l}})\circ(e_{2N}\oplus\sigma)(\bm{j}\oplus\bm{j}),   
    \end{equation}
    where $e_{2N}$ is the identity permutation in $S_{2N}$.
    
    To shorten the notation, we will denote the composition $\sigma\circ\tau$ of two permutations $\sigma,\tau\in S_n$ as $\sigma\tau$.
    Suppose now that $\varrho \in S_{4N}$ transforms $\bm{j}\oplus\sigma(\bm{j})$ into $\Omega_{\bm{k}}(\bm{j})\oplus\Omega_{\bm{l}}\,\sigma(\bm{j})$. Then
    \begin{equation}
        \varrho\,(e_{2N}\oplus\sigma)(\bm{j}\oplus\bm{j})=(\Omega_{\bm{k}}\oplus\Omega_{\bm{l}})(e_{2N}\oplus\sigma)(\bm{j}\oplus\bm{j}),    
    \end{equation}
    which implies 
    \begin{align}
        \left[[(\Omega_{\bm{k}}\oplus\Omega_{\bm{l}})(e_{2N}\oplus\sigma)]^{-1}\varrho\,(e_{2N}\oplus\sigma)\right](\bm{j}\oplus\bm{j})=\left[(e_{2N}\oplus\sigma^{-1})(\Omega_{\bm{k}}^{-1}\oplus\Omega_{\bm{l}}^{-1})\,\varrho\,(e_{2N}\oplus\sigma)\right](\bm{j}\oplus\bm{j})=\bm{j}\oplus\bm{j},
        \label{eq:rho_def_1}
    \end{align}
    where $(\sigma\oplus\tau)^{-1}=\sigma^{-1}\oplus\tau^{-1}$. This means that $(e_{2N}\oplus\sigma^{-1})(\Omega_{\bm{k}}^{-1}\oplus\Omega_{\bm{l}}^{-1})\,\varrho\,(e_{2N}\oplus\sigma)$ is a permutation that leaves the sequence 
    $\bm{j}\oplus\bm{j}$ invariant.
    
    Let $S^\star\subset S_{4N}$ be defined as
    \begin{equation}
        S^\star=\{\nu\in S_{4N}\,\vert\,\nu(\bm{j}\oplus\bm{j})=\bm{j}\oplus\bm{j}\}.
        \label{eq:star_subset}
    \end{equation}
    Then, every permutation $\varrho$ that takes $\bm{j}\oplus\sigma(\bm{j})$ into $\Omega_{\bm{k}}(\bm{j})\oplus\Omega_{\bm{l}}\,\sigma(\bm{j})$ can be written in the form
    \begin{equation}    \varrho=(\Omega_{\bm{k}}\oplus\Omega_{\bm{l}})\left[(e_{2N}\oplus\sigma)\,\nu\,(e_{2N}\oplus\sigma^{-1})\right],\;\,\, \nu \in S^\star.
        \label{eq:rho_def_2}
    \end{equation}
    
    It can be easily shown that $S^\star$ forms a subgroup of $S_{4N}$. Indeed, the identity permutation $e_{4N}\in S_{4N}$ clearly leaves the sequence $\bm{j}\oplus\bm{j}$ invariant, so $e_{4N}\in S^\star$. Let $\nu_1,\nu_2\in S^\star$, we have $\nu_1\nu_2(\bm{j}\oplus\bm{j})=\nu_1(\nu_2(\bm{j}\oplus\bm{j}))=\nu_1(\bm{j}\oplus\bm{j})= \bm{j}\oplus\bm{j}$, which implies that $\nu_1\nu_2\in S^\star$. We can show that $\nu_2\nu_1\in S^*$ in a completely analogous way. Finally, for any $\nu_1\in S^\star$, $\bm{j}\oplus\bm{j}=\nu^{-1}_1\nu_1(\bm{j}\oplus\bm{j})=\nu_1^{-1}(\nu_1(\bm{j}\oplus\bm{j}))=\nu_1^{-1}(\bm{j}\oplus\bm{j})$, which implies that $\nu_1^{-1}\in S^\star$.

    The subset $S_\sigma^{\star}\subset S_{4N}$ defined by 
    \begin{equation}
        S_\sigma^\star=\{\tau\in S_{4N}\,\vert\,\tau=(e_{2N}\oplus\sigma)\,\nu\,(e_{2N}\oplus\sigma^{-1}),\;\nu \in S^\star\}
        \label{eq:star_subset_2}
    \end{equation}
    is also subgroup of $S_{4N}$ for every $\sigma \in S_{2N}$; it is the group of all permutations in $S_{4N}$ that leave invariant the sequence $\bm{j}\oplus\sigma(\bm{j})$.
    
    The permutations in $S^\star$ can only make the interchange $k \leftrightarrow k+2N$ for $k \in \{1,\dots, 2N\}$. In this way we guarantee that $\bm{j}\oplus\bm{j}$ remains invariant. Thus, an arbitrary permutation in $S^\star$ will take a subset $B$ of $\{1,\dots,2N\}$, make the interchange $k \leftrightarrow k+2N$ for all $k \in B$, and let all the remaining elements in $\{1,\dots,4N\}$ unchanged. We may then write every $\nu\in S^\star$ as $\nu \equiv \nu_B$, where $B\in \mathcal{P}(\{1,\dots,2N\})$ and $\mathcal{P}(A)$ stands for the power set of $A$. Furthermore, we may write the action of $\nu_B$ over the sequence $(1,\dots, 4N)$ as
    \begin{equation}
        \nu_B(k)=
        \begin{cases}
            k & k\not\in B,\\
            k + 2N & k\in B,\\
            k & k-2N\not\in B\text{ and }k\in\{2N+1,\dots,4N\},\\
            k-2N & k-2N\in B\text{ and }k\in\{2N+1,\dots,4N\}.
        \end{cases}
        \label{eq:nu_b_permutation}
    \end{equation}

    We can readily see that $|S^*|=|S_\sigma^\star|=2^{2N}$ for all $\sigma\in S_{2N}$, since there are as many $\nu_B$ as there are subsets of $\{1,\dots, 2N\}$. Moreover, we can write the action of the permutations $\tau_B \in S_\sigma^\star$ over the sequence $(1,\dots, 4N)$ as
    \begin{equation}
        \tau_B(k)=[(e_{2N}\oplus\sigma)\,\nu_B\,(e_{2N}\oplus\sigma^{-1})](k)=
        \begin{cases}
            k & k\not\in B,\\
            \sigma(k) + 2N & k\in B,\\
            k & \sigma^{-1}(k-2N)\not\in B\text{ and }k\in\{2N+1,\dots,4N\},\\
            \sigma^{-1}(k-2N) & \sigma^{-1}(k-2N)\in B\text{ and }k\in\{2N+1,\dots,4N\}.
        \end{cases}
        \label{eq:nu_b_sigma_permutation}
    \end{equation}

    Having characterized the permutations in $S^\star$ and $S_\sigma^\star$, let us return to the expression for $\varrho$ given in Eq.~\eqref{eq:rho_def_2}. In order to compute $b_{2N}\left(\bm{j}\oplus\sigma(\bm{j}),\Omega_{\bm{k}}(\bm{j})\oplus\Omega_{\bm{l}}\sigma(\bm{j})\right)$, we require that $\varrho = (\Omega_{\bm{k}}\oplus \Omega_{\bm{l}})\tau_B = \beta$ for some $\beta \in H_{2N}$. This implies that $\tau_B = (\Omega_{\bm{k}}^{-1}\oplus \Omega_{\bm{l}}^{-1})\beta$, or equivalently, $\tau_B \in (\Omega_{\bm{k}}^{-1}\oplus \Omega_{\bm{l}}^{-1}) H_{2N}$, i.e., $\tau_B$ must be an element of the \textit{left coset} of $H_{2N}$ corresponding to $\Omega_{\bm{k}}^{-1}\oplus \Omega_{\bm{l}}^{-1}$. On this account, determining the adequate $\varrho$ amounts to finding the elements of $(\Omega_{\bm{k}}^{-1}\oplus \Omega_{\bm{l}}^{-1}) H_{2N}\cap S_\sigma^\star$. Then, we may write 
    \begin{equation}
        b_{2N}\left(\bm{j}\oplus\sigma(\bm{j}),\Omega_{\bm{k}}(\bm{j})\oplus\Omega_{\bm{l}}\sigma(\bm{j})\right) = \left\vert(\Omega_{\bm{k}}^{-1}\oplus \Omega_{\bm{l}}^{-1}) H_{2N}\cap S_\sigma^\star\right\vert. 
        \label{eq:b_coeff_group_expression}
    \end{equation}
    This is the $\bm{j}$-independent expression of the coefficients $b_{2N}$. 
    
    In the following we will show how to compute $\left\vert(\Omega_{\bm{k}}^{-1}\oplus \Omega_{\bm{l}}^{-1}) H_{2N}\cap S_\sigma^\star\right\vert$ as a function of $\bm{k}$, $\bm{l}$ and $\sigma$. However, in order to do so, we need to prove a number of important intermediate results. The first of these results concerns the following definition:    
    \begin{definition}[Hyperoctahedral group]
        \label{def:hyperoctahedral_definition}
        The hyperoctahedral group of degree $n$, $H_n\subset S_{2n}$ is the centralizer of the permutation $\Omega^\star \in S_{2n}$ that transforms the sequence of indices $(g_1,\dots,g_{2n})$ as $\Omega^\star[(g_1,g_2,\dots,g_{2n-1},g_{2n})]=(g_2,g_1,\dots,g_{2n},g_{2n-1})$~\cite{islami2020symmetric}.
    \end{definition}

    Let us recall that the centralizer of a permutation $\tau \in S_m$ is the set of all permutations in $S_m$ that commute with $\tau$, i.e., $\{\varrho\in S_m\,\vert\,\varrho\tau=\tau\varrho\}$. As a product of disjoint transpositions, $\Omega^\star$ reads $\Omega^\star = (1\,2)(3\,4)\dots(2n-1,2n)$. Notice that $(\Omega^\star)^2 = e_{2n}$. The following proposition gives us a useful way of characterizing the elements of $H_n$.

    \begin{proposition}
        \label{prop:hyperpc_perms}
        $\sigma \in H_n$ if and only if for every $p\in \{1,\dots,n\}$, there exists a $q\in\{1,\dots,n\}$ such that $\sigma(2p)=2q$, $\sigma(2p-1)=2q-1$ or $\sigma(2p)=2q-1$, $\sigma(2p-1)=2q$.
    \end{proposition}

    \begin{proof}        
        Since $H_n$ is the centralizer of $\Omega^\star$, $\sigma\in H_n$ if and only if $\Omega^\star\sigma\,\Omega^\star = \sigma$. Let us consider an arbitrary $p\in\{1,\dots,n\}$. We have 
        \begin{equation}
            \sigma(2p)=[\Omega^\star\sigma\,\Omega^\star](2p)=[\Omega^\star\sigma](2p-1)=\Omega^*(\sigma(2p-1)).    
        \end{equation}
        Since $\sigma$ is a bijection from $\{1,\dots,2n\}$ to itself, there exists 
        a unique $q\in\{1,\dots,n\}$ such that $\sigma(2p-1)=2q$ or $\sigma(2p-1)=2q-1$. In the first case we have
        \begin{equation}
            \sigma(2p)=\Omega^\star(\sigma(2p-1))=\Omega^\star(2q)=2q-1.  
        \end{equation}
        In the second case we obtain
        \begin{equation}
            \sigma(2p)=\Omega^\star(\sigma(2p-1))=\Omega^\star(2q-1)=2q.
        \end{equation}
        
        On the other hand, suppose that $\sigma \in S_{2n}$ satisfies that for 
        every $p\in \{1,\dots,n\}$, there exists a $q\in\{1,\dots,n\}$ such that 
        $\sigma(2p)=2q$, $\sigma(2p-1)=2q-1$ or $\sigma(2p)=2q-1$, $\sigma(2p-1)=2q$. In the first case we have
        \begin{align}
            &[\Omega^\star\sigma](2p)=\Omega^\star(\sigma(2p))=\Omega^\star(2q)=2q-1,\\
            &[\Omega^\star\sigma](2p-1)=\Omega^\star(\sigma(2p-1))=\Omega^\star(2q-1)=2q,\\
            &[\sigma\,\Omega^\star](2p)=\sigma(\Omega^\star(2p))=\sigma(2p-1)=2q-1,\\
            &[\sigma\,\Omega^\star](2p-1)=\sigma(\Omega^\star(2p-1))=\sigma(2p)=2q.
        \end{align}
        Thus, $[\sigma\,\Omega^\star](2p)=[\Omega^\star\sigma](2p)$ and 
        $[\sigma\,\Omega^\star](2p-1)=[\Omega^\star\sigma](2p-1)$. Since these relations hold for an arbitrary $p$, we conclude that
        $\sigma\,\Omega^\star=\Omega^\star\sigma$ and, consequently, $\sigma\in H_n$. For the second case we may write
        \begin{align}
            &[\Omega^\star\sigma](2p)=\Omega^\star(\sigma(2p))=\Omega^\star(2q-1)=2q,\\
            &[\Omega^\star\sigma](2p-1)=\Omega^\star(\sigma(2p-1))=\Omega^\star(2q)=2q-1,\\
            &[\sigma\,\Omega^\star](2p)=\sigma(\Omega^\star(2p))=\sigma(2p-1)=2q,\\
            &[\sigma\,\Omega^\star](2p-1)=\sigma(\Omega^\star(2p-1))=\sigma(2p)=2q-1.
        \end{align}
        We conclude again that $[\sigma\,\Omega^\star](2p)=[\Omega^\star\sigma](2p)$ and 
        $[\sigma\,\Omega^\star](2p-1)=[\Omega^\star\sigma](2p-1)$ for every $p$, therefore $\sigma\,\Omega^\star=\Omega^\star\sigma$ and $\sigma \in H_n$.
    \end{proof}

    A useful interpretation of the previous proposition follows from analyzing the undirected graphs $\bm{\Gamma}(\sigma)$ associated with the $\sigma\in H_n$. As mentioned in Sec.~\ref{sec:final_form_lxe_score}, these graphs have $n$ connected components with two edges (See Fig.~\ref{fig:gamma_hn}). Since the first set of edges of $\bm{\Gamma}(\sigma)$ is of the form $\{(2p-1,2p)\,\vert\, p\in\{1,\dots, n\}\}$, the only way we can obtain $n$ connected components is that the edges in the second set $\{(\sigma(2p-1),\sigma(2p))\,\vert\, p\in\{1,\dots, n\}\}$ be of the form $(2q-1,2q)$ or $(2q,2q-1)$, with $q\in\{1,\dots,n\}$, for every $p$. 

    \begin{figure}[!ht]
        {
          \includegraphics[width=0.6\columnwidth]{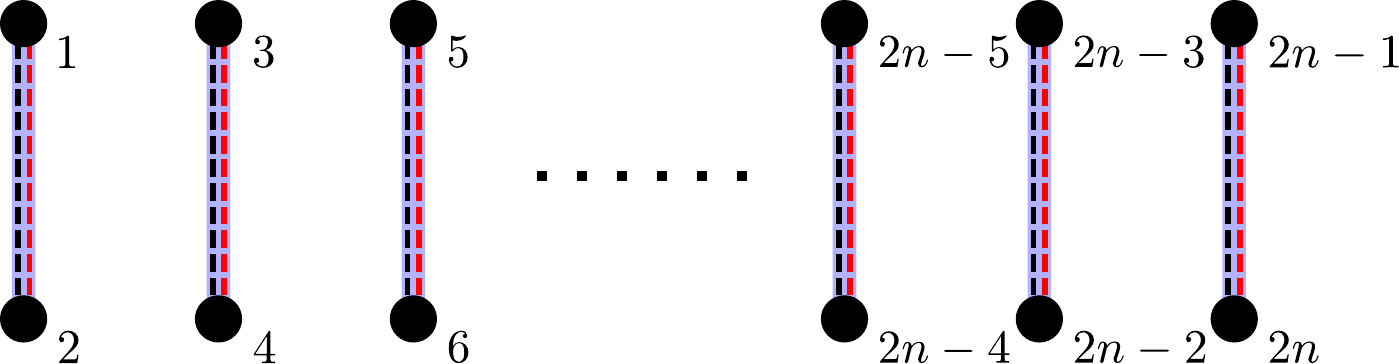}%
        }
        \caption{Structure of the undirected graphs $\bm{\Gamma}(\sigma)$ for $\sigma \in H_n$. The vertices of $\bm{\Gamma}(\sigma)$ are represented by black circles. The edges corresponding to $\{(2p-1,2p)\,\vert\,p\in\{1,\dots,n\}\}$ are shown as black dashed lines, while the edges corresponding to $\{(\sigma(2p-1),\sigma(2p))\,\vert\,p\in\{1,\dots,n\}\}$ are shown as red dashed lines. Each connected component of the graph is highlighted with a light blue, thick line. For $\bm{\Gamma}(\sigma)$ to have $n$ connected components, so that $\sigma\in H_n$, $\sigma$ must \textit{keep together} the pairs of the form $(2q-1,2q)$. This is the statement of Proposition ~\ref{prop:hyperpc_perms}.}
        \label{fig:gamma_hn}
    \end{figure}

    Next, we show how to determine the coset-type of the permutations $\Omega_{\bm{k}}\in S_{2N}$. From their definition (Eq.~\eqref{eq:big_omega_permutation_def}) we see that the $\Omega_{\bm{k}}$ are written as a product of disjoint cycles of even length, where a cycle of length $2a$ appears a total of $k_a$ times (notice that the permutations $\omega_a \in S_{2a}$ defined just below Eq.~\eqref{eq:big_omega_permutation_def} are precisely these cycles). Moreover, all the elements within the cycles are consecutive. These features of $\Omega_{\bm{k}}$ greatly facilitate the determination of its corresponding undirected graph $\bm{\Gamma}(\Omega_{\bm{k}})$.
    
    \begin{proposition}{\label{prop:coset_type_capital_omega}}
        $\Omega_{\bm{k}}$ has coset-type $\eta(\Omega_{\bm{k}})=(N^{k_N},\dots,2^{k_2},1^{k_1})$, where $a^{k_a}$ indicates that $a$ appears $k_a$ times in the partition $\eta(\Omega_{\bm{k}})$. For example, $(3^2, 1^4)\equiv (3,3,1,1,1,1)$. If $k_a=0$, the corresponding $a$ does not appear in the partition. It follows that the length of the coset-type of $\Omega_{\bm{k}}$ is $\ell(\Omega_{\bm{k}})=\sum_{a=1}^Nk_a$.
    \end{proposition}

    \begin{figure}[!ht]
        {
          \includegraphics[width=0.85\columnwidth]{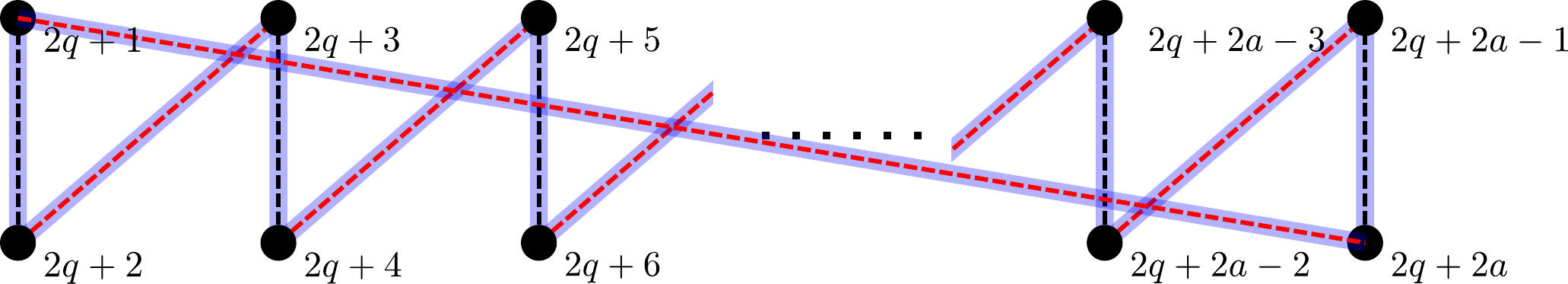}%
        }
        \caption{Illustration of the connected component of $\bm{\Gamma}(\Omega_{\bm{k}})$ corresponding to the vertices $(2q+1,\dots,2q+2a)$ (black circles). The edges corresponding to $\{(2p-1,2p)\,\vert\,p\in\{q+1,\dots,2q+2a\}\}$ are shown as black dashed lines, while the edges corresponding to $\{(\Omega_{\bm{k}}(2p-1),\Omega_{\bm{k}}(2p))\,\vert\,p\in\{q+1,\dots,2q+2a\}\}=\{(\omega_{a}(2p-1),\omega_{a}(2p))\,\vert\,p\in\{q+1,\dots,2q+2a\}\}$ are shown as red dashed lines. The connected component is highlighted by a light blue, thick line. As mentioned in the proof of Proposition~\ref{prop:coset_type_capital_omega}, a cycle of length $2a$ acting over $(2q+1,\dots,2q+2a)$ leads to a connected component of length $2a$ in $\bm{\Gamma}(\Omega_{\bm{k}})$.}
        \label{fig:connected_comp_omega}
    \end{figure}
    
    \begin{proof}
        Let us consider one of the cycles of length $2a$ that appear in the definition of $\Omega_{\bm{k}}$. This cycle acts over a sequence of 
        consecutive indices of the form $(2q+1,\dots,2q+2a)$, where we use $2q$
        to label the beginning of this particular cycle. After applying $\Omega_{\bm{k}}$, this sequence transforms as
        $(2q+1,\dots,2q+2a)\rightarrow\omega_a[(2q+1,\dots,2q+2a)]=(2q+2,\dots,2q+2a,2q+1)$. 
        
        When constructing $\bm{\Gamma}(\Omega_{\bm{k}})$ for this subset of 
        vertices, we can see that the first set of edges of the graph has the 
        form $\{(2q+1,2q+2),\dots,(2q+2a-1, 2q+2a)\}$, while the second reads 
        $\{(\omega_a(2q+1),\omega_a(2q+2)),\dots,(\omega_a(2q+2a-1), \omega_a(2q+2a))\} = \{(2q+2,2q+3),(2q+4,2q+5),\dots,(2q+2a, 2q+1)\}$.
        These edges lead to a connected component of length $2a$ in $\bm{\Gamma}(\Omega_{\bm{k}})$ (see Fig.~\ref{fig:connected_comp_omega}).
        
        Since all the cycles in $\Omega_{\bm{k}}$ are disjoint, each cycle of
        length $2a$ leads to a different connected component of length $2a$ in
        $\bm{\Gamma}(\Omega_{\bm{k}})$. Since there are $k_a$ cycles of length
        $2a$ in $\Omega_{\bm{k}}$, there will be $k_a$ connected components of
        length $2a$ in $\bm{\Gamma}(\Omega_{\bm{k}})$. This implies that $a$ will appear a total of $k_a$ times in $\eta(\Omega_{\bm{k}})$, and thus $\eta(\Omega_{\bm{k}})=(N^{k_N},\dots,2^{k_2},1^{k_1})$. Notice that since $k_1+2k_2+\cdots+Nk_N=N$, $\eta(\Omega_{\bm{k}})$ will be, as expected, an integer partition of $N$. Finally, there will be a total of $k_1+\dots+k_N$ connected components in $\bm{\Gamma}(\Omega_{\bm{k}})$, so $\ell(\Omega_{\bm{k}})=\sum_{a=1}^Nk_a$. 
    \end{proof}

    From the previous proposition, we can see that $\Omega_{\bm{k}}$ and $\Omega_{\bm{l}}$ will have different coset-types whenever $\bm{k}\neq\bm{l}$.
    
    We move on to stating two important results, one of them proved in~\cite{macdonald1998symmetric}, regarding coset-types and left cosets.

    \begin{proposition}[(Macdonald) Theorem 2.1 \cite{macdonald1998symmetric}]
        \label{prop:coset_type_equality}
        \textbf{(i)} Two permutations $\sigma_1, \sigma_2 \in S_{2n}$ have the same coset-type if and only if $\sigma_2\in H_n\sigma_1H_n$, that is, if and only if there exist $\tau_1,\tau_2\in H_n$ such that $\sigma_2 = \tau_1\sigma_1\tau_2$. \textbf{(ii)} $\sigma\in S_{2n}$ has the same coset-type as $\sigma^{-1}$.        
    \end{proposition}

    \begin{proposition}
        \label{prop:coset_type_and_left_coset}
        If $\sigma_1 \in \sigma_2^{-1}H_n$, then $\sigma_1$ and $\sigma_2$ have the same coset-type.
    \end{proposition}
    \begin{proof}
        If $\sigma_1\in \sigma_2^{-1}H_n$, then there exists $\tau\in H_n$ such that $\sigma_1=\sigma^{-1}_2\tau$. Since $H_n$ is a group, $e_{2n}\in H_n$, and thus we may write $\sigma_1=e_{2n}\sigma^{-1}_2\tau$. This means that $\sigma_1\in H_n\sigma_2^{-1}H_n$, which implies that $\sigma_1$ and $\sigma_2^{-1}$ have the same coset-type.  Given that $\sigma_2$ and $\sigma_2^{-1}$ have the same coset-type, we conclude that $\sigma_1$ and $\sigma_2$ have the same coset-type. 
    \end{proof}

    Equivalently, we may say that if $\sigma_1,\sigma_2 \in S_{2n}$ have different coset-types, then $\sigma_1\not\in \sigma_2^{-1}H_n$. Notice that the converse of Proposition~\ref{prop:coset_type_and_left_coset} does not necessarily hold; the fact that $\sigma_1$ and $\sigma_2$ have the same coset-type does not imply that $\sigma_1\in \sigma_2^{-1}H_n$.
 As a direct consequence of Propositions~\ref{prop:coset_type_equality} and~\ref{prop:coset_type_and_left_coset}, we obtain the following result, which will be of particular importance in our proof of Eq.~\eqref{eq:counting_part}: 

    \begin{proposition}
        \label{prop:no_simultaneaous_breaking}
        Let $\sigma_1,\sigma_2,\tau\in S_{2n}$ and suppose that $\sigma_1$ and $\sigma_2$ have different coset-types. Then $\tau\in\sigma_1^{-1}H_n$ implies that $\tau^{-1}\not\in\sigma_2^{-1}H_n$.   
    \end{proposition}
    \begin{proof}
        If $\tau\in\sigma_1^{-1}H_n$, $\tau$ and $\tau^{-1}$ have the same coset-type as $\sigma_1$. Since $\sigma_1$ and $\sigma_2$ have different coset-types, so will be the case for $\tau^{-1}$ and $\sigma_{2}$ and, consequently, $\tau^{-1}\not\in\sigma_2^{-1}H_n$.
    \end{proof}

    We are now in the position to compute $\left\vert(\Omega_{\bm{k}}^{-1}\oplus \Omega_{\bm{l}}^{-1}) H_{2N}\cap S_\sigma^\star\right\vert$ for the case $\bm{k}\neq\bm{l}$. A permutation $\tau_B\in S_{\sigma}^\star$, corresponding to a subset $B\subseteq\{1,\dots,2N\}$, will be an element of $(\Omega_{\bm{k}}^{-1}\oplus \Omega_{\bm{l}}^{-1}) H_{2N}$ if $(\Omega_{\bm{k}}\oplus\Omega_{\bm{l}})\tau_B=\beta$ for $\beta\in H_{2N}$. This relation will set some constraints over $B$ and $\sigma$. In order to find them, let us combine Eqs.~\eqref{eq:direct_sum_perm_def} and~\eqref{eq:nu_b_sigma_permutation}, and express the action of $\beta$ over pairs of the form $(2p-1,2p)$, with $p\in\{1,\dots,2N\}$, as
    \begin{equation}
        \beta(2p-1)=
        \begin{cases}
            \Omega_{\bm{k}}(2p-1) & 2p-1\not\in B\text{ and }p\in\{1,\dots,N\},\\
            \Omega_{\bm{l}}(\sigma(2p-1)) + 2N & 2p-1\in B\text{ and }p\in\{1,\dots,N\},\\
            \Omega_{\bm{l}}(2p-2N-1) + 2N & \sigma^{-1}(2p-2N-1)\not\in B\text{ and }p\in\{N+1,\dots,2N\},\\
            \Omega_{\bm{k}}(\sigma^{-1}(2p-2N-1)) & \sigma^{-1}(2p-2N-1)\in B\text{ and }p\in\{N+1,\dots,2N\},
        \end{cases}
        \label{eq:sigma_octahedral_permutation_1}
    \end{equation}
    \begin{equation}
        \beta(2p)=
        \begin{cases}
            \Omega_{\bm{k}}(2p) & 2p\not\in B\text{ and }p\in\{1,\dots,N\},\\
            \Omega_{\bm{l}}(\sigma(2p)) + 2N & 2p\in B\text{ and }p\in\{1,\dots,N\},\\
            \Omega_{\bm{l}}(2p-2N) + 2N & \sigma^{-1}(2p-2N)\not\in B\text{ and }p\in\{N+1,\dots,2N\},\\
            \Omega_{\bm{k}}(\sigma^{-1}(2p-2N)) & \sigma^{-1}(2p-2N)\in B\text{ and }p\in\{N+1,\dots,2N\}.
        \end{cases}
        \label{eq:sigma_octahedral_permutation_2}
    \end{equation}

    According to Proposition~\ref{prop:hyperpc_perms}, $\beta\in H_{2N}$ if for every $p\in\{1,\dots,2N\}$ there exists a $q\in\{1,\dots,2N\}$ such that $\beta(2p)=2q$, $\beta(2p-1)=2q-1$ or $\beta(2p)=2q-1$, $\beta(2p-1)=2q$. Keeping in mind that $B^c$ is the complement of $B$, this condition allows us to draw the following conclusions:
    \begin{enumerate}
        \item $B$ must be non-empty because $\Omega_{\bm{k}}$, $\Omega_{\bm{l}}$ are not, in general, elements of $H_{N}$, which would be necessary in order to have $\beta\in H_{2N}$.
        
        \item When $p\in\{1,\dots,N\}$, it is necessary that either $2p-1,\,2p\in B$ or $2p-1,\,2p\not\in B$, i.e., the pair $(2p-1,2p)$ must remain together either in $B$ or in $B^{\mathrm{c}}$. Indeed, if $2p\in B$ but $2p-1\not\in B$, then $\beta(2p-1)=\Omega_{\bm{k}}(2p-1)\in\{1,\dots,2N\}$, while $\beta(2p)=\Omega_{\bm{l}}(\sigma(2p))+2N\in\{2N+1,\dots,4N\}$, which breaks the condition for $\beta$ to be in $H_{2N}$. The case $2p-1\in B$, $2p\not\in B$ yields to a similar conclusion.

        \item When $p\in\{N+1,\dots,2N\}$, it is necessary that either $\sigma^{-1}(2p-2N),\, \sigma^{-1}(2p-2N-1) \in B$ or $\sigma^{-1}(2p-2N),\,\sigma^{-1}(2p-2N-1) \not\in B$. Just as before, if $\sigma^{-1}(2p-2N)\in B$ but $\sigma^{-1}(2p-2N-1)\not\in B$, then $\beta(2p-1)=\Omega_{\bm{l}}(2p-2N-1)+2N\in\{2N+1,\dots,4N\}$, while $\beta(2p)=\Omega_{\bm{k}}(\sigma^{-1}(2p-2N))\in \{1,\dots,2N\}$, which implies that $\beta\not\in H_{2N}$. The case $\sigma^{-1}(2p-2N-1)\in B$, $\sigma^{-1}(2p-2N)\not\in B$ leads to a similar conclusion.

        \item Assuming that the two conditions above hold, we need that $B^{\mathrm{c}}$ contains the pairs $(2p-1, 2p)$, with $p\in\{1,\dots,N\}$, for which $\Omega_{\bm{k}}(2p-1)=2q-1$, $\Omega_{\bm{k}}(2p)=2q$ or $\Omega_{\bm{k}}(2p-1)=2q$, $\Omega_{\bm{k}}(2p)=2q-1$ for some $q\in\{1,\dots,N\}$. Moreover, we need that $B^{\mathrm{c}}$ contains the pairs $(\sigma^{-1}(2p-2N-1),\sigma^{-1}(2p-2N))$, with $p\in\{N+1,\dots,2N\}$,  for which $\Omega_{\bm{l}}(2p-2N-1)=2q'-1$, $\Omega_{\bm{l}}(2p-2N)=2q'$ or $\Omega_{\bm{l}}(2p-2N-1)=2q'$, $\Omega_{\bm{l}}(2p-2N)=2q'-1$ for some $q'\in\{1,\dots,N\}$. In other words, the elements in $B^{\mathrm{c}}$ must correspond to transpositions in the cycle decomposition of both $\Omega_{\bm{k}}$ and $\Omega_{\bm{l}}$. Recall that the permutations $\Omega_{\bm{k}}$ are defined as products of cycles of even length. The only possible cycles that naturally satisfy the condition described here are those of length two, i.e., the transpositions.

        \item We can guarantee that the elements in $B^{\mathrm{c}}$ correspond to transpositions in both $\Omega_{\bm{k}}$ and $\Omega_{\bm{l}}$ by constructing $B$ using the following prescription. There are $k_1$ transpositions (cycles of length two) in $\Omega_{\bm{k}}$ and the last element in $\{1,\dots,2N\}$ that belongs to a transposition in $\Omega_{\bm{k}}$ is precisely $2k_1$. Similarly, the last element in $\{1,\dots,2N\}$ that belongs to a transposition in $\Omega_{\bm{l}}$ is $2l_1$. Let $2p^*=\mathrm{min}(2k_1,2l_1)$, then $B=\{2p^*+1,\dots,2N\}$ contains all the elements in $\{1,\dots,2N\}$ that correspond to cycles of length greater than two in both $\Omega_{\bm{k}}$ and $\Omega_{\bm{l}}$. Notice that $B$ will also contain elements that correspond to some transpositions in $\Omega_{\bm{k}}$ or $\Omega_{\bm{l}}$ as long as $k_1\neq l_1$. This construction ensures that the elements of $B^{\mathrm{c}}$ lead only to transpositions in both $\Omega_{\bm{k}}$ and $\Omega_{\bm{l}}$. 

        \item The prescription given above is not the only way to construct $B$. Indeed, the only condition that $B$ should satisfy is that it contains all the elements in $\{1,\dots,2N\}$ that correspond to cycles of length greater than two in both $\Omega_{\bm{k}}$ and $\Omega_{\bm{l}}$. This means that $B$ can also contain any number of pairs of the form $(2p-1,2p)$ that correspond to transpositions in $\Omega_{\bm{k}}$ or $\Omega_{\bm{l}}$.

        \item Let $B$ satisfy items 5 or 6 above, and suppose that there is some $p\in\{N+1,\dots, 2N\}$ such that $\sigma^{-1}(2p-2N-1),\,\sigma^{-1}(2p-2N)\not\in B$ but $2p-2N-1,\,2p - 2N\in B$. Then, in general, we cannot guarantee that $\beta(2p-1)=\Omega_{\bm{l}}(2p-2N-1)+2N$, $\beta(2p)=\Omega_{\bm{l}}(2p-2N)+2N$ will satisfy the conditions that let $\beta\in H_{2N}$. This is because most (or all) of the elements in $B$ correspond to cycles of length greater than two in $\Omega_{\bm{l}}$. In order to avoid this, we demand that $\sigma^{-1}(B)=\sigma(B)=B$, i.e., the image of $B$ under $\sigma$, $\sigma^{-1}$ is $B$ itself. In this way, the condition $\sigma^{-1}(2p-2N-1),\,\sigma^{-1}(2p-2N)\not\in B$ will always correspond to a transposition in $\Omega_{\bm{l}}$. This constraints $\sigma$ to have the form $\sigma = \bar{\sigma}\oplus\widetilde{\sigma}$, where $\bar{\sigma}\in S_{2N-|B|}$ acts only over $B^{\mathrm{c}}$, and $\widetilde{\sigma}\in S_{|B|}$ acts only over $B$.
    \end{enumerate}
    Following these considerations, we can recast Eqs.~\eqref{eq:sigma_octahedral_permutation_1} and~\eqref{eq:sigma_octahedral_permutation_2} as
    \begin{equation}
        \beta(2p-1)=
        \begin{cases}
            \bar{\Omega}_{\bm{k}}(2p-1) & 2p-1\not\in B\text{ and }p\in\{1,\dots,N\},\\
            \widetilde{\Omega}_{\bm{l}}(\widetilde{\sigma}(2p-1)) + 2N & 2p-1\in B\text{ and }p\in\{1,\dots,N\},\\
            \bar{\Omega}_{\bm{l}}(2p-2N-1) + 2N & 2p-2N-1\not\in B\text{ and }p\in\{N+1,\dots,2N\},\\
            \widetilde{\Omega}_{\bm{k}}(\widetilde{\sigma}^{-1}(2p-2N-1)) & 2p-2N-1\in B\text{ and }p\in\{N+1,\dots,2N\},
        \end{cases}
        \label{eq:sigma_octahedral_permutation_3}
    \end{equation}
    \begin{equation}
        \beta(2p)=
        \begin{cases}
            \bar{\Omega}_{\bm{k}}(2p) & 2p\not\in B\text{ and }p\in\{1,\dots,N\},\\
            \widetilde{\Omega}_{\bm{l}}(\widetilde{\sigma}(2p)) + 2N & 2p\in B\text{ and }p\in\{1,\dots,N\},\\
            \bar{\Omega}_{\bm{l}}(2p-2N) + 2N & 2p-2N\not\in B\text{ and }p\in\{N+1,\dots,2N\},\\
            \widetilde{\Omega}_{\bm{k}}(\widetilde{\sigma}^{-1}(2p-2N)) & 2p-2N\in B\text{ and }p\in\{N+1,\dots,2N\},
        \end{cases}
        \label{eq:sigma_octahedral_permutation_4}
    \end{equation}
    where we used the decomposition of $\Omega_{\bm{k}}$ and $\Omega_{\bm{l}}$ as $\Omega_{\bm{k}}=\bar{\Omega}_{\bm{k}}\oplus\widetilde{\Omega}_{\bm{k}}$, $\Omega_{\bm{l}}=\bar{\Omega}_{\bm{l}}\oplus\widetilde{\Omega}_{\bm{l}}$ with $\bar{\Omega}_{\bm{k}},\,\bar{\Omega}_{\bm{l}}\in S_{2N-|B|}$ and $\widetilde{\Omega}_{\bm{k}},\,\widetilde{\Omega}_{\bm{l}}\in S_{|B|}$ ($\bar{\Omega}_{\bm{k}}$, $\bar{\Omega}_{\bm{l}}$ will be a product of some of the transpositions of $\Omega_{\bm{k}}$ and $\Omega_{\bm{l}}$, and $\widetilde{\Omega}_{\bm{k}}$, $\widetilde{\Omega}_{\bm{l}}$ will be the products of the remaining cycles). Notice also that $\bar{\Omega}_{\bm{k}},\,\bar{\Omega}_{\bm{l}}\in H_{N-|B|/2}$.

    We can see now that finding a way for $\beta$ to be in $H_{2N}$ is
    equivalent to finding a $\widetilde{\sigma}\in S_{|B|}$ that satisfies the
    following two conditions simultaneously: 
    \begin{itemize}
        \item For every $2p-1,2p\in B$, there exist $2q-1,2q\in B$ such that $\widetilde{\Omega}_{\bm{k}}(\widetilde{\sigma}^{-1}(2p-1))=2q-1$, $\widetilde{\Omega}_{\bm{k}}(\widetilde{\sigma}^{-1}(2p))=2q$ or $\widetilde{\Omega}_{\bm{k}}(\widetilde{\sigma}^{-1}(2p-1))=2q$, $\widetilde{\Omega}_{\bm{k}}(\widetilde{\sigma}^{-1}(2p))=2q-1$.
        \item For every $2p-1,2p\in B$, there exist $2q'-1,2q'\in B$ such that $\widetilde{\Omega}_{\bm{l}}(\widetilde{\sigma}(2p-1))=2q'-1$, $\widetilde{\Omega}_{\bm{l}}(\widetilde{\sigma}(2p))=2q'$ or $\widetilde{\Omega}_{\bm{l}}(\widetilde{\sigma}(2p-1))=2q'$, $\widetilde{\Omega}_{\bm{l}}(\widetilde{\sigma}(2p))=2q'-1$.
    \end{itemize}
    
    These statements are equivalent to having $\widetilde{\Omega}_{\bm{k}}\widetilde{\sigma}^{-1},\, \widetilde{\Omega}_{\bm{l}}\widetilde{\sigma}\in H_{|B|/2}$, which may also be written as $\widetilde{\sigma}^{-1}\in \widetilde{\Omega}_{\bm{k}}^{-1}H_{|B|/2}$ and $\widetilde{\sigma}\in \widetilde{\Omega}_{\bm{l}}^{-1}H_{|B|/2}$. However, Proposition~\ref{prop:no_simultaneaous_breaking} forbids the existence of $\widetilde{\sigma}$. Indeed, $\widetilde{\Omega}_{\bm{k}}$ and $\widetilde{\Omega}_{\bm{l}}$ have different coset-types, and thus, according to Proposition~\ref{prop:no_simultaneaous_breaking}, the conditions $\widetilde{\sigma}^{-1}\in \widetilde{\Omega}_{\bm{k}}^{-1}H_{|B|/2}$, $\widetilde{\sigma}\in \widetilde{\Omega}_{\bm{l}}^{-1}H_{|B|/2}$ cannot hold simultaneously.

    To confirm that $\widetilde{\Omega}_{\bm{k}}$ and $\widetilde{\Omega}_{\bm{l}}$ have different coset-types, recall that $B$ is defined in such a way that it contains all the elements in $\{1,\dots,2N\}$ that correspond to cycles of length greater than 2 in both $\Omega_{\bm{k}}$ and $\Omega_{\bm{l}}$. Suppose for a moment that $k_1=l_1$. In order to have $\bm{k}\neq\bm{l}$ there must be at least one $a\neq1$ such that $k_a\neq l_a$, which implies that the number of connected components of length $2a$ in $\bm{\Gamma}(\widetilde{\Omega}_{\bm{k}})$ will be different form that of $\bm{\Gamma}(\widetilde{\Omega}_{\bm{l}})$, and thus $\widetilde{\Omega}_{\bm{k}}$ and $\widetilde{\Omega}_{\bm{l}}$ will have different coset-types. If $k_1\neq l_1$, $B$ will contain at least one pair of elements that will correspond to a transposition in $\Omega_{\bm{k}}$ but not in $\Omega_{\bm{l}}$ (or vice versa). This implies that $\bm{\Gamma}(\widetilde{\Omega}_{\bm{k}})$ and $\bm{\Gamma}(\widetilde{\Omega}_{\bm{l}})$ will differ by at least one connected component of length two, and so $\widetilde{\Omega}_{\bm{k}}$ and $\widetilde{\Omega}_{\bm{l}}$ will have different coset-types.

    The following proposition is a direct consequence of the fact that conditions $\widetilde{\sigma}^{-1}\in \widetilde{\Omega}_{\bm{k}}^{-1}H_{|B|/2}$, $\widetilde{\sigma}\in \widetilde{\Omega}_{\bm{l}}^{-1}H_{|B|/2}$ cannot hold at the same time:
    \begin{proposition}
    \label{prop:counting_result_different_coset}
        Let $\bm{k}\neq\bm{l}$, then $\left\vert(\Omega_{\bm{k}}^{-1}\oplus \Omega_{\bm{l}}^{-1}) H_{2N}\cap S_\sigma^\star\right\vert=0$ for every $\sigma\in S_{2N}$.
    \end{proposition} 
    \begin{proof}
        We have seen that in order for $\tau_B\in S_{\sigma}^\star$, with
        $B\subseteq\{1,\dots,2N\}$, to also be an element of
        $(\Omega_{\bm{k}}^{-1}\oplus \Omega_{\bm{l}}^{-1}) H_{2N}$,
        Eqs.~\eqref{eq:sigma_octahedral_permutation_1}
        and~\eqref{eq:sigma_octahedral_permutation_2}, later transformed into
        Eqs.~\eqref{eq:sigma_octahedral_permutation_3}
        and~\eqref{eq:sigma_octahedral_permutation_4}, must hold. This can
        only happen if $B$ contains all the elements in $\{1,\dots,2N\}$ that
        correspond to cycles of length greater than two in both
        $\Omega_{\bm{k}}$ and $\Omega_{\bm{l}}$. This, in turn, makes
        $\sigma$ take the form $\sigma=\bar{\sigma}\oplus\widetilde{\sigma}$,
        where $\bar{\sigma}\in S_{2N-|B|}$ and $\widetilde{\sigma}\in S_{|B|}$.
        Moreover, no matter the choice of $B$, $\widetilde{\sigma}$ must satisfy
        the conditions $\widetilde{\sigma}^{-1}\in
        \widetilde{\Omega}_{\bm{k}}^{-1}H_{|B|/2}$, $\widetilde{\sigma}\in
        \widetilde{\Omega}_{\bm{l}}^{-1}H_{|B|/2}$ simultaneously. However, as
        proved in Proposition~\eqref{prop:no_simultaneaous_breaking}, this
        cannot happen because $\Omega_{\bm{k}}$ and $\Omega_{\bm{l}}$, and 
        consequently $\widetilde{\Omega}_{\bm{k}}$ and $\widetilde{\Omega}_{\bm{l}}$,
        have different coset-types when $\bm{k}\neq\bm{l}$. Therefore, there
        is no $\widetilde{\sigma}\in S_{|B|}$, and thus no $\sigma\in S_{2N}$, for which $\tau_B\in S_{\sigma}^\star$ and $\tau_B \in (\Omega_{\bm{k}}^{-1}\oplus \Omega_{\bm{l}}^{-1}) H_{2N}$. It follows that $\left\vert(\Omega_{\bm{k}}^{-1}\oplus \Omega_{\bm{l}}^{-1}) H_{2N}\cap S_\sigma^\star\right\vert=0$ for every $\sigma\in S_{2N}$.
    \end{proof}

    We now focus on the case $\bm{k}=\bm{l}$. Let us set $B=\{2k_1+1,\dots,2N\}$, so that the elements in $B^{\mathrm{c}}$ correspond to all the transpositions in $\Omega_{\bm{k}}$ (recall that $B^{\mathrm{c}}$ is the complement of $B$). We will discuss possible modifications of $B$ later on.

    According to relations~\eqref{eq:sigma_octahedral_permutation_3} 
    and~\eqref{eq:sigma_octahedral_permutation_4}, in order to have $\beta\in H_{2N}$, we need $\widetilde{\Omega}_{\bm{k}}\widetilde{\sigma},\, \widetilde{\Omega}_{\bm{k}}\widetilde{\sigma}^{-1}\in H_{N-k_1}$. We will see how
    these relations determine the possible $\widetilde{\sigma}$ that we can use.
    First, notice that $\widetilde{\Omega}_{\bm{k}}^2\in H_{N-k_1}$. Indeed, let
    the sequence $(2q+1,\dots,2q+2a)$ correspond to a cycle of length $2a$ in
    $\widetilde{\Omega}_{\bm{k}}$. After one application of $\widetilde{\Omega}_{\bm{k}}$, $(2q+1,\dots,2q+2a)$ maps to $(2q+1,\dots,2q+2a) 
    \mapsto (2q+2,2q+3,\dots,2q+2a-1,2q+2a,2q+1)$. After a second application
    of the permutation, we have $(2q+1,\dots,2q+2a)\mapsto(2q+3,2q+4,\dots,2q+2a-1,2q+2a,2q+1,2q+2)$, which brings the pairs of the 
    form $(2p-1, 2p)$, with $p\in\{1,\dots,a\}$, back together. Since this
    holds for an arbitrary $a$, we see that $\widetilde{\Omega}_{\bm{k}}^2$
    satisfies the conditions to be in $H_{N-k_1}$. In view of this, we see that $\widetilde{\Omega}_{\bm{k}}^{-2}\in H_{N-k_1}$, and thus 
    $\widetilde{\Omega}_{\bm{k}}^{-2}\widetilde{\Omega}_{\bm{k}}\widetilde{\sigma}=\widetilde{\Omega}_{\bm{k}}^{-1}\widetilde{\sigma}^{-1} \in H_{N-k_1}$. Consequently, the conditions that $\widetilde{\sigma}$ must satisfy can be written as $\widetilde{\Omega}_{\bm{k}}\widetilde{\sigma},\, \widetilde{\sigma}\widetilde{\Omega}_{\bm{k}}\in H_{N-k_1}$.

    Another interesting property of $\widetilde{\Omega}_{\bm{k}}$ is that every pair of the form $(2p-1,2p)$,
    with $p\in\{k_1+1,\dots,N\}$, belongs to the same cycle in
    $\widetilde{\Omega}_{\bm{k}}$. This is because all the elements within the
    cycles of $\widetilde{\Omega}_{\bm{k}}$ are consecutive and, moreover, all of
    the cycles have even length. We will use this fact to find the structure
    of $\widetilde{\sigma}$.

    Consider again the sequence $(2q+1,\dots,2q+2a)$, corresponding to a
    cycle of length $2a$ in $\widetilde{\Omega}_{\bm{k}}$. Every element in this
    sequence can be written as $2q+2l-1$ or $2q+2l$ with $l\in\{1,\dots,a\}$.
    Therefore, every pair of the form $(2p-1,2p)$ with elements in this
    sequence can be written as $(2q+2l-1,2q+2l)$. This holds for every cycle
    in $\widetilde{\Omega}_{\bm{k}}$.
    
    If $\widetilde{\Omega}_{\bm{k}}\widetilde{\sigma}\in H_{N-k_1}$, then, for all $l\in\{1,\dots,a\}$, $\widetilde{\Omega}_{\bm{k}}(\widetilde{\sigma}(2q+2l-1))=2q'+2l'-1$ and $\widetilde{\Omega}_{\bm{k}}(\widetilde{\sigma}(2q+2l))=2q'+2l'$, or $\widetilde{\Omega}_{\bm{k}}(\widetilde{\sigma}(2q+2l-1))=2q'+2l'$ and $\widetilde{\Omega}_{\bm{k}}(\widetilde{\sigma}(2q+2l))=2q'+2l'-1$, with $2q'+1$ being the first element in another cycle of length $2a'$ in $\widetilde{\Omega}_{\bm{k}}$. Here $l'\in\{1,\dots,a'\}$. We can use these relations to find the action of $\widetilde{\sigma}$ over the elements in $(2q+1,\dots,2q+2a)$. We will analyze these two possible conditions separately, and during this procedure we will keep in mind $\widetilde{\sigma}\widetilde{\Omega}_{\bm{k}}\in H_{N-k_1}$ also holds.
    
    \textbf{Case 1:} $\widetilde{\Omega}_{\bm{k}}(\widetilde{\sigma}(2q+2l-1))=2q'+2l'-1$ and $\widetilde{\Omega}_{\bm{k}}(\widetilde{\sigma}(2q+2l))=2q'+2l'$.
    
    The first of these relations implies that $\widetilde{\sigma}(2q+2l-1)=\Omega_{\bm{k}}^{-1}(2q'+2l'-1)$. Consequently, if $l'\neq1$, $\widetilde{\sigma}(2q+2l-1)=2q'+2l'-2$. If $l'=1$, $\widetilde{\sigma}(2q+2l-1)=2q'+2a'$. The second of these relations implies that $\widetilde{\sigma}(2q+2l)=\Omega_{\bm{k}}^{-1}(2q'+2l')=2q'+2l'-1$. 
    
    Moreover, we have that $\widetilde{\sigma}(\widetilde{\Omega}_{\bm{k}}(2q+2l-1))=\widetilde{\sigma}(2q+2l)=2q'+2l'-1$, which, given that $\widetilde{\sigma}\widetilde{\Omega}_{\bm{k}}\in H_{N-k_1}$, implies $\widetilde{\sigma}(\widetilde{\Omega}_{\bm{k}}(2q+2l))=2q'+2l'$. If $l\neq a$, $\widetilde{\sigma}(\widetilde{\Omega}_{\bm{k}}(2q+2l))=\widetilde{\sigma}(2q+2l+1)=2q'+2l'$. If $l=a$, $\widetilde{\sigma}(\widetilde{\Omega}_{\bm{k}}(2q+2l))=\widetilde{\sigma}(2q+1)=2q'+2l'$. 
    
    Therefore, $\widetilde{\sigma}$ must have the following structure:
    \begin{equation}
        \begin{array}{cc|c|c|c|c}
            p:&\dots& 2q+2l-1&2q+2l&2q+2l+1&\dots,\\
            \hline
            \widetilde{\sigma}(p):&\dots& 2q'+2l'-2&2q'+2l'-1&2q'+2l'&\dots.
        \end{array}
        \label{eq:sigma_structure_1}
    \end{equation}
    The first and last elements of the cycles can be included in this structure by taking the addition and subtraction modulo $a$ or $a'$.
    
    \textbf{Case 2:} $\widetilde{\Omega}_{\bm{k}}(\widetilde{\sigma}(2q+2l-1))=2q'+2l'$ and $\widetilde{\Omega}_{\bm{k}}(\widetilde{\sigma}(2q+2l))=2q'+2l'-1$.
    
    The first relation implies that $\widetilde{\sigma}(2q+2l-1)=\Omega_{\bm{k}}^{-1}(2q'+2l')=2q'+2l'-1$. The second relation implies that $\widetilde{\sigma}(2q+2l)=\Omega_{\bm{k}}^{-1}(2q'+2l'-1)$. If $l'=1$, $\widetilde{\sigma}(2q+2l)=\Omega_{\bm{k}}^{-1}(2q'+1)=2q'+2a'$, otherwise $\widetilde{\sigma}(2q+2l)=\Omega_{\bm{k}}^{-1}(2q'+2l'-1)=2q'+2l'-2$.   
    
    Moreover, we have that $\widetilde{\sigma}(\widetilde{\Omega}_{\bm{k}}(2q+2l-1))=\widetilde{\sigma}(2q+2l)$. Thus, considering that 
    $\widetilde{\sigma}\widetilde{\Omega}_{\bm{k}}\in H_{N-k_1}$, $\widetilde{\sigma}(\widetilde{\Omega}_{\bm{k}}(2q+2l))=2q'+2a'-1$ for $l'= 1$, and 
    $\widetilde{\sigma}(\widetilde{\Omega}_{\bm{k}}(2q+2l))=2q'+2l'-3$ for $l'\neq 1$. These relations lead to four results: 
    \begin{align}
     \sigma(2q+2l+1)&=2q'+2l'-3,\text{ for }l'\neq 1, l\neq a;\\
     \sigma(2q+1)&=2q'+2l'-3,\text{ for }l'\neq 1, l= a;\\
     \sigma(2q+2l+1)&=2q'+2a'-1,\text{ for }l'= 1, l\neq a;\\
     \sigma(2q+1)&=2q'+2a'-1,\text{ for }l'= 1, l= a.
    \end{align}
    
    Gathering all these results, we find that $\widetilde{\sigma}$ must have the 
    following structure:
    \begin{equation}
        \begin{array}{cc|c|c|c|c} p:&\dots&2q+2l-1&2q+2l&2q+2l+1&\dots,\\
            \hline
            \widetilde{\sigma}(p):&\dots& 2q'+2l'-1&2q'+2l'-2&2q'+2l'-3&\dots,
        \end{array}
        \label{eq:sigma_structure_2}
    \end{equation}
    where the first and last elements of the cycles can be included by taking the addition and subtraction modulo $a$ or $a'$.

    As we can see, the results of case 1 show that the image of $(2q+1,\dots, 2q+2a)$ under $\widetilde{\sigma}$ corresponds to shifting the elements in
    $(2q'+1,\dots, 2q'+2a')$ while maintaining their ascending order. The results of case 2 show that we are also
    allowed to make shifts using a descending order of the elements within the cycles. The value of $l'$
    determines by how much $\widetilde{\sigma}$ shifts the elements in
    $(2q'+1,\dots, 2q'+2a')$, while the value of $q'$ determines which cycle
    corresponds to the image of $(2q+1,\dots, 2q+2a)$.

    It is apparent that we need to have $a=a'$, so that every element in the sequence $(2q+1,\dots, 2q+2a)$ has an image in $(2q'+1,\dots, 2q'+2a')$. We can show that this is in fact a consequence of the relations $\widetilde{\sigma}\widetilde{\Omega}_{\bm{k}},\, \widetilde{\sigma}^{-1}\widetilde{\Omega}_{\bm{k}}\in H_{N-k_1}$ (the second of these is obtained by noticing that $\widetilde{\Omega}_{\bm{k}}\widetilde{\sigma}\in H_{N-k_1}$ implies  
    $\widetilde{\sigma}^{-1}\widetilde{\Omega}_{\bm{k}}^{-1}\in H_{N-k_1}$, which leads to $\widetilde{\sigma}^{-1}\widetilde{\Omega}_{\bm{k}}^{-1}\widetilde{\Omega}_{\bm{k}}^{2}=\widetilde{\sigma}^{-1}\widetilde{\Omega}_{\bm{k}}\in H_{N-k_1}$).
    
    Suppose that $a>a'$, then there must be some $l^*\in\{1,\dots,a\}$ such that $\widetilde{\sigma}(2q+2l^*-1),\, \widetilde{\sigma}(2q+2l^*)\not\in\{2q'+1,\dots, 2q'+2a'\}$, but $\widetilde{\sigma}(2q+2l^*-3),\, \widetilde{\sigma}(2q+2l^*-2)\in\{2q'+1,\dots, 2q'+2a'\}$. Let $\widetilde{\sigma}(2q+2l^*-1)=r$, $\widetilde{\sigma}(2q+2l^*)=r'$. After applying $\widetilde{\Omega}_{\bm{k}}$ and $\widetilde{\sigma}$ over $(2q+1,\dots,2q+2a)$, we have 
    
    \begin{equation}
        \begin{array}{cc|c|c|c|c|c}p:&\dots&2q+2l^*-3&2q+2l^*-2&2q+2l^*-1&2q+2l^*&\dots,\\
            \hline
            \widetilde{\Omega}_{\bm{k}}(p):&\dots& 2q+2l^*-2&2q+2l^*-1&2q+2l^*&2q+2l^*+1 &\dots,\\
            \hline
            \widetilde{\sigma}(\widetilde{\Omega}_{\bm{k}}(p)):&\dots&2q'+2l'+ j &r&r'&\dots&\dots,\\
        \end{array}
    \end{equation}
    where $j\in\{1,\dots,a'\}$. We see that the image of the pair 
    $(2q+2l^*-3,2q+2l^*-2)$ is $(2q'+2l'+j, r)$ under 
    $\widetilde{\sigma}\widetilde{\Omega}_{\bm{k}}$, and this breaks the condition 
    $\widetilde{\sigma}\widetilde{\Omega}_{\bm{k}}\in H_{N-k_1}$ because the
    companion of $2q'+2l'+j$ should be an element of $\{2q'+1,\dots, 2q'+2a'\}$, but $r$ is not. Consequently, we must have $a\leq a'$.

    Suppose now that $a<a'$, then there must exist $m^*\in\{1,\dots,a'\}$
    such that $\widetilde{\sigma}^{-1}(2q'+2m^*-1),\, \widetilde{\sigma}^{-1}(2q'+2m^*)\not\in\{2q+1,\dots, 2q+2a\} $, but $\widetilde{\sigma}^{-1}(2q'+2m^*-3),\, \widetilde{\sigma}^{-1}(2q'+2m^*-2)\in\{2q+1,\dots, 2q+2a\}$.
    Let $\widetilde{\sigma}^{-1}(2q'+2m^*-1)=s$, $\widetilde{\sigma}^{-1}(2q'+2m^*)=s'$. After applying $\widetilde{\Omega}_{\bm{k}}$ and
    $\widetilde{\sigma}^{-1}$ to $(2q'+1,\dots,2q'+2a')$, we obtain  

    \begin{equation}
        \begin{array}{cc|c|c|c|c|c}p:&\dots&2q'+2m^*-3&2q'+2m^*-2&2q'+2m^*-1&2q'+2m^*&\dots,\\
            \hline
            \widetilde{\Omega}_{\bm{k}}(p):&\dots& 2q'+2m^*-2&2q'+2m^*-1&2q'+2m^*&2q'+2m^*+1&\dots,\\
            \hline
            \sigma^{-1}(\widetilde{\Omega}_{\bm{k}}(p)):&\dots&2q+2l+j &s&s'&\dots&\dots,\\
        \end{array}
    \end{equation}
    where $j\in\{1,\dots,a\}$. Just as before, we see that the image under $\widetilde{\sigma}^{-1}\widetilde{\Omega}_{\bm{k}}$ of the pair $(2q'+2m^*-3,2q'+2m^*-2)$ is $(2q+2l+j, s)$, which breaks the condition $\widetilde{\sigma}^{-1}\widetilde{\Omega}_{\bm{k}}\in H_{N-k_1}$ because  $2q+2l+j$ should be paired with an element of $\{2q+1,\dots, 2q+2a\}$, not with $s\not\in \{2q+1,\dots, 2q+2a\}$. We therefore conclude that $a=a'$, i.e., the image under $\widetilde{\sigma}$ of any cycle of length $2a$ in $\widetilde{\Omega}_{\bm{k}}$ must correspond to another cycle of length $2a$ in $\widetilde{\Omega}_{\bm{k}}$.

    Having found the structure of $\widetilde{\sigma}$, we can easily count how many of them there are.

    \begin{proposition}
        \label{prop:how_many_sigma_p1}
        There are $\prod_{a=2}^N(2a)^{k_a}k_a!$ possible $\widetilde{\sigma}$.
    \end{proposition}
    \begin{proof}
    As we just showed, the image under $\widetilde{\sigma}$ of the elements within a cycle of length $2a$ in $\widetilde{\Omega}_{\bm{k}}$ corresponds to a shift of the elements within another cycle of length $2a$ in $\widetilde{\Omega}_{\bm{k}}$. For all of these cycles, there are $a$ possible values of the shift (because there are $a$ possible values of $l'$ in Eqs.~\eqref{eq:sigma_structure_1},~\eqref{eq:sigma_structure_2}). Additionally, we can make each shift with the elements within the cycles having either an ascending or descending order (as shown also in Eqs.~\eqref{eq:sigma_structure_1},~\eqref{eq:sigma_structure_2}). Since there are $k_a$ cycles of length $2a$ in $\widetilde{\Omega}_{\bm{k}}$, we have a total of $(2a)^{k_a}$ ways of choosing the shift and order of all of them. Moreover, as mentioned below Eq.~\eqref{eq:sigma_structure_2}, $q'$ determines the image of a given cycle. The number of possible $q'$ is equal to the number of cycles of length $2a$, namely $k_a$, which means that there are $k_a!$ ways of choosing the images of all these cycles. This implies that there are a total of $(2a)^{k_a}k_a!$ ways of choosing the image under $\widetilde{\sigma}$ of all the cycles of length $2a$ in $\widetilde{\Omega}_{\bm{k}}$. Since $a$ is arbitrary, and we are only considering cycles of length greater than two, we conclude that there are $\prod_{a=2}^N(2a)^{k_a}k_a!$ possible ways of choosing the image of $\widetilde{\sigma}$.
    \end{proof}

    Having found the number of possible $\widetilde{\sigma}$ and their structure, we can focus on $\bar{\sigma}$. This will also allow us to discuss possible modifications to $B$. 
    As shown in Eqs.~\eqref{eq:sigma_octahedral_permutation_3} and~\eqref{eq:sigma_octahedral_permutation_4}, once we choose 
    $\widetilde{\sigma}$ for $B=\{2k_1+1,\dots,2N\}$, we have complete freedom in
    the choice of $\bar{\sigma}\in S_{2k_1}$. This is because $\bar{\sigma}$
    plays no role in the computation of $\beta$. Indeed, we needed
    $\widetilde{\sigma}$ to ``break'' all the connected components of length greater
    than two in $\bm{\Gamma}(\Omega_{\bm{k}})$, while we ``hid'' the action of
    $\bar{\sigma}$ with our choice of $B$. We can now fix $\widetilde{\sigma}$, so
    that $\sigma=\bar{\sigma}\oplus\widetilde{\sigma}$ keeps breaking the connected
    components of interest, and see if we can ``reveal'' some of the structure 
    of $\bar{\sigma}$ while keeping $\beta\in H_{2N}$.

    If we extend $B$ to include a pair of points $(2p-1,2p)$, with $p\in\{1,\dots,k_1\}$, computing $\beta$ amounts to calculating $\bar{\Omega}_{\bm{k}}(\bar{\sigma}(2p-1))$, $\bar{\Omega}_{\bm{k}}(\bar{\sigma}(2p))$ and $\bar{\Omega}_{\bm{k}}(\bar{\sigma}^{-1}(2p-1))$, $\bar{\Omega}_{\bm{k}}(\bar{\sigma}^{-1}(2p))$. Since $\bar{\Omega}_{\bm{k}}$ consists only of transpositions, in order to 
    have $\beta\in H_{2N}$ we need the image of $(2p-1,2p)$ under $\bar{\sigma}$ to be of the form $(2q-1,2q)$ or $(2q,2q-1)$, with $q\in\{1,\dots,k_1\}$. Similarly, the image of $(2p-1,2p)$ under $\bar{\sigma}^{-1}$ must be of the 
    form $(2q'-1,2q')$ or $(2q',2q'-1)$, with $q'\in\{1,\dots,k_1\}$. This means we can extend $B$ to include $(2p-1,2p)$ as long as this pair of
    points leads to a connected component of length two in $\bm{\Gamma}(\bar{\sigma})$. Notice that extending $B$ to include $(2p-1,2p)$ may imply 
    $\sigma(B)\neq B,\, \sigma^{-1}(B)\neq B$. However, since the pair 
    $(2p-1,2p)$ leads to a connected component of length two in $\bm{\Gamma}(\bar{\sigma})$, we can guarantee that it will also lead to connected components of length two in 
    $\bm{\Gamma}(\bar{\Omega}_{\bm{k}}\bar{\sigma}),\,\bm{\Gamma}(\bar{\Omega}_{\bm{k}}\bar{\sigma}^{-1})$, and thus $\bm{\Gamma}(\Omega_{\bm{k}}\sigma),\,\bm{\Gamma}(\Omega_{\bm{k}}\sigma^{-1})$, making $\beta$ satisfy the conditions to be an element of $H_{2N}$. Consequently, given an specific $\bar{\sigma}$, any modification of $B$ will consist in appending pairs of the form $\{(2p-1,2p)\}$ that correspond to connected components of length two in $\bm{\Gamma}(\bar{\sigma})$. Notice that the choice of including a given pair is independent of the choice of including any other pair.

    The exact number of possible modifications of $B$, for given $\bar{\sigma}$ and $\widetilde{\sigma}$, will depend on the number of connected components of 
    length two in $\bm{\Gamma}(\bar{\sigma})$. Let $\eta(\bar{\sigma})$ be the
    coset-type of $\bar{\sigma}$, and let $\eta^{(1)}(\bar{\sigma})$ be the
    number of ones in $\eta(\bar{\sigma})$, which is equivalent to the number of
    connected components of length two in $\bm{\Gamma}(\bar{\sigma})$. Each of
    these connected components will correspond to a pair $(2p-1,2p)$ with $p\in\{1,\dots,k_1\}$. For each pair we can decide whether to include it in $B$ or
    not. This means that we can modify $B$ in $2^{\eta^{(1)}(\bar{\sigma})}$
    ways. This result allows us to state the following propositions:

    \begin{proposition}
    \label{prop:counting_result_same_coset}
        $\left\vert(\Omega_{\bm{k}}^{-1}\oplus \Omega_{\bm{k}}^{-1}) H_{2N}\cap S_\sigma^\star\right\vert=2^{\eta^{(1)}(\bar{\sigma})}$ if $\sigma=\bar{\sigma}\oplus\widetilde{\sigma}$, and vanishes otherwise. Here, $\bar{\sigma}\in S_{2k_1}$ acts over $\{1,\dots,2k_1\}$, and $\eta^{(1)}(\bar{\sigma})$ is the number of ones in the coset-type of $\bar{\sigma}$. $\widetilde{\sigma}\in S_{2N-2k_1}$ acts over $\{2k_1+1,\dots,2N\}$ and has the form indicated by Eqs.~\eqref{eq:sigma_structure_1} or~\eqref{eq:sigma_structure_2}. 
    \end{proposition}

    \begin{proof}
        We have seen that in order for $\tau_B\in S_{\sigma}^\star$, with $B\subseteq\{1,\dots,2N\}$, to also be an element of $(\Omega_{\bm{k}}^{-1}\oplus \Omega_{\bm{k}}^{-1}) H_{2N}$, Eqs.~\eqref{eq:sigma_octahedral_permutation_1} and~\eqref{eq:sigma_octahedral_permutation_2}, later transformed into Eqs.~\eqref{eq:sigma_octahedral_permutation_3} and~\eqref{eq:sigma_octahedral_permutation_4}, must hold. By setting $B=\{2k_1+1,\dots,2N\}$, we found that these equations are satisfied if $\sigma=\bar{\sigma}\oplus\widetilde{\sigma}$, with $\bar{\sigma}$ acting over $B^{\mathrm{c}}$ and $\widetilde{\sigma}$ over $B$. Moreover, $\widetilde{\sigma}$ must be defined by Eqs.~\eqref{eq:sigma_structure_1} or~\eqref{eq:sigma_structure_2}. Fixing $\widetilde{\sigma}$, we can modify $B$ in order to include any number of pairs of the form $(2p-1,2p)\in\{1,\dots,2N\}$ that correspond to connected components of length two in $\bm{\Gamma}(\bar{\sigma})$. The total number of this type of connected components is equal to $\eta^{(1)}(\bar{\sigma})$, the number of ones in the coset-type of $\bar{\sigma}$. For each pair we can choose whether to include it in $B$ or not, which leads to a total of $2^{\eta^{(1)}(\bar{\sigma})}$ possible modifications to $B$. Therefore, there are $2^{\eta^{(1)}(\bar{\sigma})}$ possible $\tau_B$ for a given $\sigma=\bar{\sigma}\oplus\widetilde{\sigma}$.  
    \end{proof}

    \begin{proposition}
        \label{prop:sum_over_sigma}
        \begin{equation}
            \sum_{\sigma\in S_{2N}}\left\vert(\Omega_{\bm{k}}^{-1}\oplus \Omega_{\bm{k}}^{-1}) H_{2N}\cap S_\sigma^\star\right\vert=\prod_{a=2}^N(2a)^{k_a}k_a!\sum_{\bar{\sigma}\in S_{2k_1}}2^{\eta^{(1)}(\bar{\sigma})}.
            \label{eq:sum_over_sigma}
        \end{equation}
    \end{proposition}
    \begin{proof}
        The sum over all $\sigma\in S_{2N}$ reduces to a sum over permutations of the form $\sigma=\bar{\sigma}\oplus\widetilde{\sigma}$ because the summand vanishes otherwise. For a given $\bar{\sigma}$ there are always $\prod_{a=2}^N(2a)^{k_a}k_a!$ possible $\widetilde{\sigma}$, so we only need to sum over $\bar{\sigma}$.
    \end{proof}

    Now that we have computed $\left\vert(\Omega_{\bm{k}}^{-1}\oplus \Omega_{\bm{l}}^{-1}) H_{2N}\cap S_\sigma^\star\right\vert$ as a function of $\bm{k}$, $\bm{l}$ and $\sigma$, we are almost ready to attack Eq.~\eqref{eq:counting_part}. Combining all the results we have obtained so far, we can see that
    \begin{align}
        &\sum_{\bm{k}^{(c)}, \bm{l}^{(c)}}\,\prod_{a=1}^N\frac{1}{k_a!l_a!(2a)^{k_a+l_a}}\sum_{\sigma\in S_{2N}}b_{2N}\left(\bm{j}\oplus\sigma(\bm{j}),\Omega_{\bm{k}}(\bm{j})\oplus\Omega_{\bm{l}}\circ\sigma(\bm{j})\right)\nonumber\\
        &=\sum_{\bm{k}^{(c)}, \bm{l}^{(c)}}\,\prod_{a=1}^N\frac{1}{k_a!l_a!(2a)^{k_a+l_a}}\sum_{\sigma\in S_{2N}}\left\vert(\Omega_{\bm{k}}^{-1}\oplus \Omega_{\bm{l}}^{-1}) H_{2N}\cap S_\sigma^\star\right\vert\nonumber\\
        &=\sum_{\bm{k}^{(c)}}\left(\prod_{a=1}^N\frac{1}{k_a!(2a)^{k_a}}\right)^2\sum_{\sigma\in S_{2N}}\left\vert(\Omega_{\bm{k}}^{-1}\oplus \Omega_{\bm{k}}^{-1}) H_{2N}\cap S_\sigma^\star\right\vert\nonumber\\
        &=\sum_{\bm{k}^{(c)}}\left(\prod_{a=1}^N\frac{1}{k_a!(2a)^{k_a}}\right)^2\prod_{a=2}^Nk_a!(2a)^{k_a}\sum_{\bar{\sigma}\in S_{2k_1}}2^{\eta^{(1)}(\bar{\sigma})}\nonumber\\
        &=\sum_{\bm{k}^{(c)}}\left(\prod_{a=2}^N\frac{1}{k_a!(2a)^{k_a}}\right)\left(\frac{1}{ 2^{k_1}k_1!}\right)^{2}\sum_{\bar{\sigma}\in S_{2k_1}}2^{\eta^{(1)}(\bar{\sigma})}.
        \label{eq:counting_part_semi_final_form}
    \end{align}
    In order to complete our proof, we need two additional results.
    \begin{proposition}
        \label{prop:sum_gen_function}
        Let $S_{2n}$ be the symmetric group of degree $2n$, and $\eta^{(1)}(\sigma)$ stand for the number of ones in the coset-type of $\sigma\in S_{2n}$. Moreover, let $\alpha$ be a real parameter satisfying $|\alpha|<1$. Then,
        \begin{equation}
            \sum_{\sigma\in S_{2n}}2^{\eta^{(1)}(\sigma)}=2^{2n}n!\left.\frac{d^n}{d\alpha^n}\left[e^{\alpha/2}(1-\alpha)^{-1/2}\right]\right|_{\alpha=0}.
            \label{eq:sum_gen_function}
        \end{equation}
    \end{proposition}
    \begin{proof}
        Let $H_{\eta}$ denote the set of permutations $\sigma \in S_{2n}$ that have coset-type $\eta$. Since every $\sigma$ has a unique coset-type, it will belong to only one $H_{\eta}$. This allows us to write the sum over $\sigma\in S_{2n}$ in Eq.~\eqref{eq:sum_gen_function} as
        \begin{equation}
            \sum_{\sigma\in S_{2n}}2^{\eta^{(1)}(\sigma)}=\sum_{\eta\vdash n}|H_\eta|\,2^{\eta^{(1)}},    
        \end{equation}
        where $\eta\vdash n$ indicates that the sum is taken over all the integer partitions of $n$ (recall that the coset-type of a permutation in $S_{2n}$ is an  integer partition of $n$), and $\eta^{(1)}$ is the number of ones in $\eta$. 
        
        Given a partition $\eta$ of $n$, we can construct a vector of non-negative integers $\bm{k}=(k_1,\dots,k_n)$ whose components satisfy $k_1+2k_2+\cdots+nk_n=n$. Indeed, we need only define $k_a$ as the multiplicity of $a$ in $\eta$, i.e., $k_a$ is the number of times $a$ appears in $\eta$. This implies that
        \begin{equation}
            \sum_{\sigma\in S_{2n}}2^{\eta^{(1)}(\sigma)}=\sum_{\eta\vdash n}|H_\eta|\,2^{\eta^{(1)}}=\sum_{\bm{k}^{(c)}}|H_{\bm{k}}|\,2^{k_1},   
        \end{equation}
        where we $\eta^{(1)}=k_1$ has been used, and we write $\bm{k}^{(c)}$ to indicate that the components of $\bm{k}$ are constrained. 
        
        According to Refs.~\cite{matsumoto2012general, macdonald1998symmetric}, we have
        \begin{equation}
            |H_\eta|=|H_{\bm{k}}|=(2^n n!)^2\prod_{a=1}^n\frac{1}{k_a!(2a)^{k_a}},    
        \end{equation}
        which leads to the expression
        \begin{equation}
            \sum_{\sigma\in S_{2n}}2^{\eta^{(1)}(\sigma)}=(2^nn!)^2\sum_{\bm{k}^{(c)}}\left(\prod_{a=1}^n\frac{1}{k_a!(2a)^{k_a}}\right)2^{k_1}=(2^nn!)Z_n\left(1,\frac{1}{2},\cdots,\frac{1}{2}\right),    
        \end{equation}
        where we recalled the definition of the cycle index of $S_{n}$ (Eq.~\eqref{eq:cycle_index} in the main text):
        \begin{equation}
            Z_n(y_1,\dots,y_n)=\sum_{\bm{k}^{(c)}}\,\prod_{a=1}^n\frac{1}{k_a!a^{k_a}}\prod_{a=1}^ny_a^{k_a}.
        \end{equation}

        Using the generating function of $Z_n(y_1,\dots,y_n)$
        \begin{equation}
            \exp\left[\sum_{l=1}^\infty y_l\frac{\alpha^l}{l}\right]=\sum_{n=0}^\infty Z_n(y_1,\dots,y_n)\alpha^n,   
        \end{equation}
        we can write
        \begin{equation}
            \exp\left[\alpha + \frac{1}{2}\sum_{l=2}^\infty\frac{\alpha^l}{l}\right]=e^{\alpha/2}\exp\left[\frac{1}{2}\sum_{l=1}^\infty\frac{\alpha^l}{l}\right]=\sum_{n=0}^\infty Z_n\left(1,\frac{1}{2},\dots,\frac{1}{2}\right)\alpha^n.    
        \end{equation}
        Assuming that $|\alpha|<1$, we can write $\sum_{l=1}^\infty \alpha^l/l=-\log(1-\alpha)$, and obtain
        \begin{equation}
            e^{\alpha/2}\exp\left[-\frac{1}{2}\log(1-\alpha)\right]=e^{\alpha/2}(1-\alpha)^{-1/2}=\sum_{n=0}^\infty Z_n\left(1,\frac{1}{2},\dots,\frac{1}{2}\right)\alpha^n.    
        \end{equation}
        Eq.~\eqref{eq:sum_gen_function} follows from repeatedly differentiating this expression with respect to $\alpha$, and then evaluating at $\alpha=0$.
    \end{proof}

    \begin{proposition}
        \label{prop:sum_reorganization}
        Let $x$ be a real parameter. Then,
        \begin{equation}
            \sum_{\bm{k}^{(c)}}\left(\prod_{a=1}^n\frac{1}{k_a!(2a)^{k_a}}\right)2^{k_1}x^{k_1}=\frac{1}{n!}\sum_{m=0}^n\binom{n}{m}x^m g_{n-m},
            \label{eq:sum_reorganization}
        \end{equation}
        where $g_m$ is defined by the relation
        \begin{equation}
            e^{-\beta/2}(1-\beta)^{-1/2}=\sum_{m=0}^{\infty}\frac{g_m}{m!}\beta^m,\quad g_m = (-1/2)^m\,{}_2F_0(1/2,-m;;2),
            \label{eq:useful_expansion_1}
        \end{equation}
        with $|\beta|<1$ and ${}_2F_0(a,b;;z)$ the hypergeometric function ${}_2F_0$.
    \end{proposition}
    \begin{proof}
        We can readily notice that
        \begin{equation}
            \sum_{\bm{k}^{(c)}}\left(\prod_{a=1}^n\frac{1}{k_a!(2a)^{k_a}}\right)2^{k_1}x^{k_1}=Z_n\left(x,\frac{1}{2},\dots,\frac{1}{2}\right),   
        \end{equation}
        and thus we can use the generating function of $Z_n$ to find an new expression for the left-hand side of Eq.~\eqref{eq:sum_reorganization}. We have
        \begin{align}
            \exp\left[x\beta + \frac{1}{2}\sum_{l=2}^\infty \frac{\beta^l}{l}\right] = e^{x\beta}e^{-\beta/2}\exp\left[\frac{1}{2}\sum_{l=1}^\infty \frac{\beta^l}{l}\right]=\sum_{n=0}^\infty Z_n\left(x,\frac{1}{2},\dots,\frac{1}{2}\right)\beta^n. 
        \end{align}
        Using the equality $\sum_{l=1}^\infty \beta^l/l=-\log(1-\beta)$ 
        for $|\beta|<1$, we can recast this expression as
        \begin{equation}
            e^{x\beta}e^{-\beta/2}(1-\beta)^{-1/2}=\sum_{n=0}^\infty Z_n\left(x,\frac{1}{2},\dots,\frac{1}{2}\right)\beta^n.    
        \end{equation}
        Consequently, 
        \begin{equation}
            \sum_{\bm{k}^{(c)}}\left(\prod_{a=1}^n\frac{1}{k_a!(2a)^{k_a}}\right)2^{k_1}x^{k_1}=Z_n\left(x,\frac{1}{2},\dots,\frac{1}{2}\right)=\frac{1}{n!}\left.\frac{d^n}{d\beta^n}\left[e^{x\beta}e^{-\beta/2}(1-\beta)^{-1/2}\right]\right|_{\beta=0}.    
        \end{equation}

        Now, let us expand $e^{x\beta}$ and $e^{-\beta/2}(1-\beta)^{-1/2}$ as
        \begin{equation}
            e^{x\beta}=\sum_{m=0}^\infty \frac{x^m}{n!}\beta^m\,,\quad e^{-\beta/2}(1-\beta)^{-1/2}=\sum_{m=0}^\infty\frac{g_m}{m!}\beta^m.
            \label{eq:bashypergometric}
        \end{equation}
        These expressions allow us to write
        \begin{equation}
            e^{x\beta}e^{-\beta/2}(1-\beta)^{-1/2}=\left(\sum_{m=0}^\infty \frac{x^m}{n!}\beta^m\right)\left(\sum_{p=0}^\infty\frac{g_p}{p!}\beta^p\right)=\sum_{n=0}^\infty\frac{1}{n!}\left(\sum_{m=0}^n\binom{n}{m}x^m g_{n-m}\right)\beta^n,    
        \end{equation}
        where we used Leibniz formula for the product of two series~\cite{comtet1974advanced,wilf2005generatingfunctionology} to obtain the last equality. Repeated differentiation with respect to $\beta$, and evaluation at $\beta=0$, leads to Eq.~\eqref{eq:sum_reorganization}.        
    \end{proof}

    Using Proposition~\ref{prop:sum_gen_function}, we can recast the last line in Eq.~\eqref{eq:counting_part_semi_final_form} as 
    \begin{equation}
        \sum_{\bm{k}^{(c)}}\left(\prod_{a=2}^N\frac{1}{k_a!(2a)^{k_a}}\right)\left(\frac{1}{ 2^{k_1}k_1!}\right)^{2}\sum_{\bar{\sigma}\in S_{2k_1}}2^{\eta^{(1)}(\bar{\sigma})}=\sum_{\bm{k}^{(c)}}\left(\prod_{a=1}^N\frac{1}{k_a!(2a)^{k_a}}\right)2^{k_1}\left.\frac{d^{k_1}}{d\alpha^{k_1}}\left[e^{\alpha/2}(1-\alpha)^{-1/2}\right]\right|_{\alpha=0}.
        \label{eq:counting_part_final_form}
    \end{equation}
    This expression leads to the following result:
    \begin{proposition}
        \label{prop:counting_part_final_result}
        \begin{equation}
            \sum_{\bm{k}^{(c)}}\left(\prod_{a=1}^N\frac{1}{k_a!(2a)^{k_a}}\right)2^{k_1}\left.\frac{d^{k_1}}{d\alpha^{k_1}}\left[e^{\alpha/2}(1-\alpha)^{-1/2}\right]\right|_{\alpha=0}=1.
            \label{eq:counting_part_final_result}
        \end{equation}
    \end{proposition}
    \begin{proof}
        Let us define the differential operator
        \begin{equation}
            \mathcal{D}_{\alpha}=\sum_{\bm{k}^{(c)}}\left(\prod_{a=1}^N\frac{1}{k_a!(2a)^{k_a}}\right)2^{k_1}\frac{d^{k_1}}{d\alpha^{k_1}},    
        \end{equation}
        which allows us to write
        \begin{equation}
            \sum_{\bm{k}^{(c)}}\left(\prod_{a=1}^N\frac{1}{k_a!(2a)^{k_a}}\right)2^{k_1}\left.\frac{d^{k_1}}{d\alpha^{k_1}}\left[e^{\alpha/2}(1-\alpha)^{-1/2}\right]\right|_{\alpha=0}=\left.\mathcal{D}_\alpha\left[e^{\alpha/2}(1-\alpha)^{-1/2}\right]\right|_{\alpha=0}.    
        \end{equation}
        Using Proposition~\ref{prop:sum_reorganization}, we can reorganize the sum that defines $\mathcal{D}_{\alpha}$ to obtain
        \begin{equation}
            \mathcal{D}_{\alpha}=\sum_{\bm{k}^{(c)}}\left(\prod_{a=1}^N\frac{1}{k_a!(2a)^{k_a}}\right)2^{k_1}\frac{d^{k_1}}{d\alpha^{k_1}}=\frac{1}{N!}\sum_{m=0}^N\binom{N}{m} g_{N-m}\frac{d^m}{d \alpha^m}.
        \end{equation}
        Expanding $e^{\alpha/2}(1-\alpha)^{-1/2}$ as
        \begin{equation}
            e^{\alpha/2}(1-\alpha)^{-1/2}=\sum_{p=0}^\infty\frac{a_p}{p!}\alpha^p\;;\quad a_p=\frac{(2p-1)!!}{2^p}{}_1F_1\left(-p;\frac{1}{2}-p;\frac{1}{2}\right),    
        \end{equation}
        where ${}_1F_1(a;b;c)$ denotes the confluent hypergeometric function of the first kind, we have
        \begin{equation}
            \left.\mathcal{D}_\alpha\left[e^{\alpha/2}(1-\alpha)^{-1/2}\right]\right|_{\alpha=0} = \frac{1}{N!}\sum_{m=0}^N\binom{N}{m} g_{N-m}\left.\frac{d^m}{d \alpha^m}\left(\sum_{p=0}^\infty\frac{a_p}{p!}\alpha^p\right)\right|_{\alpha=0}= \frac{1}{N!}\sum_{m=0}^N\binom{N}{m} a_m g_{N-m}.    
        \end{equation}
        Now, notice that
        \begin{equation}
            \left[e^{\alpha/2}(1-\alpha)^{-1/2}\right]\left[e^{-\alpha/2}(1-\alpha)^{-1/2}\right]=\left(\sum_{p=0}^\infty\frac{a_p}{p!}\alpha^p\right)\left(\sum_{m=0}^\infty \frac{g_m}{n!}\alpha^m\right)=\sum_{n=0}^\infty\frac{1}{n!}\left(\sum_{m=0}^n\binom{n}{m}a_m g_{n-m}\right)\alpha^n=(1-\alpha)^{-1},    
        \end{equation}
        which, combined with the fact that $(1-\alpha)^{-1}=\sum_{n=0}^\infty\alpha^n$ for $|\alpha|<1$, allows us to conclude that
        \begin{equation}
            \left.\mathcal{D}_\alpha\left[e^{\alpha/2}(1-\alpha)^{-1/2}\right]\right|_{\alpha=0} = \frac{1}{N!}\sum_{m=0}^N\binom{N}{m} a_m g_{N-m}=1    
        \end{equation}
        for all $N$.
    \end{proof}
    This concludes our proof of Eq.~\eqref{eq:counting_part}.

\end{document}